\documentclass[11pt]{article}
\usepackage{graphicx,amsmath,amsfonts,amscd,amssymb,bm,cite,epsfig,epsf,url,color}
\usepackage{fullpage} \usepackage[small,bf]{caption}
\newtheorem{theorem}{Theorem}[section]
\newtheorem{lemma}[theorem]{Lemma}
\newtheorem{corollary}[theorem]{Corollary}
\newtheorem{proposition}[theorem]{Proposition}

\newcommand\tr{{{\operatorname{trace}}}}
\renewcommand{\L}{{\mathcal L}}
\newcommand{\eps}{\varepsilon}
\renewcommand{\j}{{x}}
\renewcommand{\k}{{y}}
\newcommand{\U}{{U}}
\newcommand{\V}{{V}}

\newcommand{\R}{\mathbb{R}}

\newcommand{\<}{\langle}
\renewcommand{\>}{\rangle}

\newcommand{\goto}{\rightarrow}

\renewcommand{\P}{\operatorname{\mathbb{P}}}
\newcommand{\E}{\operatorname{\mathbb{E}}}

\newcommand{\PO}{\mathcal{P}_\Omega}
\newcommand{\PT}{\mathcal{P}_T}
\newcommand{\PTp}{\mathcal{P}_{T^\perp}}
\newcommand{\QO}{\mathcal{Q}_\Omega}
\newcommand{\QT}{\mathcal{Q}_T}
\newcommand{\rmu}{r_\mu}

\newcommand{\vct}[1]{{#1}}
\newcommand{\mtx}[1]{{#1}}

\newcommand{\rank}{\operatorname{rank}}

\newcommand{\trace}{\operatorname{trace}}

\newcommand{\OpId}{\mathcal{I}}

\numberwithin{equation}{section}

\def \endprf{\hfill {\vrule height6pt width6pt depth0pt}\medskip}
\newenvironment{proof}{\noindent {\bf Proof} }{\endprf\par}

\title{The Power of Convex Relaxation:\\
Near-Optimal Matrix Completion}

\author{Emmanuel J. Cand\`es$^{\dagger}$ and Terence Tao$^{\sharp}$\\
  \vspace{-.1cm}\\
  $\dagger$ Applied and Computational Mathematics,
  Caltech, Pasadena, CA 91125\\
  \vspace{-.3cm}\\
  $\sharp$ Department of Mathematics, University of California,
  Los Angeles, CA 90095}

\date{\today}

\begin{document}

\maketitle

\vspace{-0.3in}

\begin{abstract}
  This paper is concerned with the problem of recovering an unknown
  matrix from a small fraction of its entries. This is known as the
  {\em matrix completion} problem, and comes up in a great number of
  applications, including the famous {\em Netflix Prize} and other
  similar questions in collaborative filtering.  In general, accurate
  recovery of a matrix from a small number of entries is impossible; but the
  knowledge that the unknown matrix has low rank radically changes
  this premise, making the search for solutions meaningful.

  This paper presents optimality results quantifying the minimum
  number of entries needed to recover a matrix of rank $r$ exactly by
  any method whatsoever (the information theoretic limit). More
  importantly, the paper shows that, under certain incoherence assumptions on the 
  singular vectors of the matrix, recovery is possible by solving a
  convenient convex program as soon as the number of entries is on the
  order of the information theoretic limit (up to logarithmic factors).  This convex program
  simply finds, among all matrices consistent with the observed
  entries, that with minimum nuclear norm. As an example, we show that
  on the order of $nr \log(n)$ samples are needed to recover a random
  $n \times n$ matrix of rank $r$ by any method, and to be sure,
  nuclear norm minimization succeeds as soon as the number of entries
  is of the form $nr \text{polylog}(n)$.
\end{abstract}


\maketitle

\section{Introduction}
\label{sec:intro}

\subsection{Motivation}
\label{sec:motivation}

Imagine we have an $n_1 \times n_2$ array of real\footnote{Much of the
  discussion below, as well as our main results, apply also to the
  case of complex matrix completion, with some minor adjustments in
  the absolute constants; but for simplicity we restrict attention to
  the real case.} numbers and that we are interested in knowing the
value of each of the $n_1 n_2$ entries in this array. Suppose,
however, that we only get to see a small number of the entries so that
most of the elements about which we wish information are simply
missing. Is it possible from the available entries to guess the many
entries that we have not seen?  This problem is now known as the {\em
  matrix completion} problem \cite{CR}, and comes up in a great number
of applications, including the famous {\em Netflix Prize} and other
similar questions in collaborative filtering \cite{Goldberg92}. In a
nutshell, collaborative filtering is the task of making automatic
predictions about the interests of a user by collecting taste
information from many users. Netflix is a commercial company
implementing collaborative filtering, and seeks to predict users'
movie preferences from just a few ratings per user. There are many
other such recommendation systems proposed by Amazon, Barnes and
Noble, and Apple Inc.~to name just a few. In each instance, we have a
partial list about a user's preferences for a few rated items, and
would like to predict his/her preferences for all items from this and
other information gleaned from many other users.

In mathematical terms, the problem may be posed as follows: we have a
data matrix $M \in \R^{n_1 \times n_2}$ which we would like to know as
precisely as possible. Unfortunately, the only information available
about $M$ is a sampled set of entries $M_{ij}$, $(i,j) \in \Omega$,
where $\Omega$ is a subset of the complete set of entries $[n_1]
\times [n_2]$. (Here and in the sequel, $[n]$ denotes the list
$\{1,\ldots,n\}$.) Clearly, this problem is ill-posed for there is no
way to guess the missing entries without making any assumption about
the matrix $M$.

An increasingly common assumption in the field is to suppose that the
unknown matrix $M$ has low rank or has approximately low rank. In a
recommendation system, this makes sense because often times, only a
few factors contribute to an individual's taste. In \cite{CR}, the
authors showed that this premise radically changes the problem, making
the search for solutions meaningful. Before reviewing these results,
we would like to emphasize that the problem of recovering a low-rank
matrix from a sample of its entries, and by extension from fewer
linear functionals about the matrix, comes up in many application
areas other than collaborative filtering. For instance, the completion
problem also arises in computer vision. There, many pixels may be
missing in digital images because of occlusion or tracking failures in
a video sequence. Recovering a scene and inferring camera motion from
a sequence of images is a matrix completion problem known as the
structure-from-motion problem \cite{Tomasi,ChenSuter}. Other examples
include system identification in control \cite{Mesbahi97}, multi-class
learning in data analysis \cite{Abernethy06, Argyriou07, Amit07},
global positioning---e.g. of sensors in a network---from partial
distance information \cite{TohYe,Singer1,Singer2}, remote sensing
applications in signal processing where we would like to infer a full
covariance matrix from partially observed correlations \cite{Brian},
and many statistical problems involving succinct factor models.

\subsection{Minimal sampling}
\label{sec:minimal}

This paper is concerned with the theoretical underpinnings of matrix
completion and more specifically in quantifying the minimum number of
entries needed to recover a matrix of rank $r$ exactly. This number
generally depends on the matrix we wish to recover. For simplicity,
assume that the unknown rank-$r$ matrix $M$ is $n \times n$. Then it
is not hard to see that matrix completion is impossible unless the
number of samples $m$ is at least $2nr - r^2$, as a matrix of rank $r$
depends on this many degrees of freedom. The singular value
decomposition (SVD)
\begin{equation}
  \label{eq:svd}
  \mtx{M} = \sum_{k \in [r]} \sigma_k \vct{u}_k \vct{v}_k^*,
\end{equation}
where $\sigma_1,\ldots,\sigma_r \geq 0$ are the singular values, and
the singular vectors $u_1,\ldots,u_r \in \R^{n_1} = \R^n$ and
$v_1,\ldots,v_r \in \R^{n_2} = \R^n$ are two sets of orthonormal
vectors, is useful to reveal these degrees of freedom. Informally, the
singular values $\sigma_1 \ge \ldots \ge \sigma_r$ depend on $r$
degrees of freedom, the left singular vectors $\vct{u}_k$ on $(n-1) +
(n-2) + \ldots + (n-r) = nr - r(r+1)/2$ degrees of freedom, and
similarly for the right singular vectors $\vct{v}_k$. If $m < 2nr -
r^2$, no matter which entries are available, there can be an infinite
number of matrices of rank at most $r$ with exactly the same entries,
and so exact matrix completion is impossible.  In fact, if the
observed locations are sampled at random, we will see later that the
minimum number of samples is better thought of as being on the order
of $nr \log n$ rather than $nr$ because of a coupon collector's
effect.

In this paper, we are interested in identifying large classes of
matrices which can provably be recovered by a tractable algorithm from
a number of samples approaching the above limit, i.e.~from about $nr
\log n$ samples. Before continuing, it is convenient to introduce some
notations that will be used throughout: let $\PO: \R^{n \times n} \to
\R^{n \times n}$ be the orthogonal projection onto the subspace of
matrices which vanish outside of $\Omega$ ($(i,j) \in \Omega$ if and
only if $M_{ij}$ is observed); that is, $\mtx{Y} = \PO(\mtx{X})$ is
defined as
\[
Y_{ij} = \begin{cases} X_{ij}, & (i,j) \in \Omega,\\
  0, & \text{otherwise}, \end{cases}
\]
so that the information about $M$ is given by $\PO(\mtx{M})$. The
matrix $M$ can be, in principle, recovered from $\PO(\mtx{M})$ if
it is the unique matrix of rank less or equal to $r$ consistent with
the data. In other words, if $M$ is the unique solution to
\begin{equation}
\label{eq:rank}
  \begin{array}{ll}
    \text{minimize}   & \quad \rank(X)\\
    \text{subject to} & \quad \PO(\mtx{X}) = \PO(\mtx{M}). 
 \end{array}
\end{equation}
Knowing when this happens is a delicate question which shall be
addressed later. For the moment, note that attempting recovery via
\eqref{eq:rank} is not practical as rank minimization is in general an NP-hard problem for
which there are no known algorithms capable of solving problems in practical time once, say, $n \geq 10$.

In \cite{CR}, it was proved 1) that matrix completion is not as
ill-posed as previously thought and 2) that exact matrix completion is
possible by convex programming. The authors of \cite{CR} proposed
recovering the unknown matrix by solving the nuclear norm minimization
problem
\begin{equation}
\label{eq:cvx2}
  \begin{array}{ll}
    \text{minimize}   & \quad \|\mtx{X}\|_*\\
    \text{subject to} & \quad \PO(\mtx{X}) = \PO(\mtx{M}), 
 \end{array}
\end{equation}
where the \emph{nuclear norm} $\| X \|_*$ of a matrix $X$ is defined as the sum of its singular values,
\begin{equation}
  \label{eq:nuclear}
  \|X\|_* := \sum_{i} \sigma_i(X).
\end{equation}
(The problem \eqref{eq:cvx2} is a semidefinite program
\cite{fazelRank}.) They proved that if $\Omega$ is sampled uniformly
at random among all subset of cardinality $m$ and $M$ obeys a low
coherence condition which we will review later, then with large
probability, the unique solution to \eqref{eq:cvx2} is exactly $M$,
provided that the number of samples obeys
\begin{equation}
  \label{eq:CR}
  m \ge C \, n^{6/5} r \, \log n  
\end{equation}
(to be completely exact, there is a restriction on the range of values
that $r$ can take on). 

In \eqref{eq:CR}, the number of samples per degree of freedom is
not logarithmic or polylogarithmic in the dimension, and one would
like to know whether better results approaching the $nr \log n$ limit
are possible. This paper provides a positive answer. In details, this
work develops many useful matrix models for which nuclear norm
minimization is guaranteed to succeed as soon as the number of entries
is of the form $nr \text{polylog}(n)$.

\subsection{Main results}
\label{sec:main}

A contribution of this paper is to develop simple hypotheses about the
matrix $M$ which makes it recoverable by semidefinite programming from
nearly minimally sampled entries. To state our assumptions, we recall
the SVD of $M$ \eqref{eq:svd} and denote by $P_U$ (resp.~$P_V$) the
orthogonal projections onto the column (resp.~row) space of $M$;
i.e.~the span of the left (resp.~right) singular vectors. Note that
\begin{equation}\label{puv}
P_U = \sum_{i \in [r]} u_i u_i^*; \quad P_V = \sum_{i \in [r]} v_i v_i^*.
\end{equation}
Next, define the matrix $E$ as
\begin{equation}
  \label{eq:E}
  E := \sum_{i \in [r]} u_i v_i^*.
\end{equation}
We observe that $E$ interacts well with $P_U$ and $P_V$, in particular obeying the identities
$$ P_U E = E = E P_V; \quad E^* E = P_V; \quad EE^* = P_U.$$
One can view $E$ as a sort of matrix-valued ``sign pattern'' for $M$ (compare \eqref{eq:E} with \eqref{eq:svd}), and is also closely related to the subgradient $\partial \| M \|_*$ of the nuclear norm at $M$ (see \eqref{eqn:subdiffNorm}). 

It is clear that some assumptions on the singular vectors $u_i, v_i$ (or on the spaces $U, V$) is needed in order to have a hope of efficient matrix completion.  For instance, if $u_1$ and $v_1$ are Kronecker delta functions at positions $i, j$ respectively, then the singular value $\sigma_1$ can only be recovered if one actually samples the $(i,j)$ coordinate, which is only likely if one is sampling a significant fraction of the entire matrix.  Thus we need the vectors $u_i,v_i$ to be ``spread out'' or ``incoherent'' in some sense.  In our arguments, it will be convenient to phrase this incoherence assumptions using the projection matrices $P_U, P_V$ and the sign pattern matrix $E$.  More precisely, our assumptions are as follows.

\begin{description}
\item[{A1}] There exists $\mu_1 > 0$ such that for all pairs $(a,a') \in
  [n_1] \times [n_1]$ and $(b,b') \in [n_2] \times [n_2]$,
\begin{subequations}
\label{eq:grp}
\begin{align}
\label{eq:block}
\Bigl|\<e_a, P_U e_{a'}\> - \frac{r}{n_1} 1_{a = a'}\Bigr| & \le \mu_1 \frac{\sqrt{r}}{n_1},\\
\label{eq:block2}
\Bigl|\<e_b, P_V e_{b'}\> - \frac{r}{n_2} 1_{b = b'}\Bigr| & \le
\mu_1 \frac{\sqrt{r}}{n_2}.
\end{align}
\end{subequations}
\item[{A2}] There exists $\mu_2 > 0$ such that for all $(a,b) \in [n_1]
  \times [n_2]$,
\begin{equation}\label{eab}
 |E_{ab}| \le \mu_2 \frac{\sqrt{r}}{\sqrt{n_1 n_2}}. 
\end{equation}
\end{description}
We will say that the matrix $M$ obey the {\em strong incoherence
  property} with parameter $\mu$ if one can take $\mu_1$ and $\mu_2$
both less than equal to $\mu$.  (This property is related to, but slightly different from, the \emph{incoherence property}, which will be discussed in Section \ref{inco}.)

\emph{Remark.}  Our assumptions only involve the singular vectors $u_1,\ldots,u_r,v_1,\ldots,v_r$ of $M$; the singular \emph{values} $\sigma_1,\ldots,\sigma_r$ are completely unconstrained.  This lack of dependence on the singular values is a consequence of the geometry of the nuclear norm (and in particular, the fact that the subgradient $\partial \| X \|_*$ of this norm is independent of the singular values, see \eqref{eqn:subdiffNorm}).

It is not hard to see that $\mu$ must be greater than 1.  For instance, \eqref{eab} implies
\[
r = \sum_{(a,b) \in [n_1] \times [n_2]} |E_{ab}|^2 \le \mu_2^2 \, r
\]
which forces $\mu_2 \geq 1$.  The Frobenius norm identities 
\[
r = \|P_U\|_F^2 = \sum_{a,a' \in [n_1]} |\< e_a, P_U e_{a'} \>|^2
\] 
and \eqref{eq:block}, \eqref{eq:block2} also place a similar lower
bound on $\mu_1$.

We will show that 1) matrices obeying the strong incoherence property
with a small value of the parameter $\mu$ can be recovered from fewer
entries and that 2) many matrices of interest obey the strong
incoherence property with a small $\mu$.  We will shortly develop
three models, the {\em uniformly bounded orthogonal model}, the {\em
  low-rank low-coherence model}, and the {\em random orthogonal model}
which all illustrate the point that if the singular vectors of $M$ are
``spread out'' in the sense that their amplitudes all have about the
same size, then the parameter $\mu$ is low. In some sense, ``most''
low-rank matrices obey the strong incoherence property with $\mu =
O(\sqrt{\log n})$, where $n = \max(n_1,n_2)$. Here, $O(\cdot)$ is the
standard asymptotic notation, which is reviewed in Section
\ref{sec:notation}.

Our first matrix completion result is as follows.

\begin{theorem}[Matrix completion I]\label{teo:main1}  
  Let $M \in \R^{n_1 \times n_2}$ be a fixed matrix of rank $r = O(1)$
  obeying the strong incoherence property with parameter $\mu$.  Write
  $n := \max(n_1,n_2)$.  Suppose we observe $m$ entries of $\mtx{M}$
  with locations sampled uniformly at random.  Then there is a
  positive numerical constant $C$ such that if
  \begin{equation}
    \label{eq:main1}
    m \ge C\, \mu^4 n  (\log n)^2,
  \end{equation}
  then $M$ is the unique solution to \eqref{eq:cvx2} with probability
  at least $1 - n^{-3}$.  In other words: with high probability,
  nuclear-norm minimization recovers all the entries of $\mtx{M}$ with
  no error.
\end{theorem}
This result is noteworthy for two reasons. The first is that the
matrix model is deterministic and only needs the strong incoherence
assumption. The second is more substantial.  Consider the class of
bounded rank matrices obeying $\mu = O(1)$. We shall see that no
method whatsoever can recover those matrices unless the number of
entries obeys $m \ge c_0 \, n \log n$ for some positive numerical
constant $c_0$; this is the information theoretic limit. Thus Theorem
\ref{teo:main1} asserts that exact recovery by nuclear-norm
minimization occurs nearly as soon as it is information theoretically
possible. Indeed, if the number of samples is slightly larger, by a
logarithmic factor, than the information theoretic limit, then
\eqref{eq:cvx2} fills in the missing entries with no error.

We stated Theorem \ref{teo:main1} for bounded ranks, but our proof
gives a result for all values of $r$. Indeed, the argument will establish that the
recovery is exact with high probability provided that
\begin{equation}
m \ge C\, \mu^4 n r^2 (\log n)^2.
\label{eq:quadratic}
\end{equation}
When $r = O(1)$, this is Theorem \ref{teo:main1}. We will prove a stronger and
near-optimal result below (Theorem \ref{teo:main2}) in which we 
replace the quadratic dependence on $r$ with linear dependence. The reason why we state
Theorem \ref{teo:main1} first is that its proof is somewhat simpler than
that of Theorem \ref{teo:main2}, and we hope that it will provide the
reader with a useful lead-in to the claims and proof of our main
result. 
\begin{theorem}[Matrix completion II]
\label{teo:main2}
Under the same hypotheses as in Theorem \ref{teo:main1}, there is a
numerical constant $C$ such that if
  \begin{equation}
    \label{eq:main2}
    m \ge C\, \mu^2 nr \log^6 n,
  \end{equation}
  $M$ is the unique solution to \eqref{eq:cvx2} with probability at
  least $1 - n^{-3}$.
\end{theorem}
This result is general and nonasymptotic. 

The proof of Theorems \ref{teo:main1}, \ref{teo:main2} will occupy the bulk of the paper, starting at Section \ref{sec:strategy}.

\subsection{A surprise}
\label{sec:surprise}

We find it unexpected that nuclear norm-minimization works so well,
for reasons we now pause to discuss. For simplicity, consider matrices
with a strong incoherence parameter $\mu$ polylogarithmic in the
dimension. We know that for the rank minimization program
\eqref{eq:rank} to succeed, or equivalently for the problem to be well
posed, the number of samples must exceed a constant times $nr \log
n$. However, Theorem \ref{teo:main2} proves that the convex relaxation
is rigorously exact nearly as soon as our problem has a unique
low-rank solution.  The surprise here is that admittedly, there is
{\em a priori} no good reason to suspect that convex relaxation might
work so well. There is {\em a priori} no good reason to suspect that
the gap between what combinatorial and convex optimization can do is
this small.  In this sense, we find these findings a little
unexpected.

The reader will note an analogy with the recent literature on
compressed sensing, which shows that under some conditions, the
sparsest solution to an underdetermined system of linear equations is
that with minimum $\ell_1$ norm.

\subsection{Model matrices}
\label{sec:model}

We now discuss model matrices which obey the conditions \eqref{eq:grp} and
\eqref{eab} for small values of the strong incoherence parameter
$\mu$.  For simplicity we restrict attention to the square matrix case $n_1=n_2=n$.

\subsubsection{Uniformly bounded model}
\label{sec:ubm}

In this section we shall show, roughly speaking, that almost all $n \times n$ matrices $M$ with singular vectors obeying the
size property 
\begin{equation}
  \label{eq:bdd}
  \|u_k\|_{\ell_\infty}, \|v_k\|_{\ell_\infty} \le \sqrt{\mu_B/n}, 
\end{equation}
with $\mu_B = O(1)$ also satisfy the assumptions {\bf A1} and {\bf A2}
with $\mu_1, \mu_2 = O(\sqrt{\log n})$. This justifies our earlier
claim that when the singular vectors are spread out, then the strong
incoherence property holds for a small value of $\mu$. 

We define a random model obeying \eqref{eq:bdd} as follows: take two
arbitrary families of $n$ orthonormal vectors $[u_1, \ldots, u_n]$ and
$[v_1, \ldots, v_n]$ obeying \eqref{eq:bdd}.  We allow the $u_i$ and $v_i$ to be deterministic; for instance one could have $u_i=v_i$ for all $i \in [n]$.

\begin{enumerate}
\item Select $r$ left singular vectors
  $u_{\alpha(1)},\ldots,u_{\alpha(r)}$ at random with replacement from
  the first family, and $r$ right singular vectors
  $v_{\beta(1)},\ldots,v_{\beta(r)}$ from the second family, also at
  random. We do \emph{not} require that the $\beta$ are chosen
  independently from the $\alpha$; for instance one could have
  $\beta(k)=\alpha(k)$ for all $k \in [r]$.
  
\item Set $M := \sum_{k \in [r]} \epsilon_k \sigma_k u_{\alpha(k)}
  v_{\beta(k)}^*$, where the signs $\epsilon_1,\ldots,\epsilon_r \in
  \{-1,+1\}$ are chosen independently at random (with probability
  $1/2$ of each choice of sign), and $\sigma_1,\ldots,\sigma_r > 0$
  are arbitrary distinct positive numbers (which are allowed to depend
  on the previous random choices).
\end{enumerate}

We emphasize that the only assumptions about the families $[u_1,
\ldots, u_n]$ and $[v_1, \ldots, v_n]$ is that they have small
components. For example, they may be the same. Also note that this
model allows for any kind of dependence between the left and right
singular selected vectors. For instance, we may select the same
columns as to obtain a symmetric matrix as in the case where the two
families are the same. Thus, one can think of our model as producing a
generic matrix with uniformly bounded singular vectors.

We now show that $P_U$, $P_V$ and $E$ obey \eqref{eq:grp} and
\eqref{eab}, with $\mu_1, \mu_2 = O(\mu_B \sqrt{\log n})$, with large
probability. For \eqref{eab}, observe that
\[
E = \sum_{k \in [r]} \epsilon_k u_{\alpha(k)} v_{\beta(k)}^*,
\]
and $\{\epsilon_k\}$ is a sequence of i.i.d.~$\pm 1$
symmetric random variables. Then Hoeffding's inequality shows that
$\mu_2 = O(\mu_B \sqrt{\log n})$; see \cite{CR} for details.

For \eqref{eq:grp}, we will use a beautiful concentration-of-measure result of McDiarmid.
\begin{theorem}\label{teo:mcdiarmid}  \cite{McDiarmid}
  Let $\{a_1, \ldots, a_n\}$ be a sequence of scalars obeying $|a_i|
  \le \alpha$. Choose a random set $S$ of size $s$ without replacement
  from $\{1, \ldots, n\}$ and let $Y = \sum_{i \in S} a_i$. Then for
  each $t \ge 0$, 
  \begin{equation}
    \label{eq:mcdiarmid}
    \P(|Y-\E Y| \ge t) \le 2 e^{-\frac{t^2}{2 s \alpha^2}}.  
  \end{equation}
\end{theorem}

From \eqref{puv} we have
\[
P_U = \sum_{k \in S} u_k u_k^*, 
\]
where $S := \{ \alpha(1),\ldots,\alpha(r) \}$. For any fixed $a, a' \in [n]$, set
\[
Y := \<P_U e_a, P_U e_{a'}\> = \sum_{k \in S} \<e_a, u_k\> \<u_k, e_{a'}\>
\]
and note that $\E Y = \frac{r}{n} 1_{a = a'}$. Since $|\<e_a, u_k\>
\<u_k, e_{a'}\>| \le \mu_B/n$, we apply \eqref{eq:mcdiarmid} and
obtain
\begin{equation*}
  \P\Bigl(\bigl|\<P_U e_a, P_U e_{a'}\> - 1_{\{a =
    a'\}}r/n\bigr| \ge 
  \lambda \,\mu_B \frac{\sqrt{r}}{n} \Bigr) \leq 2 e^{-\lambda^2/2}.
\end{equation*}
Taking $\lambda$ proportional to $\sqrt{\log n}$ and applying the union bound for $a,a' \in [n]$ proves \eqref{eq:grp} with probability at least $1-n^{-3}$ (say) with $\mu_1 = O(\mu_B \sqrt{\log n})$. 

Combining this computation with Theorems \ref{teo:main1}, \ref{teo:main2}, we have established the following corollary:
\begin{corollary}[Matrix completion, uniformly bounded model]
\label{teo:ubm}
Let $M$ be a matrix sampled from a uniformly bounded model. Under
the hypotheses of Theorem \ref{teo:main1}, if 
\begin{equation*}
    m \ge C\, \mu_B^2 nr \log^7 n,
\end{equation*}
$M$ is the unique solution to \eqref{eq:cvx2} with probability at
least $1 - n^{-3}$.  As we shall see below, when $r = O(1)$, it
suffices to have
\begin{equation*}
    m \ge C\, \mu_B^2 n \log^2 n. 
\end{equation*}
\end{corollary}

\emph{Remark.} For large values of the rank, the assumption that the $\ell_\infty$
norm of the singular vectors is $O(1/\sqrt{n})$ is not sufficient to
conclude that \eqref{eq:grp} holds with $\mu_1 = O(\sqrt{\log
  n})$. Thus, the extra randomization step (in which we select the $r$ singular vectors from a list of $n$ possible vectors) is in some sense necessary.
As an example, take $[u_1, \ldots, u_r]$ to be the first $r$ columns
of the Hadamard transform where each row corresponds to a
frequency. Then $\|u_k\|_{\ell_\infty} \le 1/\sqrt{n}$ but if $r \le
n/2$, the first two rows of $[u_1, \ldots, u_r]$ are identical. Hence
\[
\<P_U e_1, P_U e_2\> = r/n. 
\]
Obviously, this does not scale like $\sqrt{r}/n$. Similarly, the sign
flip (step 2) is also necessary as otherwise, we could have $E = P_U$
as in the case where $[u_1, \ldots, u_n] = [v_1, \ldots, v_n]$ and the
same columns are selected. Here, 
\[
\max_{a} E_{aa} = \max_{a} \|P_U e_a\|^2 \ge \frac{1}{n} \sum_{a}
\|P_U e_a\|^2 = \frac{r}{n},
\]
which does not scale like $\sqrt{r}/n$ either.

\subsubsection{Low-rank low-coherence model}
\label{sec:lrlc}

When the rank is small, the assumption that the
singular vectors are spread is sufficient to show that the parameter
$\mu$ is small. To see this, suppose that the singular vectors obey
\eqref{eq:bdd}. Then 
\begin{equation}
  \label{eq:naive}
  \Bigl|\<P_U e_a, P_U e_{a'}\> - 1_{\{a = a'\}} \frac{r}{n}\Bigr| \le \max_{a \in [n]} \|P_U e_a\|^2 \le  \frac{\mu_B r}{n}.
\end{equation}
The first inequality follows from the Cauchy-Schwarz inequality
$$|\<P_U e_a, P_U e_{a'}\>| \le \|P_U
e_a\| \|P_U e_{a'}\|$$ 
for $a \neq a'$ and from the Frobenius norm bound
$$\max_{a \in [n]} \|P_U e_a\|^2 \ge \frac{1}{n} \|P_U\|_F^2 = \frac{r}{n}.$$
This gives $\mu_1 \le \mu_B \sqrt{r}$. Also, by another application of Cauchy-Schwarz we have
\begin{equation}
  \label{eq:Enaive}
  |E_{ab}| \le \max_{a \in [n]} \|P_U e_a\|  \max_{b \in [n]} \|P_V e_b\| \le \frac{\mu_B r}{n}
\end{equation}
so that we also have $\mu_2 \le \mu_B \sqrt{r}$. In short, $\mu \le
\mu_B \sqrt{r}$. 

Our low-rank low-coherence model assumes that $r = O(1)$ and that the
singular vectors obey \eqref{eq:bdd}. When $\mu_B = O(1)$, this model
obeys the strong incoherence property with $\mu = O(1)$. In this case,
Theorem \ref{teo:main1} specializes as follows:

\begin{corollary}[Matrix completion, low-rank low-coherence model]
\label{teo:lrlc}
Let $M$ be a matrix of bounded rank $(r = O(1))$ whose singular
vectors obey \eqref{eq:bdd}. Under the hypotheses of Theorem
\ref{teo:main1}, if
\begin{equation*}
    m \ge C\, \mu_B^2 n \log^2 n,
\end{equation*}
then $M$ is the unique solution to \eqref{eq:cvx2} with probability at
least\ $1 - n^{-3}$. 
\end{corollary}

\subsubsection{Random orthogonal model}
\label{sec:rom}

Our last model is borrowed from \cite{CR} and assumes that the column
matrices $[u_1, \ldots, u_r]$ and $[v_1, \ldots, v_r]$ are independent
random orthogonal matrices, with no assumptions whatsoever on the singular values $\sigma_1,\ldots,\sigma_r$. Note that this is a special case of the
uniformly bounded model since this is equivalent to selecting two $n
\times n$ random orthonormal bases, and then selecting the singular
vectors as in Section \ref{sec:ubm}. Since we know that the maximum
entry of an $n \times n$ random orthogonal matrix is bounded by a
constant times $\sqrt{\frac{\log n}{n}}$ with large probability, then
Section \ref{sec:ubm} shows that this model obeys the strong
incoherence property with $\mu = O(\log n)$. Theorems \ref{teo:main1}, \ref{teo:main2} then give

\begin{corollary}[Matrix completion, random orthogonal model]
\label{teo:rom}
Let $M$ be a matrix sampled from the random orthogonal model. Under
the hypotheses of Theorem \ref{teo:main1}, if
\begin{equation*}
  m \ge C\, nr\, \log^8 n,
\end{equation*}
then $M$ is the unique solution to \eqref{eq:cvx2} with probability at
least $1 - n^{-3}$.  The exponent $8$ can be lowered to $7$ when $r
\ge \log n$ and to $6$ when $r = O(1)$.
\end{corollary}

As mentioned earlier, we have a lower bound $m \geq 2nr - r^2$ for
matrix completion, which can be improved to $m \geq C nr \log n$ under
reasonable hypotheses on the matrix $M$.  Thus, the hypothesis on $m$
in Corollary \ref{teo:rom} cannot be substantially improved.  However,
it is likely that by specializing the proofs of our general results
(Theorems \ref{teo:main1} and \ref{teo:main2}) to this special case,
one may be able to improve the power of the logarithm here, though it
seems that a substantial effort would be needed to reach the optimal
level of $nr \log n$ even in the bounded rank case.


Speaking of logarithmic improvements, we have shown that $\mu = O(\log
n)$, which is sharp since for $r = 1$, one cannot hope for better
estimates. For $r$ much larger than $\log n$, however, one can improve
this to $\mu = O(\sqrt{\log n})$. As far as $\mu_1$ is concerned, this
is essentially a consequence of the Johnson-Lindenstrauss lemma. For
$a \neq a'$, write
\[
\<P_U e_a, P_U e_{a'}\> = \frac{1}{4} \left(\|P_U e_a + P_U
  e_{a'}\|^2 - \|P_U e_a - P_U e_{a'}\|^2\right). 
\]
We claim that for each $a \neq a'$, 
\begin{equation}
\label{eq:JL}
\Bigl| \|P_U(e_a \pm e_{a'})\|^2 - \frac{2r}{n} \Bigr| \le C
\frac{\sqrt{r\log n}}{n}
\end{equation}
with probability at least $1 - n^{-5}$, say. This inequality is indeed
well known. Observe that $\|P_U x\|$ has the same distribution than
the Euclidean norm of the first $r$ components of a vector uniformly
distributed on the $n-1$ dimensional sphere of radius $\|x\|$. Then we
have \cite{Barvinok}:
\[
\P\Bigl(\sqrt{\frac{r}{n}}(1-\eps) \|x\| \le \|P_U x\| \le
\sqrt{\frac{r}{n}}(1-\eps)^{-1}\|x\|\Bigr) \le 2e^{-\epsilon^2 r/4} +
2 e^{-\epsilon^2 n/4}.
\]
Choosing $x = e_a \pm e_{a'}$, $\epsilon = C_0 \sqrt{\frac{\log
    n}{r}}$, and applying the union bound proves the claim as long as
long as $r$ is sufficiently larger than $\log n$. Finally, since a
bound on the diagonal term $\|P_U e_a\|^2 - r/n$ in \eqref{eq:grp}
follows from the same inequality by simply choosing $x = e_a$, we have
$\mu_1 = O(\sqrt{\log n})$. Similar arguments for $\mu_2$ exist but we
forgo the details.

\subsection{Comparison with other works}

\subsubsection{Nuclear norm minimization}\label{inco}

The mathematical study of matrix completion began with \cite{CR}, which made slightly different incoherence assumptions than in this paper.  Namely, let us say that the matrix $M$ obeys the \emph{incoherence property} with a parameter $\mu_0 > 0$ if
\begin{equation}
  \label{eq:coherence} \|P_U e_a\|^2 \le  \frac{\mu_0 r}{n_1}, 
  \quad \|P_V e_b\|^2 \le \frac{\mu_0 r}{n_2}
\end{equation}
for all $a \in [n_1]$, $b \in [n_2]$. Again, this implies $\mu_0 \ge 1$.  

In \cite{CR} it was shown that if a fixed matrix $M$ obeys the incoherence property with parameter $\mu_0$, then nuclear minimization succeeds with large probability if
\begin{equation}
  \label{eq:CR2}
  m \ge C \, \mu_0 n^{6/5} r \log n
\end{equation}
provided that $\mu_0 r \le n^{1/5}$. 

Now consider a matrix $M$ obeying the strong incoherence property with
$\mu = O(1)$. Then since $\mu_0 \ge 1$, \eqref{eq:CR2} guarantees
exact reconstruction only if $m \ge C \, n^{6/5} r \log n$ (and $r =
O(n^{1/5})$) while our results only need $nr \text{polylog}(n)$
samples. Hence, our results provide a substantial improvement over
\eqref{eq:CR2} at least in the regime which permits minimal sampling.

We would like to note that there are obvious relationships between the
best incoherence parameter $\mu_0$ and the best strong incoherence
parameters $\mu_1$, $\mu_2$ for a given matrix $M$, which we take to
be square for simplicity. On the one hand, \eqref{eq:grp} implies that
\[
\|P_U e_a\|^2 \le \frac{r}{n} + \frac{\mu_1 \sqrt{r}}{n}
\]
so that one can take $\mu_0 \le 1 + \mu_1/\sqrt{r}$. This shows that
one can apply results from the incoherence model (in which we only
know \eqref{eq:coherence}) to our model (in which we assume strong
incoherence).  On the other hand,
\[
|\<P_U e_a, P_U e_{a'}\>| \le \|P_U e_a\| \|P_U e_{a'}\| \le
\frac{\mu_0 r}{n}
\]
so that $\mu_1 \le \mu_0 \sqrt{r}$. Similarly, $\mu_2 \le \mu_0
\sqrt{r}$ so that one can transfer results in the other direction as
well.

We would like to mention another important paper \cite{Recht07}
inspired by compressed sensing, and which also recovers low-rank
matrices from partial information. The model in \cite{Recht07},
however, assumes some sort of Gaussian measurements and is completely
different from the completion problem discussed in this paper.

\subsubsection{Spectral methods}

An interesting new approach to the matrix completion problem has been
recently introduced in \cite{MontanariISIT}. This algorithm starts by
trimming each row and column with too few entries; i.e.~one replaces the
entries in those rows and columns by zero.  Then one computes the SVD
of the trimmed matrix and truncate it as to only keep the top $r$
singular values (note that one would need to know $r$ \emph{a priori}).
Then under some conditions (including the incoherence property \eqref{eq:coherence} with $\mu=O(1)$), this work shows that accurate recovery is
possible from a minimal number of samples, namely, on the order of $nr
\log n$ samples. Having
said this, this work is not directly comparable to ours because it
operates in a different regime. Firstly, the results are asymptotic and
are valid in a regime when the dimensions of the matrix tend to
infinity in a fixed ratio while ours are not. Secondly, there is a
strong assumption about the range of the singular values the unknown
matrix can take on while we make no such assumption; they must be
clustered so that no singular value can be too large or too small
compared to the others. Finally, this work only shows approximate
recovery---not exact recovery as we do here---although exact recovery
results have been announced.  This work is of course very interesting
because it may show that methods---other than convex
optimization---can also achieve minimal sampling bounds.

\subsection{Lower bounds}
\label{sec:lowerintro}

We would like to conclude the tour of the results introduced in this
paper with a simple lower bound, which highlights the fundamental role
played by the coherence in controlling what is
information-theoretically possible.

\begin{theorem}[Lower bound, Bernoulli model]
  \label{teo:lower} Fix $1 \leq m, r \leq n$ and $\mu_0 \geq 1$, let
  $0 < \delta < 1/2$, and suppose that we do \emph{not} have the
  condition
  \begin{equation}
    \label{eq:lower}
    - \log\Bigl(1-\frac{m}{n^2}\Bigr) \ge 
\frac{\mu_0r}{n} \log\left(\frac{n}{2\delta}\right).
  \end{equation}
  Then there exist infinitely many pairs of distinct $n \times n$
  matrices $M \neq M'$ of rank at most $r$ and obeying the incoherence
  property \eqref{eq:coherence} with parameter $\mu_0$ such that
  $\PO(M) = \PO(M')$ with probability at least $\delta$. Here, each
  entry is observed with probability $p = m/n^2$ independently from
  the others.
\end{theorem}


Clearly, even if one knows the rank and the coherence of a matrix
ahead of time, then no algorithm can be guaranteed to succeed based on
the knowledge of $\PO(M)$ only, since they are many candidates which
are consistent with these data. We prove this theorem in Section
\ref{sec:lower}.  Informally, Theorem \ref{teo:lower} asserts that
\eqref{eq:lower} is a necessary condition for matrix completion to
work with high probability if all we know about the matrix $M$ is that
it has rank at most $r$ and the incoherence property with parameter
$\mu_0$.  When the right-hand side of \eqref{eq:lower} is less than
$\eps < 1$, this implies
  \begin{equation}
    \label{eq:lower2}
    m \ge (1-\eps/2) \mu_0 n r  \log\left(\frac{n}{2\delta}\right).
\end{equation}

Recall that the number of degrees of freedom of a rank-$r$ matrix is
$2nr(1 - r/2n)$. Hence, to recover an arbitrary rank-$r$ matrix with
the incoherence property with parameter $\mu_0$ with any decent
probability by any method whatsoever, the minimum number of samples
must be about the number of degrees of freedom times $\mu_0 \log n$;
in other words, the oversampling factor is directly proportional to
the coherence. Since $\mu_0 \ge 1$, this justifies our earlier
assertions that $nr \log n$ samples are really needed.

In the {\em Bernoulli model} used in Theorem \ref{teo:lower}, the number of entries
is a binomial random variable sharply concentrating around its mean
$m$. There is very little difference between this model and the \emph{uniform model} which
assumes that $\Omega$ is sampled uniformly at random among all subsets
of cardinality $m$. Results holding for one hold for the other with
only very minor adjustments. Because we are concerned with essential
difficulties, not technical ones, we will often prove our results
using the Bernoulli model, and indicate how the results may easily be
adapted to the uniform model.


\subsection{Notation}
\label{sec:notation}

Before continuing, we provide here a brief summary of the notations
used throughout the paper. To simplify the notation, we shall work exclusively with square matrices, thus
$$ n_1 = n_2 = n.$$
The results for non-square matrices (with $n=\max(n_1,n_2)$) are proven in exactly the same fashion, but will add more subscripts to a notational system which is already quite complicated, and we will leave the details to the interested reader.  We will also assume that $n \geq C$ for some sufficiently large absolute constant $C$, as our results are vacuous in the regime $n=O(1)$.

Throughout, we will always assume that $m$
is at least as large as $2nr$, thus
\begin{equation}\label{rnp}
2r \leq np, \qquad p := m/n^2. 
\end{equation}

A variety of norms on matrices $X \in \R^{n \times n}$ will be
discussed. The \emph{spectral norm} (or \emph{operator norm}) of a
matrix is denoted by
$$\|\mtx{X}\| := \sup_{x \in \R^n: \|x\|=1} \|\mtx{X} x\| = \sup_{1 \leq j \leq n} \sigma_j(\mtx{X}).$$
The Euclidean inner product between two matrices is defined by the
formula
$$\<\mtx{X}, \mtx{Y}\> := \trace(\mtx{X}^* \mtx{Y}),$$ 
and the corresponding Euclidean norm, called the \emph{Frobenius norm}
or \emph{Hilbert-Schmidt norm}, is denoted
$$\|\mtx{X}\|_F :=\<\mtx{X},\mtx{X}\>^{1/2}  = (\sum_{j=1}^n \sigma_j(\mtx{X})^2)^{1/2}.$$  
The \emph{nuclear norm} of a matrix $\mtx{X}$ is denoted
$$\|\mtx{X}\|_* := \sum_{j=1}^n \sigma_j(\mtx{X}).$$
For vectors, we will only consider the usual Euclidean $\ell_2$ norm
which we simply write as $\|\vct{x}\|$.

Further, we will also manipulate linear transformation which acts on
the space $\R^{n \times n}$ matrices such as $\PO$, and we will use
calligraphic letters for these operators as in ${\mathcal A}(\mtx{X})$.
In particular, the identity operator on this space will be denoted by
$\OpId: \R^{n \times n} \to \R^{n \times n}$, and should \emph{not} be
confused with the identity matrix $I \in \R^{n \times n}$. The only
norm we will consider for these operators is their spectral norm (the
top singular value)
$$\|{\mathcal A}\| := \sup_{\mtx{X} : \|\mtx{X}\|_F \le 1} \,
\|{\mathcal A}(\mtx{X})\|_F.$$
Thus for instance
$$\|\PO\| = 1.$$ 

We use the usual asymptotic notation, for instance writing $O(M)$ to
denote a quantity bounded in magnitude by $CM$ for some absolute
constant $C > 0$.  We will sometimes raise such notation to some
power, for instance $O(M)^M$ would denote a quantity bounded in
magnitude by $(CM)^M$ for some absolute constant $C>0$.  We also write
$X \lesssim Y$ for $X=O(Y)$, and $\text{poly}(X)$ for
$O(1+|X|)^{O(1)}$.  

We use $1_E$ to denote the indicator function of an event $E$,
e.g. $1_{a=a'}$ equals $1$ when $a=a'$ and $0$ when $a \neq a'$.

If $A$ is a finite set, we use $|A|$ to denote its cardinality.

We record some (standard) conventions involving empty sets.  The set
$[n] := \{1,\ldots,n\}$ is understood to be the empty set when $n=0$.
We also make the usual conventions that an empty sum $\sum_{x \in
  \emptyset} f(x)$ is zero, and an empty product $\prod_{x \in
  \emptyset} f(x)$ is one.  Note however that a $k$-fold sum such as
$\sum_{a_1,\ldots,a_k \in [n]} f(a_1,\ldots,a_k)$ does not vanish when
$k=0$, but is instead equal to a single summand $f()$ with the empty
tuple $() \in [n]^0$ as the input; thus for instance the identity
$$\sum_{a_1,\ldots,a_k \in [n]} \prod_{i=1}^k f(a_i) = \Bigl(\sum_{a \in [n]} f(a)\Bigr)^k$$
is valid both for positive integers $k$ and for $k=0$ (and both for
non-zero $f$ and for zero $f$, recalling of course that $0^0=1$).  We
will refer to sums over the empty tuple as \emph{trivial sums} to
distinguish them from \emph{empty sums}.

\section{Lower bounds}
\label{sec:lower}

This section proves Theorem \ref{teo:lower}, which asserts that no
method can recover an arbitrary $n \times n$ matrix of rank $r$ and
coherence at most $\mu_0$ unless the number of random samples obeys
\eqref{eq:lower}. As stated in the theorem, we establish lower bounds
for the Bernoulli model, which then apply to the model where exactly
$m$ entries are selected uniformly at random, see the Appendix for
details.

It may be best to consider a simple example first to understand the
main idea behind the proof of Theorem \ref{teo:lower}. Suppose that $r
= 1$, $\mu_0 > 1$ in which case $M = x y^*$.  For simplicity, suppose
that $y$ is fixed, say $y = (1,\ldots,1)$, and $x$ is chosen arbitrarily 
from the cube $[1,\sqrt{\mu_0}]^n$ of $\R^n$.  One easily
verifies that $M$ obeys the coherence property with parameter $\mu_0$
(and in fact also obeys the strong incoherence property with a
comparable parameter).  Then to recover $M$, we need to see at least
one entry per row. For instance, if the first row is unsampled, one
has no information about the first coordinate $x_1$ of $x$ other than
that it lies in $[1,\sqrt{\mu_0}]$, and so the claim follows in this case by varying $x_1$ along the infinite set $[1,\sqrt{\mu_0}]$.

Now under the Bernoulli model, the number of observed entries in the
first row---and in any fixed row or column---is a binomial random
variable with a number of trials equal to $n$ and a probability of
success equal to $p$.  Therefore, the probability $\pi_0$ that any row
is unsampled is equal to $\pi_0 = (1-p)^n$. By independence, the
probability that all rows are sampled at least once is $(1-\pi_0)^n$,
and any method succeeding with probability greater $1-\delta$ would
need
\[
(1-\pi_0)^n \ge 1-\delta. 
\]
or $-n\pi_0 \ge n\log(1-\pi_0) \ge \log (1-\delta)$. When $\delta <
1/2$, $\log (1-\delta) \ge -2\delta$ and thus, any method would
need
\[
\pi_0 \le \frac{2\delta}{n}.
\] 
This is the desired conclusion when $\mu_0 > 1$, $r = 1$.

This type of simple analysis easily extends to general values of the
rank $r$ and of the coherence. Without loss of generality, assume that
$\ell := \frac{n}{\mu_0 r}$ is an integer, and consider a
(self-adjoint) $n \times n$ matrix $M$ of rank $r$ of the form
\[
M := \sum_{k = 1}^r \sigma_k u_k u_k^*,
\]
where the $\sigma_k$ are drawn arbitrarily from $[0,1]$ (say), and the singular vectors $u_1,\ldots,u_r$ are defined as follows:
\[
u_{i,k} := \sqrt{\frac{1}{\ell}} \sum_{i \in B_k} e_i, \quad B_k =
\{(k-1)\ell + 1, (k-1)\ell + 2, \ldots, k\ell\};
\]
that is to say, $u_{k}$ vanishes everywhere except on a support of
$\ell$ consecutive indices. Clearly, this matrix is incoherent with parameter $\mu_0$.  Because the supports of the singular
vectors are disjoint, $M$ is a block-diagonal matrix with diagonal
blocks of size $\ell \times \ell$. We now argue as before. Recovery with positive probability is
impossible unless we have sampled at least one entry per row of each diagonal
block, since otherwise we would be forced to guess at least one of the $\sigma_k$ based on no information (other than that $\sigma_k$ lies in $[0,1]$), and the theorem will follow by varying this singular value. Now the probability $\pi_0$ that the first row of the first
block---and any fixed row of any fixed block---is unsampled is equal
to $(1-p)^\ell$.  Therefore, any method succeeding with probability
greater $1-\delta$ would need
\[
(1-\pi_1)^{n} \ge 1-\delta,
\]
which implies $\pi_1 \le 2\delta/n$ just as before. With $\pi_1 =
(1-p)^{\ell}$, this gives \eqref{eq:lower} under the Bernoulli
model. The second part of the theorem, namely, \eqref{eq:lower2}
follows from the equivalent characterization 
\[
m \ge n^2\bigl(1 - e^{-\frac{\mu_0 r}{n} \log(n/2\delta)}\bigr)
\]
together with $1 - e^{-x} > x - x^2/2$ whenever $x \ge 0$. 


\section{Strategy and Novelty}
\label{sec:strategy}

This section outlines the strategy for proving our main results,
Theorems \ref{teo:main1} and \ref{teo:main2}. The proofs of these
theorems are the same up to a point where the arguments to estimate
the moments of a certain random matrix differ. In this section, we
present the common part of the proof, leading to two key moment
estimates, while the proofs of these crucial estimates are the object
of later sections.

One can of course prove our claims for the Bernoulli model with $p =
m/n^2$ and transfer the results to the uniform model, by using the arguments in the appendix.  For example,
the probability that the recovery via \eqref{eq:cvx2} is not exact is
at most twice that under the Bernoulli model.

\subsection{Duality}
\label{sec:duality}

We begin by recalling some calculations from \cite[Section 3]{CR}.
From standard duality theory, we know that the correct matrix $M \in
\R^{n \times n}$ is a solution to \eqref{eq:cvx2} if and only if
there exists a dual certificate $Y \in \R^{n_1\times n_2}$ with the
property that $\PO(Y)$ is a subgradient of the nuclear norm at
$M$, which we write as
\begin{equation}
  \label{eq:dual}
  \PO(Y) \in \partial \|M\|_*.
\end{equation}

We recall the projection matrices $P_U, P_V$ and the companion matrix
$E$ defined by \eqref{puv}, \eqref{eq:E}. It is known
\cite{Lew:MP:03,Wat:LAA:92} that
\begin{equation}\label{eqn:subdiffNorm}
  \partial\|\mtx{M}\|_*=\left\{E + \mtx{W} :
    ~\mtx{W}\in\mathbb{R}^{n \times n},~~
    P_U \mtx{W}=0,~~
    \mtx{W} P_V =0,~~
    \|\mtx{W}\| \leq 1\right\}.
\end{equation}
There is a more compact way to write \eqref{eqn:subdiffNorm}. Let $T \subset \R^{n \times n}$
be the span of matrices of the form $\vct{u}_k \vct{y}^*$ and $\vct{x}
\vct{v}_k^*$ and let $T^\perp$ be its orthogonal complement. Let $\PT: \R^{n \times n} \to T$ be the orthogonal projection onto $T$; one easily verifies the explicit formula
\begin{equation}
  \label{eq:PT}
  \PT(X) = P_U X + X P_V - P_U X P_V,
\end{equation}
and note that the complementary projection $\PTp := \OpId - \PT$ is given by the formula
\begin{equation}
\label{eq:PTp}
\PTp(X) = (I-P_U) X (I - P_V). 
\end{equation}
In particular, $\PTp$ is a contraction:
\begin{equation}\label{ptp-contract}
\| \PTp \| \leq 1.
\end{equation}
Then $Z \in \partial\|\mtx{X}\|_*$ if and only if 
\[
\PT(Z) =  E, \quad \text{and } \quad \|\PTp(Z)\| \le 1. 
\]
With these preliminaries in place, \cite{CR} establishes the following
result.
\begin{lemma}[Dual certificate implies matrix completion]
\label{teo:dual} Let the notation be as above.
Suppose that the following two conditions hold: 
\begin{enumerate} 
  \item There exists $\mtx{Y} \in \R^{n \times n}$ obeying 
    \begin{enumerate}
    \item $\PO(\mtx{Y}) = \mtx{Y}$,
    \item $\PT(\mtx{Y}) = E$, and 
    \item $\|{\mathcal P}_{T^\perp}(\mtx{Y})\| < 1$. 
    \end{enumerate}
  \item The restriction $\PO\downharpoonright_T: T \to \PO(\R^{n \times n})$ of the (sampling) operator $\PO$ restricted to $T$ is injective. 
  \end{enumerate}
Then $\mtx{M}$ is the unique solution to the convex program \eqref{eq:cvx2}.
\end{lemma}

\begin{proof} See \cite[Lemma 3.1]{CR}.
\end{proof}

The second sufficient condition, namely, the injectivity of the
restriction to $\PO$ has been studied in \cite{CR}. We recall a
useful result.

\begin{theorem}[Rudelson selection estimate]\label{teo:rudelson}  \cite[Theorem 4.1]{CR}
  Suppose $\Omega$ is sampled according to the Bernoulli model and put
  $n := \max(n_1, n_2)$. Assume that $M$ obeys \eqref{eq:coherence}.
  Then there is a numerical constant $C_R$ such that for all
  $\beta>1$, we have the bound
  \begin{equation}
    \label{eq:near-isometry}
    p^{-1} \, \|\PT \PO \PT - p \PT\| \le a
  \end{equation}
  with probability at least $1-3n^{-\beta}$ provided that $a < 1$, where $a$ is the quantity
  \begin{equation}\label{adef}
  a := C_R \sqrt{\frac{\mu_0 \, nr (\beta \log n)}{m}}
  \end{equation}
\end{theorem}

We will apply this theorem with $\beta := 4$ (say). The statement \eqref{eq:near-isometry} is stronger than the injectivity of the restriction of
$\PO$ to $T$. Indeed, take $m$ sufficiently large so
that the $a < 1$. Then if $X \in T$, we have
\[
\|\PT \PO(X) - p X\|_F < a p \|X\|_F, 
\]
and obviously, $\PO(X)$ cannot vanish unless $X = 0$. 

In order for the condition $a < 1$ to hold, we must have 
\begin{equation}\label{m01}
m \geq C_0 \mu_0 n r \log n
\end{equation}
for a suitably large constant $C_0$.  But this follows from the hypotheses in either Theorem \ref{teo:main1} or Theorem \ref{teo:main2}, for reasons that we now pause to explain.  In either of these theorems we have 
\begin{equation}\label{m11}
m \geq C_1 \mu n r \log n
\end{equation}
for some large constant $C_1$.  Recall from Section \ref{inco} that $\mu_0 \le 1 + \mu_1/\sqrt{r} \leq 1 + \mu/\sqrt{r}$, and so \eqref{m11} implies \eqref{m01} whenever $\mu_0 \geq 2$ (say).  When $\mu_0 < 2$, we can also deduce \eqref{m01} from \eqref{m11} by applying the trivial bound $\mu \geq 1$ noted in the introduction.

In summary, to prove Theorem \ref{teo:main1} or Theorem \ref{teo:main2}, it suffices (under the hypotheses of these theorems) to exhibit a dual matrix $Y$ obeying the first sufficient condition of Lemma \ref{teo:dual}, with probability at least $1-n^{-3}/2$ (say).  This is the objective of the remaining sections of the paper.

\subsection{The dual certificate}
\label{sec:certificate}

Whenever the map $\PO\downharpoonright_T: T \to \PO(\R^{n \times n})$ restricted to $T$ is injective, the linear map 
\[
\begin{array}{lll}
T & \goto & T\\
X & \mapsto & \PT\PO\PT(X)
\end{array}
\]
is invertible, and we denote its inverse by
$(\PT\PO\PT)^{-1}: T \to T$. Introduce the dual matrix $Y \in \PO(\R^{n \times n}) \subset \R^{n \times n}$ defined via
\begin{equation}
  \label{eq:ansatz}
   Y = \PO \PT (\PT \PO \PT)^{-1} E. 
\end{equation}
By construction, $\PO(Y) = Y$, $\PT(Y) = E$ and, therefore, we will
establish that $M$ is the unique minimizer if one can show that 
\begin{equation}
\label{eq:toshow}
\|\PTp(Y) \| < 1.
\end{equation}
The dual matrix $Y$ would then certify that $M$ is the unique
solution, and this is the reason why we will refer to $Y$ as a
\emph{candidate certificate}.  This certificate was also used in \cite{CR}.

Before continuing, we would like to offer a little motivation for the
choice of the dual matrix $Y$. It is not difficult to check that
\eqref{eq:ansatz} is actually the solution to the following problem:
  \[
  \begin{array}{ll}
    \text{minimize}   & \quad \|\mtx{Z}\|_F\\
    \text{subject to} & \quad \PT \PO(\mtx{Z}) = E. 
 \end{array}
\]
Note that by the Pythagorean identity, $Y$ obeys
\[
\|\mtx{Y}\|_F^2 = \|\PT(\mtx{Y})\|_F^2 + \|{\mathcal
  P}_{T^\perp}(\mtx{Y})\|_F^2 = r + \|{\mathcal P}_{T^\perp}(\mtx{Y})\|_F^2. 
\]
The interpretation is now clear: among all matrices obeying $\PO(Z) =
Z$ and $\PT(Z) = E$, $Y$ is that element which minimizes $\|{\mathcal
  P}_{T^\perp}(\mtx{Z})\|_F$. By forcing the Frobenius norm of ${\mathcal
  P}_{T^\perp}(\mtx{Y})$ to be small, it is reasonable to expect that
its spectral norm will be sufficiently small as well. In that sense,
$Y$ defined via \eqref{eq:ansatz} is a very suitable candidate. 

Even though this is a different problem, our candidate certificate
resembles---and is inspired by---that constructed in \cite{CRT:TIT:06}
to show that $\ell_1$ minimization recovers sparse vectors from
minimally sampled data.

\subsection{The Neumann series}
\label{sec:neuman}


We now develop a useful formula for the candidate certificate, and begin by introducing a normalized version $\QO: \R^{n \times n} \to \R^{n \times n}$ of $\PO$, defined by the formula
\begin{equation}
  \label{eq:QO}
  \QO := \frac{1}{p} \PO - \OpId
\end{equation}
where $\OpId: \R^{n \times n} \to \R^{n \times n}$ is the identity operator on matrices (\emph{not} the identity matrix $I \in \R^{n \times n}$!).  
Note that with the Bernoulli model for selecting $\Omega$, that $\QO$ has expectation zero.

From \eqref{eq:QO} we have $\PT \PO
\PT = p \PT(\OpId + \QO)\PT$, and owing to Theorem \ref{teo:rudelson}, one
can write $(\PT \PO \PT)^{-1}$ as the convergent Neumann series
\[
p (\PT \PO \PT)^{-1} =  \sum_{k \ge 0} (-1)^k (\PT \QO \PT)^k. 
\]
From the identity $\PTp \PT = 0$ we conclude that $\PTp \PO \PT = p (\PTp \QO \PT)$.  One can therefore express the candidate certificate $Y$ \eqref{eq:ansatz} as
\begin{align*}
  \PTp(Y) & = \sum_{k \ge 0} (-1)^k \PTp \QO  (\PT \QO \PT)^k (E)\\
  & = \sum_{k \ge 0} (-1)^k \PTp (\QO \PT)^k \QO (E),
\end{align*}
where we have used $\PT^2 = \PT$ and $\PT(E) = E$. By the triangle inequality and \eqref{ptp-contract}, it thus suffices to show that
\[
\sum_{k \ge 0} \|(\QO \PT)^k \QO (E)\| < 1
\]
with probability at least $1-n^{-3}/2$.

It is not hard to bound the tail of the series thanks to Theorem
\ref{teo:rudelson}. First, this theorem bounds the spectral norm of
$\PT\QO \PT$ by the quantity $a$ in \eqref{adef}. This gives that for each $k \ge
1$, $\|(\PT\QO\PT)^k(E)\|_F < a^k \|E\|_F = a^k \sqrt{r}$ and,
therefore,
\[
\|(\QO \PT)^k \QO (E)\|_F  = \|\QO\PT (\PT\QO\PT)^k (E)\|_F \le 
\|\QO\PT\| a^k \sqrt{r}. 
\]
Second, this theorem also bounds $\|\QO\PT\|$ (recall that this is the
spectral norm) since
\[
\|\QO\PT\|^2 = \max_{\|X\|_F \le 1} \<\QO \PT(X),\QO \PT(X)\> =
\<X,\PT \QO^2 \PT(X)\>. 
\]
Expanding the identity $\PO^2 = \PO$ in terms of $\QO$, we obtain
\begin{equation}
  \label{eq:QOsq}
\QO^2 = \frac{1}{p}[(1-2p) \QO + (1-p) \OpId], 
\end{equation}
and thus, for all $\|X\|_F \le 1$, 
\[
p \<X,\PT \QO^2 \PT(X)\> = (1-2p) \<X, \PT\QO \PT(X)\> + (1-p)
\|\PT(X)\|_F^2 \le a + 1.
\]
Hence $\|\QO\PT\| \le \sqrt{(a+1)/p}$. For each $k_0 \ge 0$, this
gives
\[
\sum_{k \ge k_0} \|(\QO \PT)^k \QO (E)\|_F \le \sqrt{\frac{3r}{2p}} \sum_{k \ge
  k_0} a^k \le \sqrt{\frac{6r}{p}} a^{k_0}
\]
provided that $a < 1/2$. With $p = m/n^2$ and $a$ defined by \eqref{adef} with $\beta=4$, we have 
\begin{equation*}
   \sum_{k \ge k_0} \|(\QO \PT)^k \QO (E)\|_F \le \sqrt{n} \times
  O\left(\frac{\mu_0 nr \log n}{m}\right)^{\frac{k_0+1}{2}}
\end{equation*}
with probability at least $1 - n^{-4}$.  When $k_0+1 \ge \log n$,
$n^{\frac{1}{k_0+1}} \le n^{\frac{1}{\log n}} = e$ and thus for each
such a $k_0$,
\begin{equation}
\label{eq:tailseries}
\sum_{k \ge k_0} \|(\QO \PT)^k \QO (E)\|_F \le
O\left(\frac{\mu_0 nr \log n}{m}\right)^{\frac{k_0+1}{2}}
\end{equation}
with the same probability. 

To summarize this section, we conclude that since both our results
assume that $m \ge c_0 \mu_0 nr \log n$ for some sufficiently large
numerical constant $c_0$ (see the discussion at the end of Section \ref{sec:duality}), it now suffices to show that 
\begin{equation}
\label{eq:kth}
\sum_{k=0}^{\lfloor \log n \rfloor} \|(\QO \PT)^k \QO E \| \le \frac{1}{2}
\end{equation}
(say) with probability at least $1-n^{-3}/4$ (say).

\subsection{Centering}

We have already normalised $\PO$ to have ``mean zero'' in some
sense by replacing it with $\QO$.  Now we perform a similar
operation for the projection $\PT: X \mapsto P_U X + X P_V - P_U X P_V$. The
eigenvalues of $\PT$ are centered around
\begin{equation}
\label{eq:rho}
\rho' := \trace(\PT)/n^2 = 2\rho - \rho^2, \quad \rho := r/n, 
\end{equation}
as this follows from the fact that $\PT$ is a an orthogonal projection
onto a space of dimension $2nr - r^2$.  Therefore, we simply split
$\PT$ as
\begin{equation}
\label{eq:QT}
 \PT = \QT + \rho' \OpId, 
\end{equation}
so that the eigenvalues of $\QT$ are centered around zero.  From now
on, $\rho$ and $\rho'$ will always be the numbers defined above. 

\begin{lemma}[Replacing $\PT$ with $\QT$]
\label{teo:equiv}  Let $0 < \sigma < 1$.
Consider the event such that
  \begin{equation}
    \label{eq:QTbound}
    \|(\QO \QT)^k \QO(E)\| \le \sigma^{\frac{k+1}{2}}, \quad \text{for all } 0 \le k < k_0.  
  \end{equation}
  Then on this event, we have that for all $0 \le k < k_0$,
 \begin{equation}
    \label{eq:PTbound}
    \|(\QO \PT)^k \QO(E)\| \le (1+4^{k+1}) \, \sigma^{\frac{k+1}{2}},  
  \end{equation}
  provided that $8nr/m < \sigma^{3/2}$.
\end{lemma}

From \eqref{eq:PTbound} and the geometric series formula we obtain the corollary
\begin{equation}
\label{eq:earlyseries}
\sum_{k = 0}^{k_0-1}
  \|(\QO \PT)^k \QO(E)\| \le 5\sqrt{\sigma} \frac{1}{1-4\sqrt{\sigma}}. 
\end{equation}
Let $\sigma_0$ be such that the right-hand side is less than 1/4, say.
Applying this with $\sigma = \sigma_0$, we conclude that to prove
\eqref{eq:kth} with probability at least $1-n^{-3}/4$, it suffices by
the union bound to show that \eqref{eq:QTbound} for this value of
$\sigma$.  (Note that the hypothesis $8nr/m < \sigma^{3/2}$ follows
from the hypotheses in either Theorem \ref{teo:main1} or Theorem
\ref{teo:main2}.)

Lemma \ref{teo:equiv}, which is proven in the Appendix, is useful
because the operator $\QT$ is easier to work with than $\PT$ in the
sense that it is more homogeneous, and obeys better estimates.  If we
split the projections $P_U, P_V$ as
\begin{equation}\label{prhoq}
P_U = \rho I + Q_U, \quad P_V = \rho I + Q_V,
\end{equation}
then $\QT$ obeys
\[
\QT(X) = (1-\rho) Q_U X + (1-\rho) X Q_V - Q_U X Q_V.
\]
Let $\U_{a,a'}, \V_{b,b'}$ denote the matrix elements of $Q_U, Q_V$:
\begin{equation}
\label{eq:UV}
\U_{a, a'} := \<e_a, Q_U e_{a'}\> = \<e_a, P_U e_{a'}\> - \rho
1_{a = a'},
\end{equation}
and similarly for $\V_{b,b'}$. The coefficients $c_{ab,a'b'}$ of $\QT$
obey 
\begin{equation}\label{cab}
  c_{ab,a'b'} := \<e_a e_b^*, \QT(e_{a'} e_{b'})\> 
=  (1 - \rho) 1_{b = b'} \U_{a,a'} +
  (1 - \rho) 1_{a = a'} \V_{b,b'} - \U_{a,a'} \V_{b,b'}.  
\end{equation}
An immediate consequence of this under the assumptions \eqref{eq:grp}, is the
estimate
\begin{equation}
  \label{eq:cbound}
  |c_{ab,a'b'}| \lesssim (1_{a = a'} + 1_{b = b'})
\frac{\mu \sqrt{r}}{n} + \frac{\mu^2 r}{n^2}.
\end{equation}
When $\mu = O(1)$, these coefficients are bounded by $O(\sqrt{r}/n)$
when $a = a'$ or $b = b'$ while in contrast, if we stayed with $\PT$
rather than $\QT$, the diagonal coefficients would be as large as
$r/n$. However, our lemma states that bounding $\|(\QO \QT)^k
\QO(E)\|$ automatically bounds $\|(\QO \PT)^k \QO(E)\|$ by nearly the
same quantity.  This is the main advantage of replacing the $\PT$ by the $\QT$ in our analysis.

\subsection{Key estimates}
\label{sec:key}

To summarize the previous discussion, and in particular the bounds \eqref{eq:earlyseries} and \eqref{eq:tailseries}, we see everything
reduces to bounding the spectral norm of $(\QO \QT)^k
\QO(E)$ for $k = 0, 1, \ldots, \lfloor \log n \rfloor$. Providing
good upper bounds on these quantities is the crux of the argument. We use the moment method, controlling a spectral norm a matrix by the trace of a high power of that matrix.  We will prove two moment estimates which ultimately imply our two main results (Theorems \ref{teo:main1} and \ref{teo:main2}) respectively.  The first such estimate is as follows:

\begin{theorem}[Moment bound I]
  \label{teo:moment1}
  Set $A = (\QO \QT)^k \QO(E)$ for a fixed $k \ge 0$. Under
  the assumptions of Theorem \ref{teo:main1}, we have that for each
  $j > 0$,
\begin{equation}
\label{eq:moment1}
\E \bigl[\trace(A^* A)^{j}\bigr] = O\bigl(j(k+1)\bigr)^{2j(k+1)}
n \Bigl(\frac{ n \rmu^2}{m}\Bigr)^{j(k+1)}, \quad \rmu := \mu^2 r,  
\end{equation}
provided that $m \ge n \rmu^2$ and $n \ge c_0 j(k+1)$ for some
numerical constant $c_0$.
\end{theorem}
By Markov's inequality, this result automatically estimates the norm of
$(\QO \QT)^k \QO(E)$ and immediately gives the following
corollary. 
\begin{corollary}[Existence of dual certificate I]
  \label{teo:dual1}
  Under the assumptions of Theorem \ref{teo:main1}, the matrix $Y$
  \eqref{eq:ansatz} is a dual certificate, and obeys $\|\PTp(Y)\| \le
  1/2$ with probability at least $1 - n^{-3}$ provided that $m$ obeys
  \eqref{eq:main1}.
\end{corollary}
\begin{proof}
  Set $A = (\QO \QT)^k \QO(E)$ with $k \le \log n$, and set $\sigma
  \le \sigma_0$. By Markov's inequality
\[
\P(\|A\| \ge \sigma^{\frac{k+1}{2}}) \le \frac{\E \|A\|^{2j}}{\sigma^{j(k+1)}}, 
\]
Now choose $j > 0$ to be the smallest integer such that $j(k+1) \ge \log
n$. Since 
\[
\| A\|^{2j} \leq \trace(A^* A)^j,
\]
Theorem \ref{teo:main1} gives
\[ 
\P\bigl(\|A\| \ge \sigma^{\frac{k+1}{2}}\bigr) \le \gamma^{j(k+1)}
\]
for some
\[ \gamma = O\Bigl(\frac{(j(k+1))^2
    n \rmu^2}{a \,m}\Bigr)
\]
 where we have used the fact that $n^{\frac{1}{j(k+1)}} \le
 n^{\frac{1}{\log n}} = e$. Hence, if
\begin{equation}
\label{eq:mint}
m \ge C_0 \frac{n \rmu^2 (\log n)^2}{\sigma}, 
\end{equation}
for some numerical constant $C_0$, we have $\gamma < 1/4$ and 
\[
\P\bigl(\|(\QO \QT)^k \QO(E)\| \ge \sigma^{\frac{k+1}{2}}\bigr) \le n^{-4}.
\]
Therefore,
\[
\bigcup_{0 \le k < \log n} \{(\QO \QT)^k \QO(E)\| \ge
a^{\frac{k+1}{2}}\} 
\]
has probability less or equal to $n^{-4} \log n \le n^{-3}/2$ for $n
\ge 2$. Since the corollary assumes $r = O(1)$, then \eqref{eq:mint}
together with \eqref{eq:earlyseries} and \eqref{eq:tailseries} prove
the claim thanks to our choice of $\sigma$.
\end{proof}


Of course, Theorem \ref{teo:main1} follows immediately from Corollary \ref{teo:dual1} and Lemma \ref{teo:dual}.
 In the same way, our second result (Theorem \ref{teo:main2}) follows from a more refined estimate stated below.

\begin{theorem}[Moment bound II]
\label{teo:moment2}  
 Set $A = (\QO \QT)^k \QO(E)$ for a fixed $k \ge 0$. Under
  the assumptions of Theorem \ref{teo:main2}, we have that for each
  $j > 0$ ($\rmu$ is given in \eqref{eq:moment1}),
\begin{equation}
\label{eq:moment2}
\E \bigl[\trace(A^* A)^{j}\bigr] \le 
\Bigl(\frac{(j(k+1))^6 n \rmu}{m}\Bigr)^{j(k+1)}
\end{equation}
provided that $n \ge c_0 j(k+1)$ for some numerical constant $c_0$.
\end{theorem}
Just as before, this theorem immediately implies the following
corollary.
\begin{corollary}[Existence of dual certificate II]
  \label{teo:dual2}
  Under the assumptions of Theorem \ref{teo:main2}, the matrix $Y$
  \eqref{eq:ansatz} is a dual certificate, and obeys $\|\PTp(Y)\| \le
  1/2$ with probability at least $1 - n^{-3}$ provided that $m$ obeys
  \eqref{eq:main2}.
\end{corollary}
The proof is identical to that of Corollary \ref{teo:dual1} and is
omitted.  Again, Corollary \ref{teo:dual2} and Lemma \ref{teo:dual} immediately imply Theorem \ref{teo:main2}.

We have learned that verifying that $Y$ is a valid dual certificate
reduces to \eqref{eq:moment1} and \eqref{eq:moment2}, and we conclude
this section by giving a road map to the proofs. In Section
\ref{sec:moments}, we will develop a formula for $\E \trace(A^*
A)^{j}$, which is our starting point for bounding this quantity. Then
Section \ref{sec:moment1} develops the first and perhaps easier bound
\eqref{eq:moment1} while Section \ref{sec:moment2} refines the
argument by exploiting clever cancellations, and establishes the nearly
optimal bound \eqref{eq:moment2}.

\subsection{Novelty}
\label{sec:novelty}

As explained earlier, this paper derives near-optimal sampling results
which are stronger than those in \cite{CR}. One of the reasons
underlying this improvement is that we use completely different
techniques. In details, \cite{CR} constructs the dual certificate
\eqref{eq:ansatz} and proceeds by showing that $\|\PTp(Y)\| < 1$ by
bounding each term in the series $\sum_{k \ge 0} \|(\QO \PT)^k \QO
(E)\| < 1$. Further, to prove that the early terms (small values of
$k$) are appropriately small, the authors employ a sophisticated array
of tools from asymptotic geometric analysis, including noncommutative
Khintchine inequalities \cite{LustPicquard86}, decoupling techniques
of Bourgain and Tzafiri and of de la Pe\~na \cite{delaPena2}, and
large deviations inequalities \cite{ledoux01co}.  They bound each term
individually up to $k = 4$ and use the same argument as that in
Section \ref{sec:neuman} to bound the rest of the series. Since the
tail starts at $k_0 = 5$, this gives that a sufficient condition is
that the number of samples exceeds a constant times $\mu_0 n^{6/5} nr
\log n$. Bounding each term $\|(\QO \PT)^k \QO (E)\|^k$ with the tools
put forth in \cite{CR} for larger values of $k$ becomes increasingly
delicate because of the coupling between the indicator variables
defining the random set $\Omega$. In addition, the noncommutative
Khintchine inequality seems less effective in higher dimensions; that
is, for large values of $k$.  Informally speaking, the reason for this
seems to be that the types of random sums that appear in the moments
$(\QO \PT)^k \QO (E)$ for large $k$ involve complicated combinations
of the coefficients of $\PT$ that are not simply components of some
product matrix, and which do not simplify substantially after a direct
application of the Khintchine inequality.

In this paper, we use a very different strategy to estimate the
spectral norm of $(\QO \QT)^k \QO (E)$, and employ moment methods,
which have a long history in random matrix theory, dating back at least to the classical work of Wigner\cite{Wigner}. We raise the matrix
$A := (\QO \QT)^k \QO (E)$ to a large power $j$ so that
\[
\sigma_1^{2j}(A) = \|A\|^{2j} \approx \trace (A^*A)^j =
\sum_{i \in [n]} \sigma_i^{2j}(A)
\]
(the largest element dominates the sum). We then need to compute the
expectation of the right-hand side, and reduce matters to a purely
combinatorial question involving the statistics of various types of
paths in a plane. It is rather remarkable that carrying out these
combinatorial calculations nearly give the quantitatively correct
answer; the moment method seems to come close to giving the ultimate limit of performance one can expect from
nuclear-norm minimization.

As we shall shortly see, the expression $\trace (A^* A)^j$ expands as a sum over ``paths'' of products of various coefficients of the operators $\QO, \QT$ and the matrix $E$.  These paths can be viewed as complicated variants of Dyck paths.  However, it does not seem that one can simply invoke standard moment method calculations in the literature to compute this sum, as in order to obtain efficient bounds, we will need to take full advantage of identities such as $\PT \PT = \PT$ (which capture certain cancellation properties of the coefficients of $\PT$ or $\QT$) to simplify various components of this sum.  It is only after performing such simplifications that one can afford to estimate all the coefficients by absolute values and count paths to conclude the argument.

\section{Moments}
\label{sec:moments}

Let $j \geq 0$ be a fixed integer.
The goal of this section is to develop a formula for
\begin{equation}
  \label{eq:moment}
   X := \E \tr(A^* A)^j.
\end{equation}
This will clearly be of use in the proofs of the moment bounds (Theorems \ref{teo:moment1}, \ref{teo:moment2}).

\subsection{First step: expansion}

We first write the matrix $A$ in components as
\[
A = \sum_{a,b \in [n]} A_{ab} e_{ab}
\]
for some scalars $A_{ab}$, where $e_{ab}$ is the standard basis for
the $n \times n$ matrices and $A_{ab}$ is the $(a,b)^{\operatorname{th}}$ entry of
$A$. Then
\[
\tr (A^* A)^j = \sum_{\substack{a_1,\ldots,a_j \in [n]\\
b_1, \ldots,b_j \in [n]}} \prod_{i \in [j]} A_{a_i b_i} A_{a_{i+1} b_i},
\]
where we adopt the cyclic convention $a_{j+1} = a_1$.  Equivalently, we can write
\begin{equation}\label{sumpath}
\tr (A^* A)^j = \sum \prod_{i \in [j]} \prod_{\mu=0}^1 A_{a_{i,\mu} b_{i,\mu}},
\end{equation}
where the sum is over all $a_{i,\mu}, b_{i,\mu} \in [n]$ for $i \in [j], \mu \in \{0,1\}$ obeying the compatibility conditions
$$ a_{i,1} = a_{i+1,0}; \quad b_{i,1} = b_{i,0} \hbox{ for all } i \in [j]$$
with the cyclic convention $a_{j+1,0}=a_{1,0}$. 

\emph{Example.}  If $j=2$, then we can write $\tr( A^* A )^j$ as
$$ \sum_{a_1,a_2,b_1,b_2 \in [n]} A_{a_1 b_1} A_{a_2 b_1} A_{a_2 b_2} A_{a_2 b_1}.$$
or equivalently as
$$ \sum \prod_{i=1}^2 \prod_{\mu=0}^1 A_{a_{i,\mu},b_{i,\mu}}$$
where the sum is over all $a_{1,0}, a_{1,1}, a_{2,0}, a_{2,1}, b_{1,0}, b_{1,1}, b_{2,0}, b_{2,1} \in [n]$ obeying the compatibility conditions
$$ a_{1,1} = a_{2,0}; \quad a_{2,1} = a_{1,0}; \quad b_{1,1} = b_{1,0}; \quad b_{2,1} = b_{2,0}.$$

\emph{Remark.}  The sum in \eqref{sumpath} can be viewed as over all
closed paths of length $2j$ in $[n] \times [n]$, where the edges of
the paths alternate between ``horizontal rook moves'' and ``vertical
rook moves'' respectively; see Figure \ref{fig:fig1}.
\begin{figure}
 \centering
\includegraphics[scale=.5]{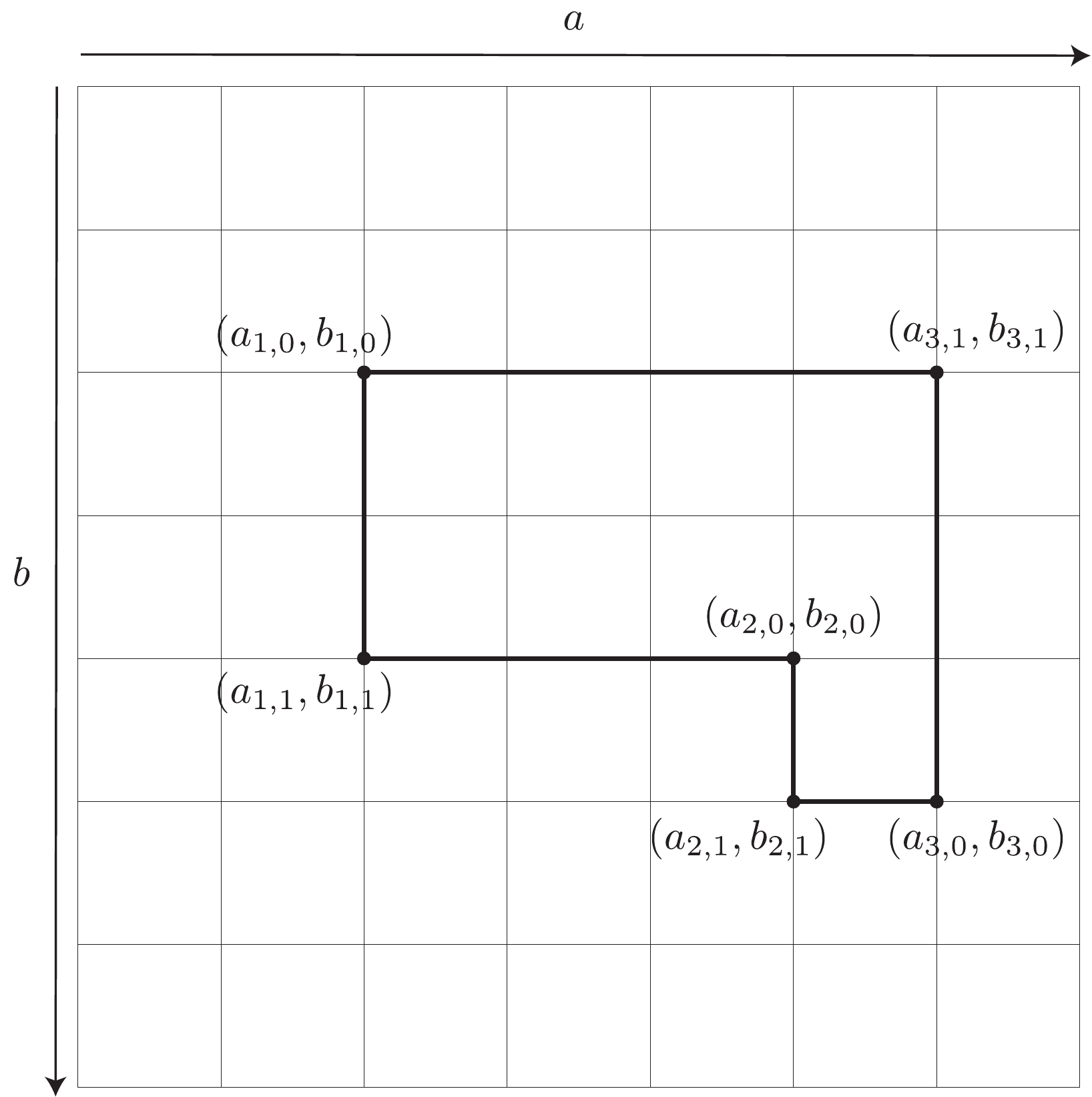}
\caption{\small A typical path in $[n] \times [n]$ that appears in the expansion of $\tr(A^* A)^j$, here with $j=3$.}
\label{fig:fig1}
\end{figure}

Second, write $\QT$ and $\QO$ in coefficients as
\[
\QT (e_{a'b'}) = \sum_{ab} c_{ab, a'b'} e_{ab}
\]
where $c_{ab,a'b'}$ is given by \eqref{cab}, and 
\[
\QO(e_{a'b'}) = \xi_{a'b'} e_{a'b'}, 
\]
where $\xi_{ab}$ are the iid, zero-expectation random variables
\[
\xi_{ab} := \frac{1}{p} 1_{(a,b) \in \Omega} - 1.
\]
With this, we have
\begin{equation}\label{aob}
A_{a_0,b_0} := \sum_{a_1,b_1,\ldots,a_k,b_k \in [n]} \bigl(\prod_{l \in
  [k]} c_{a_{l-1} b_{l-1}, a_l b_l}\bigr) \bigl(\prod_{l=0}^k \xi_{a_l b_l}\bigr)
E_{a_k b_k}
\end{equation}
for any $a_0,b_0 \in [n]$.  Note that this formula is even valid in the base case $k=0$, where it simplifies to just $A_{a_0 b_0} = \xi_{a_0 b_0} E_{a_0 b_0}$ due to our conventions on trivial sums and empty products.

\emph{Example.}  
If $k=2$, then
$$ A_{a_0,b_0} = \sum_{a_1,a_2,b_1,b_2 \in [n]} \xi_{a_0 b_0} c_{a_0 b_0,a_1,b_1} \xi_{a_1 b_1} c_{a_1 b_1,a_2 b_2} \xi_{a_2 b_2} E_{a_2 b_2}.$$

\emph{Remark.}  One can view the right-hand side of \eqref{aob} as the
sum over paths of length $k+1$ in $[n] \times [n]$ starting at the
designated point $(a_0,b_0)$ and ending at some arbitrary point
$(a_k,b_k)$.  Each edge (from $(a_i,b_i)$ to $(a_{i+1},b_{i+1})$) may
be a horizontal or vertical ``rook move'' (in that at least one of the
$a$ or $b$ coordinates does not change\footnote{Unlike the ordinary
  rules of chess, we will consider the trivial move when $a_{i+1}=a_i$
  \emph{and} $b_{i+1}=b_i$ to also qualify as a ``rook move'', which
  is simultaneously a horizontal and a vertical rook move.}), or a
``non-rook move'' in which both the $a$ and $b$ coordinates change.
It will be important later on to keep track of which edges are rook
moves and which ones are not, basically because of the presence of the
delta functions $1_{a=a'}, 1_{b=b'}$ in \eqref{cab}. Each edge in this
path is weighted by a $c$ factor, and each vertex in the path is
weighted by a $\xi$ factor, with the final vertex also weighted by an
additional $E$ factor.  It is important to note that the path is
allowed to cross itself, in which case weights such as $\xi^2$,
$\xi^3$, etc.~may appear, see Figure \ref{fig:fig2}.
\begin{figure}
 \centering
\includegraphics[scale=.75]{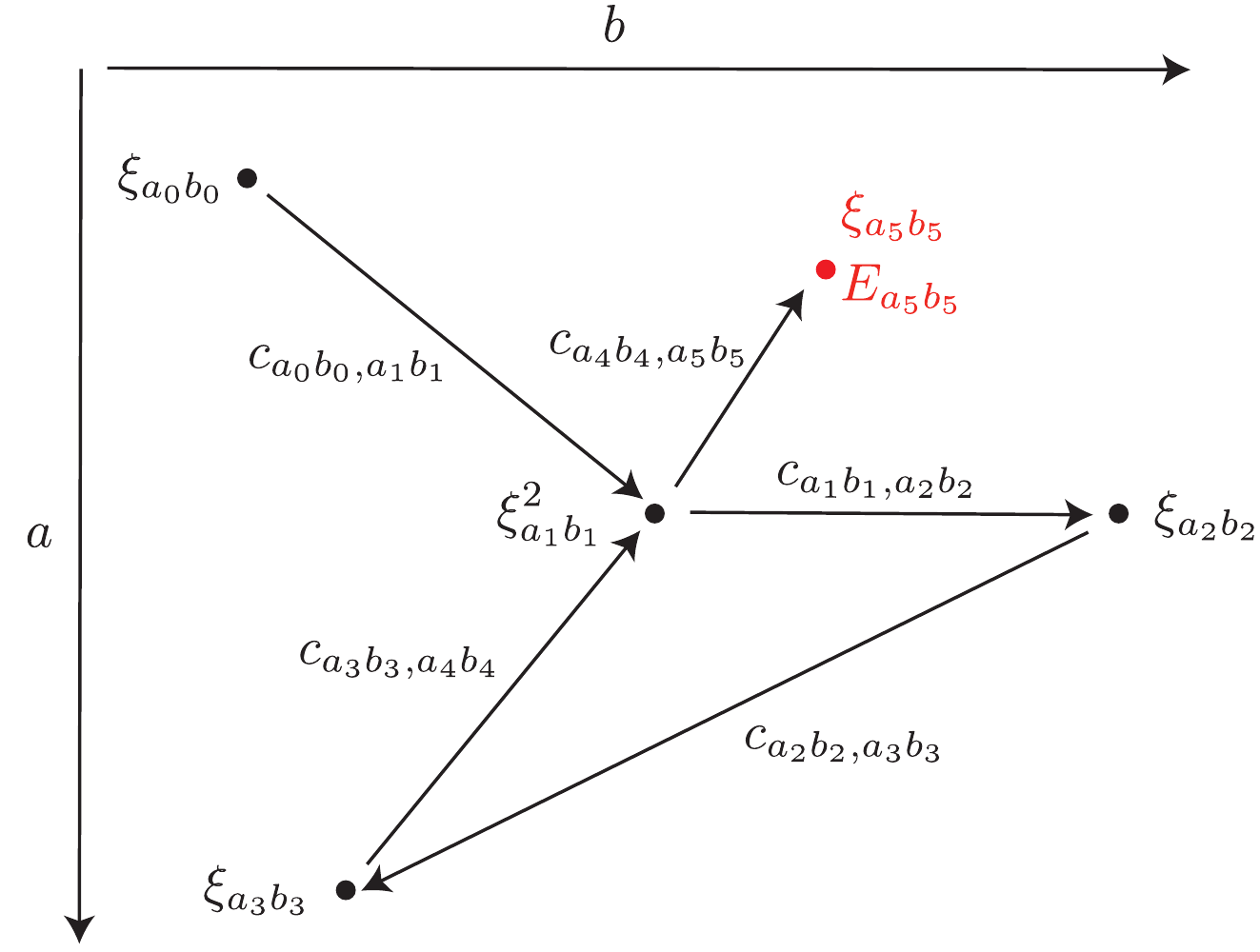}
\caption{\small A typical path appearing in the expansion \eqref{aob} of $A_{a_0 b_0}$, here with $k=5$.  Each vertex of the path gives rise to a $\xi$ factor (with the final vertex, coloured in red, providing an additional $E$ factor), while each edge of the path provides a $c$ factor.  Note that the path is certainly allowed to cross itself (leading to the $\xi$ factors being raised to powers greater than $1$, as is for instance the case here at $(a_1,b_1)=(a_4,b_4)$), and that the edges of the path may be horizontal, vertical, or neither.}
\label{fig:fig2}
\end{figure}
  
Inserting \eqref{aob} into \eqref{sumpath}, we see that $X$ can thus
be expanded as
\begin{equation}\label{X-expand}
\E \sum_* \prod_{i \in [j]} \prod_{\mu=0}^1 \Bigl[ \bigl(\prod_{l \in
  [k]} c_{a_{i,\mu,l-1} b_{i,\mu,l-1}, a_{i,\mu,l} b_{i,\mu,l}}\bigr)
\bigl(\prod_{l=0}^k \xi_{a_{i,\mu,l} b_{i,\mu,l}}\bigr) E_{a_{i,\mu,k}
  b_{i,\mu,k}} \Bigr],
\end{equation}
where the sum $\sum_*$ is over all combinations of $a_{i,\mu,l},
b_{i,\mu,l} \in [n]$ for $i \in [j]$, $\mu \in \{0,1\}$ and $0 \leq l
\leq k$ obeying the compatibility conditions
\begin{equation}\label{compat}
a_{i,1,0} = a_{i+1,0,0}; \quad b_{i,1,0} = b_{i,0,0} \hbox{ for all } i \in [j]
\end{equation}
with the cyclic convention $a_{j+1,0,0}=a_{1,0,0}$. 

\emph{Example.} Continuing our running example $j=k=2$, we have
$$
X = \E \sum_* \prod_{i=1}^2 \prod_{\mu=0}^1 \xi_{a_{i,\mu,0} b_{i,\mu,0}} c_{a_{i,\mu,0} b_{i,\mu,0}, a_{i,\mu,1} b_{i,\mu,1}} \xi_{a_{i,\mu,1} b_{i,\mu,1}} c_{a_{i,\mu,1} b_{i,\mu,1}, a_{i,\mu,2} b_{i,\mu,2}} \xi_{a_{i,\mu,2} b_{i,\mu,2}} E_{a_{i,\mu,2} b_{i,\mu,2}}
$$ 
where $a_{i,\mu,l}$ for $i=1,2$, $\mu=0,1$, $l=0,1,2$ obey the compatibility conditions
$$ a_{1,1,0} = a_{2,0,0}; \quad a_{2,1,0} = a_{1,0,0}; \quad b_{1,1,0} = b_{1,0,0}; \quad b_{2,1,0} = b_{2,0,0}.$$
Note that despite the small values of $j$ and $k$, this is already a rather complicated sum, ranging over $n^{2j (2k+1)} = n^{20}$ summands, each of which is the product of $4j(k+1)=24$ terms.

\emph{Remark.} The expansion \eqref{X-expand} is the sum over a sort
of combinatorial ``spider'', whose ``body'' is a closed path of length
$2j$ in $[n] \times [n]$ of alternating horizontal and vertical rook
moves, and whose $2j$ ``legs'' are paths of length $k$, emanating out
of each vertex of the body.  The various ``segments'' of the legs
(which can be either rook or non-rook moves) acquire a weight of $c$,
and the ``joints'' of the legs acquire a weight of $\xi$, with an
additional weight of $E$ at the tip of each leg.  To complicate things
further, it is certainly possible for a vertex of one leg to overlap
with another vertex from either the same leg or a different leg,
introducing weights such as $\xi^2$, $\xi^3$, etc.; see Figure
\ref{fig:fig3}. 
\begin{figure}
 \centering
\includegraphics[scale=.6]{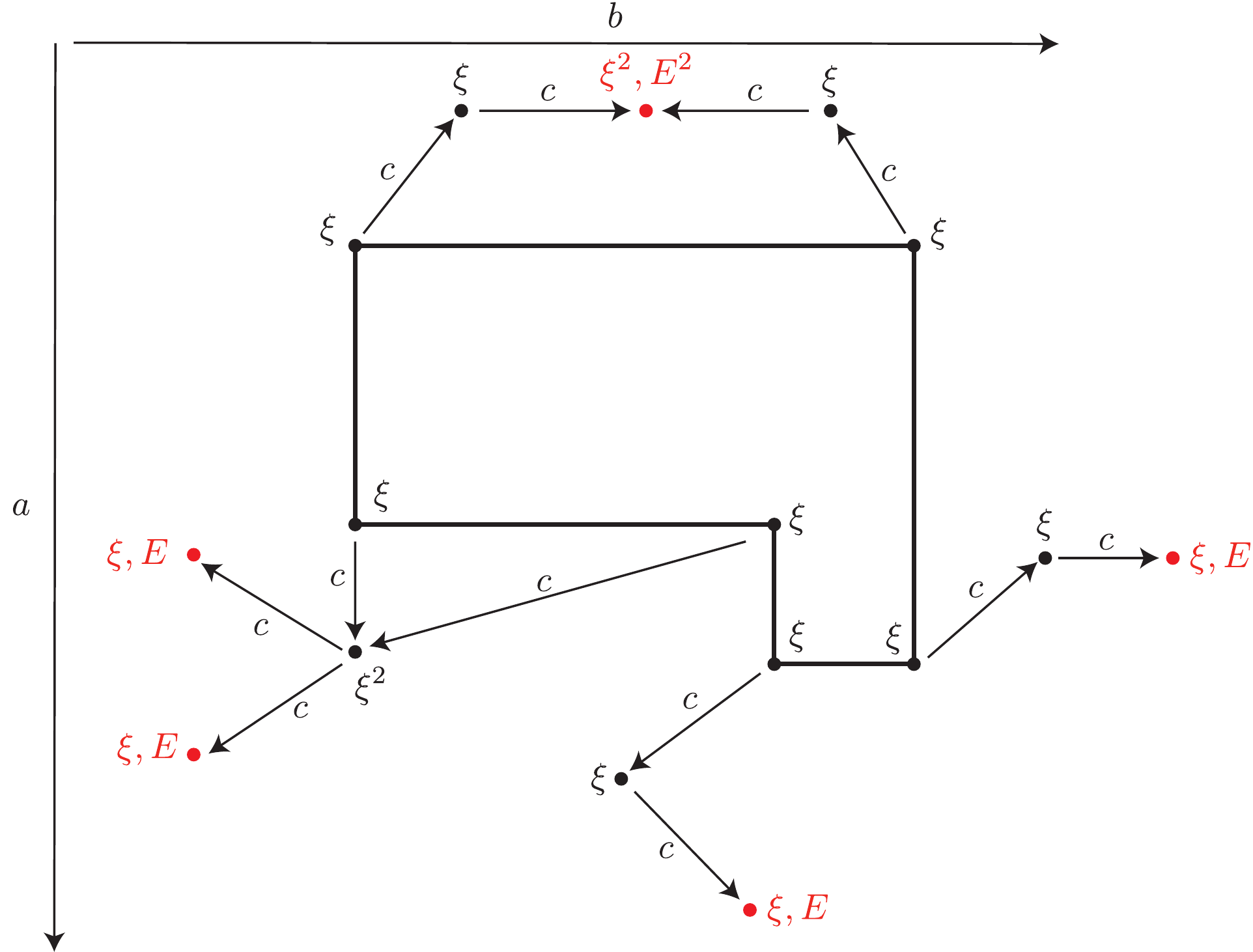}
\caption{\small A ``spider'' with $j=3$ and $k=2$, with the ``body'' in boldface lines and the ``legs'' as directed paths from the body to the tips (marked in red).}
\label{fig:fig3}
\end{figure}
 As one can see, the set of possible configurations
that this ``spider'' can be in is rather large and complicated.

\subsection{Second step: collecting rows and columns}\label{secondsec}

We now group the terms in the expansion \eqref{X-expand} into a bounded number of components, depending on how the various horizontal coordinates $a_{i,\mu,l}$ and vertical coordinates $b_{i,\mu,l}$ overlap.

It is convenient to order the $2j(k+1)$ tuples $(i,\mu,l) \in [j] \times \{0,1\} \times
\{0,\ldots,k\}$ lexicographically by declaring $(i,\mu,l) <
(i',\mu',l')$ if $i < i'$, or if $i=i'$ and $\mu < \mu'$, or if $i=i'$
and $\mu=\mu'$ and $l < l'$.

We then define the indices $s_{i,\mu,l}, t_{i,\mu,l} \in
\{1,2,3,\ldots\}$ recursively for all $(i,\mu,l) \in [j] \times
\{0,1\} \times [k]$ by setting $s_{1,0,0} = 1$ and declaring
$s_{i,\mu,l} := s_{i',\mu',l'}$ if there exists $(i',\mu',l') <
(i,\mu,l)$ with $a_{i',\mu',l'} = a_{i,\mu,l}$, or equal to the first
positive integer not equal to any of the $s_{i',\mu',l'}$ for
$(i',\mu',l') < (i,\mu,l)$ otherwise.  Define $t_{i,\mu,l}$ using
$b_{i,\mu,l}$ similarly.  We observe the \emph{cyclic
condition}
\begin{equation}\label{compat2}
s_{i,1,0} = s_{i+1,0,0}; \quad t_{i,1,0} = t_{i,0,0} \hbox{ for all } i \in [j]
\end{equation}
with the cyclic convention $s_{j+1,0,0} = s_{1,0,0}$.

\emph{Example.}  Suppose that $j=2$, $k=1$, and $n \geq 30$, with the $(a_{i,\mu,l},b_{i,\mu,l})$ given in lexicographical ordering as
\begin{align*}
(a_{0,0,0},b_{0,0,0}) &= (17, 30) \\
(a_{0,0,1},b_{0,0,1}) &= (13, 27) \\
(a_{0,1,0},b_{0,1,0}) &= (28, 30) \\
(a_{0,1,1},b_{0,1,1}) &= (13, 25) \\
(a_{1,0,0},b_{1,0,0}) &= (28, 11) \\
(a_{1,0,1},b_{1,0,1}) &= (17, 27) \\
(a_{1,1,0},b_{1,1,0}) &= (17, 11) \\
(a_{1,1,1},b_{1,1,1}) &= (13, 27) 
\end{align*}
Then we would have
\begin{align*}
(s_{0,0,0},t_{0,0,0}) &= (1, 1) \\
(s_{0,0,1},t_{0,0,1}) &= (2, 2) \\
(s_{0,1,0},t_{0,1,0}) &= (3, 1) \\
(s_{0,1,1},t_{0,1,1}) &= (2, 3) \\
(s_{1,0,0},t_{1,0,0}) &= (3, 4) \\
(s_{1,0,1},t_{1,0,1}) &= (1, 2) \\
(s_{1,1,0},t_{1,1,0}) &= (1, 4) \\
(s_{1,1,1},t_{1,1,1}) &= (2, 2). 
\end{align*}
Observe that the conditions \eqref{compat} hold for this example, which then forces \eqref{compat2} to hold also.

In addition to the property \eqref{compat2}, we see from construction of $(s,t)$ that for any $(i,\mu,l) \in [j] \times \{0,1\} \times \{0,\ldots,k\}$, the sets 
\begin{equation}\label{segment}
\{ s(i',\mu',l'): (i',\mu',l') \leq (i,\mu,l) \}, \{ t(i',\mu',l'): (i',\mu',l') \leq (i,\mu,l) \}
\end{equation}
are initial segments, i.e. of the form $[m]$ for some integer $m$. Let us call pairs $(s,t)$ of sequences with this property, as well as the property \eqref{compat2}, \emph{admissible}; thus for instance the sequences in the above example are admissible.  Given an admissible pair $(s,t)$, if we define the sets $J$, $K$ by
\begin{equation}\label{J-bang}
\begin{array}{lll}
J & := & \{ s_{i,\mu,l}: (i,\mu,l) \in [j] \times \{0,1\} \times \{0,\ldots,k\} \}\\
K & := & \{ t_{i,\mu,l}: (i,\mu,l) \in [j] \times \{0,1\} \times \{0,\ldots,k\}\}
\end{array}
\end{equation}
then we observe that $J = [|J|], K = [|K|]$.
Also, if $(s,t)$ arose from $a_{i,\mu,l}, b_{i,\mu,l}$ in the above
manner, there exist unique injections $\alpha: J \to [n], \beta: K \to
[n]$ such that $a_{i,\mu,l} = \alpha(s_{i,\mu,l})$ and $b_{i,\mu,l} =
\beta(t_{i,\mu,l})$.  

\emph{Example.}  Continuing the previous example, we have $J = [3]$, $K = [4]$, with the injections $\alpha: [3] \to [n]$ and $\beta: [4] \to [n]$ defined by 
$$\alpha(1) := 17; \alpha(2) := 13; \alpha(3) := 28$$
and
$$\beta(1) := 30; \beta(2) := 27; \beta(3) := 25; \beta(4):= 11.$$

Conversely, any admissible pair $(s,t)$ and
injections $\alpha,\beta$ determine $a_{i,\mu,l}$ and $b_{i,\mu,l}$.
Because of this, we can thus expand $X$ as
\begin{multline*}
  X =  \sum_{(s,t)} \E \sum_{\alpha,\beta} \prod_{i \in [j]} \prod_{\mu=0}^1 \Bigl[ \bigl(\prod_{l \in [k]} c_{\alpha(s_{i,\mu,l-1}) \beta(t_{i,\mu,l-1}), \alpha(s_{i,\mu,l}) \beta(t_{i,\mu,l})}\bigr) \\
  \bigl(\prod_{l=0}^L \xi_{\alpha(s_{i,\mu,l})
    \beta(t_{i,\mu,l})}\bigr) E_{\alpha(s_{i,\mu,k})
    \beta(t_{i,\mu,k})} \Bigr],
\end{multline*}
where the outer sum is over all admissible pairs $(s,t)$, and the
inner sum is over all injections. 

\emph{Remark.} As with preceding identities, the above formula is
also valid when $k = 0$ (with our conventions on trivial sums and empty products), in which case it simplifies to
\[
X = \sum_{(s,t)} \E \sum_{\alpha,\beta} \prod_{i \in [j]}
\prod_{\mu=0}^1 \xi_{\alpha(s_{i,\mu,0}) \beta(t_{i,\mu,0})}
E_{\alpha(s_{i,\mu,0}) \beta(t_{i,\mu,0})}.
\]

\emph{Remark.}  One can think of $(s,t)$ as describing the combinatorial ``configuration'' of the ``spider'' $( (a_{i,\mu,l}, b_{i,\mu,l}) )_{(i,\mu,l) \in [j] \times \{0,1\} \times \{0,\ldots,k\}}$ - it determines which vertices of the spider are equal to, or on the same row or column as, other vertices of the spider.  The injections $\alpha, \beta$ then enumerate the ways in which such a configuration can be ``represented'' inside the grid $[n] \times [n]$.

\subsection{Third step: computing the expectation}

The expansion we have for $X$ looks quite complicated.  However, the fact that the $\xi_{ab}$ are independent and have mean zero allows us to simplify this expansion to a significant degree.  Indeed, observe that the random variable $\Xi := \prod_{i \in [j]}
\prod_{\mu=0}^1 \prod_{l=0}^L \xi_{\alpha(s_{i,\mu,l})
  \beta(t_{i,\mu,l})}$ has zero expectation if there is any pair in $J
\times K$ which can be expressed exactly once in the form
$(s_{i,\mu,l},t_{i,\mu,l})$.  Thus we may assume that no pair can be
expressed exactly once in this manner.  If $\delta$ is a Bernoulli
variable with $\P(\delta = 1) = p = 1 - \P(\delta = 0)$, then for each
$s \ge 0$, one easily computes
\[
\E (\delta - p)^{s} = p(1-p)\left[(1-p)^{s-1} + (-1)^s p^{s-1}\right] 
\]
and hence 
\[
|\E (\frac{1}{p} \delta - 1)^{s}| \leq p^{1-s}.
\]
The value of the expectation of $\E \Xi$ does not depend on the choice
of $\alpha$ or $\beta$, and the calculation above shows that $\Xi$
obeys
\[
| \E \Xi| \le \frac{1}{p^{2j(k+1) - |\Omega|}},
\]
where
\begin{equation}\label{omega-def}
  \Omega := \{ (s_{i,\mu,l}, t_{i,\mu,l}): (i,\mu,l) \in [j] \times \{0,1\} \times \{0,\ldots,k\} \} \subset J \times K.
\end{equation}
Applying this estimate and the triangle inequality, we can thus bound $X$ by
\begin{multline}
\label{eq:formula}
  X \le \sum_{(s,t) \text{ strongly admissible}} (1/p)^{2j(k+1) - |\Omega|} \\
  \biggl|\sum_{\alpha,\beta} \prod_{i \in [j]} \prod_{\mu=0}^1 \Bigl[ \bigl(\prod_{l
    \in [k]} c_{\alpha(s_{i,\mu,l-1}) \beta(t_{i,\mu,l-1}),
    \alpha(s_{i,\mu,l}) \beta(t_{i,\mu,l})}\bigr) E_{\alpha(s_{i,\mu,k})
    \beta(t_{i,\mu,k})} \Bigr]\biggr|, 
\end{multline}
where the sum is over those admissible $(s,t)$ such that each element of $\Omega$
is visited at least twice by the sequence $(s_{i,\mu,l},
t_{i,\mu,l})$; we shall call such $(s,t)$ \emph{strongly
  admissible}. We will use the bound \eqref{eq:formula} as a starting point for
proving the moment estimates \eqref{eq:moment1} and
\eqref{eq:moment2}. 

\emph{Example.}  The pair $(s,t)$ in the Example in Section
\ref{secondsec} is admissible but not strongly admissible, because not
every element of the set $\Omega$ (which, in this example, is $\{
(1,1), (2,2), (3,1), $ $(2,3), (3,4),$ $(1,2), (1,4) \}$) is visited
twice by the $(s,t)$.

\emph{Remark.}  Once again, the formula \eqref{eq:formula} is valid
when $k=0$, with the usual conventions on empty products (in
particular, the factor involving the $c$ coefficients can be deleted
in this case).

\section{Quadratic bound in the rank}
\label{sec:moment1}

This section establishes \eqref{eq:moment1} under the assumptions of
Theorem \ref{teo:main1}, which is the easier of the two moment
estimates.  Here we shall just take the absolute values in
\eqref{eq:formula} inside the summation and use the estimates on the
coefficients given to us by hypothesis.  Indeed, starting with
\eqref{eq:formula} and the triangle inequality and applying
\eqref{eab} together with \eqref{cab} gives
\[
X \le O(1)^{j(k+1)} \sum_{(s,t) \text{ strongly admissible}}
(1/p)^{2j(k+1) - |\Omega|} \sum_{\alpha,\beta} (\sqrt{\rmu}/n)^{2jk +
  |Q| + 2j},
\]
where we recall that $\rmu = \mu^2 r$, and $Q$ is the set of all
$(i,\mu,l) \in [j] \times \{0,1\} \times [k]$ such that $s_{i,\mu,l-1}
\neq s_{i,\mu,l}$ and $t_{i,\mu,l-1} \neq t_{i,\mu,l}$. Thinking of
the sequence $\{(s_{i,\mu,l}, t_{i,\mu,l})\}$ as a path in $J \times
K$, we have that $(i,\mu,l) \in Q$ if and only if the move from
$(s_{i,\mu,l-1},t_{i,\mu,l-1})$ to $(s_{i,\mu,l},t_{i,\mu,l})$ is
neither horizontal nor vertical; per our earlier discussion, this is a
``non-rook'' move.

\emph{Example.}  The example in Section \ref{secondsec} is admissible, but not strongly admissible.  Nevertheless, the above definitions can still be applied, and we see that $Q = \{ (0,0,1), (0,1,1), (1,0,1), (1,1,1) \}$ in this case, because all of the four associated moves are non-rook moves.

As the number of injections $\alpha,
\beta$ is at most $n^{|J|}, n^{|K|}$ respectively, we thus have
\[
X \leq O(1)^{j(k+1)} \sum_{(s,t) \text{ str.~admiss.}} (1/p)^{2j(k+1)
  - |\Omega|} n^{|J|+|K|} (\sqrt{\rmu}/n)^{2jk + |Q| + 2j},
\]
which we rearrange slightly as
\[
X \leq O(1)^{j(k+1)} \sum_{(s,t) \text{ str.~admiss.}}
\Bigl(\frac{\rmu^2}{np}\Bigr)^{2j(k+1) - |\Omega|} \rmu^{\frac{|Q|}{2} + 2|\Omega| -
  3j(k+1)} n^{|J|+|K|-|Q|-|\Omega|}.
\]
Since $(s,t)$ is strongly admissible and every point in $\Omega$ needs to be
visited at least twice, we see that
\[ 
|\Omega| \leq j(k+1).
\]
Also, since $Q \subset [j] \times \{0,1\} \times [k]$, we have the
trivial bound
\begin{equation*}\label{triv}
  |Q| \leq  2jk.
\end{equation*}
This ensures that 
$$ \frac{|Q|}{2} + 2|\Omega| - 3j(k+1) \leq 0$$
and
$$ 2j(k+1) - |\Omega| \geq j(k+1).$$
From the hypotheses of Theorem \ref{teo:main1} we have $np \geq \rmu^2$, and thus 
\begin{equation*}
\label{eq:51}
X \leq O\Bigl(\frac{\rmu^2}{np}\Bigr)^{j(k+1)} \sum_{(s,t) \text{
    str.~admiss.}}  n^{|J|+|K|-|Q|-|\Omega|}.
\end{equation*}

\emph{Remark.} In the case where $k = 0$ in which $Q = \emptyset$, one can easily
obtain a better estimate, namely, (if $np \ge \rmu$)
\[
X \leq O\Bigl(\frac{\rmu}{np}\Bigr)^{j} \sum_{(s,t) \text{
    str.~admiss.}}  n^{|J|+|K|-|\Omega|}.
\]

Call a triple $(i,\mu,l)$ \emph{recycled} if we have $s_{i',\mu',l'} =
s_{i,\mu,l}$ or $t_{i',\mu',l'} = t_{i,\mu,l}$ for some $(i',\mu',l')
< (i,\mu,l)$, and \emph{totally recycled} if $(s_{i',\mu',l'},
t_{i',\mu',l'}) = (s_{i,\mu,l}, t_{i,\mu,l})$ for some $(i',\mu',l') < (i,\mu,l)$.  Let $Q'$ denote the set
of all $(i,\mu,l) \in Q$ which are recycled.

\emph{Example.}  The example in Section \ref{secondsec} is admissible, but not strongly admissible.  Nevertheless, the above definitions can still be applied, and we see that the triples 
$$(0,1,0), (0,1,1), (1,0,0), (1,0,1), (1,1,0), (1,1,1)$$ 
are all recycled (because they either reuse an existing value of $s$ or $t$ or both), while the triple $(1,1,1)$ is totally recycled (it visits the same location as the earlier triple $(0,0,1)$).  Thus in this case, we have $Q' = \{ (0,1,1), (1,0,1), (1,1,1) \}$.  

We observe that if $(i,\mu,l) \in [j] \times \{0,1\} \times [k]$ is not recycled, then it must have been reached from $(i,\mu,l-1)$ by a non-rook move, and thus $(i,\mu,l)$ lies in $Q$.

\begin{lemma}[Exponent bound]
\label{teo:Q}  For any admissible tuple, we have $|J|+|K|-|Q|-|\Omega| \leq -|Q'|+1$.
\end{lemma}

\begin{proof} We let $(i,\mu,l)$ increase from $(1,0,0)$ to $(j,1,k)$
  and see how each $(i,\mu,l)$ influences the quantity $|J|+|K|-|Q
  \backslash Q'|-|\Omega|$.

  Firstly, we see that the triple $(1,0,0)$ initialises $|J|, |K|,
  |\Omega| = 1$ and $|Q \backslash Q'| = 0$, so $|J|+|K|-|Q \backslash
  Q'|-|\Omega| = 1$ at this initial stage.  Now we see how each subsequent
  $(i,\mu,l)$ adjusts this quantity.

  If $(i,\mu,l)$ is totally recycled, then $J, K, \Omega, Q \backslash
  Q'$ are unchanged by the addition of $(i,\mu,l)$, and so $|J|+|K|-|Q
  \backslash Q'|-|\Omega|$ does not change.

  If $(i,\mu,l)$ is recycled but not totally recycled, then one of $J,
  K$ increases in size by at most one, as does $\Omega$, but the other
  set of $J, K$ remains unchanged, as does $Q \backslash Q'$, and so
  $|J|+|K|-|Q \backslash Q'|-|\Omega|$ does not increase.

  If $(i,\mu,l)$ is not recycled at all, then (by \eqref{compat2}) we
  must have $l>0$, and then (by definition of $Q, Q'$) we have
  $(i,\mu,l) \in Q \backslash Q'$, and so $|Q \backslash Q'|$ and
  $|\Omega|$ both increase by one.  Meanwhile, $|J|$ and $|K|$
  increase by 1, and so $|J|+|K|-|Q \backslash Q'|-|\Omega|$
  does not change.  Putting all this together we obtain the claim.
\end{proof}

This lemma gives 
\begin{equation*}
\label{eq:52}
X \leq O\Bigl(\frac{\rmu^2}{np}\Bigr)^{j(k+1)} \sum_{ \text{
    str.~admiss.}}  n^{-|Q'|+1}.
\end{equation*}

\emph{Remark.} When $k = 0$, we have the better bound
\[
X \leq O\Bigl(\frac{\rmu}{np}\Bigr)^{j} \sum_{ \text{
    str.~admiss.}}  n.
\]

To estimate the above sum, we need to count strongly admissible pairs.  This is achieved by the following lemma.

\begin{lemma}[Pair counting]\label{paircount} For fixed $q \geq 0$, the number of strongly admissible
  pairs $(s,t)$ with $|Q'|=q$ is at most $O(j(k+1))^{2j(k+1) + q}$.
\end{lemma}

\begin{proof} Firstly observe that once one fixes $q$, the number of
  possible choices for $Q'$ is $\binom{2jk}{q}$, which we can bound
  crudely by $2^{2j(k+1)} = O(1)^{2j(k+1)+q}$.  So we may without loss of generality
  assume that $Q'$ is fixed.  For similar reasons we may assume $Q$ is
  fixed.

  As with the proof of Lemma \ref{teo:Q}, we increment $(i,\mu,l)$ from $(1,0,0)$ to $(j,1,k)$ and upper
  bound how many choices we have available for $s_{i,\mu,l},
  t_{i,\mu,l}$ at each stage.

  There are no choices available for $s_{1,0,0}, t_{1,0,0}$, which
  must both be one.  Now suppose that $(i,\mu,l) > (1,0,0)$.  There are
  several cases.

  If $l=0$, then by \eqref{compat2} one of $s_{i,\mu,l}, t_{i,\mu,l}$
  has no choices available to it, while the other has at most $O(
  j(k+1) )$ choices.  If $l > 0$ and $(i,\mu,l) \not \in Q$, then at
  least one of $s_{i,\mu,l}, t_{i,\mu,l}$ is necessarily equal to its
  predecessor; there are at most two choices available for which index
  is equal in this fashion, and then there are $O( j(k+1) )$ choices
  for the other index.

  If $l > 0$ and $(i,\mu,l) \in Q \backslash Q'$, then both
  $s_{i,\mu,l}$ and $t_{i,\mu,l}$ are new, and are thus equal to the
  first positive integer not already occupied by $s_{i',\mu',l'}$ or
  $t_{i',\mu',l'}$ respectively for $(i',\mu',l') < (i,\mu,l)$.  So
  there is only one choice available in this case.

  Finally, if $(i,\mu,l) \in Q'$, then there can be $O(j(k+1))$
  choices for both $s_{i,\mu,l}$ and $t_{i,\mu,l}$.

  Multiplying together all these bounds, we obtain that the number of
  strongly admissible pairs is bounded by 
\[
O(j(k+1))^{2j + 2jk - |Q| + 2|Q'|} = O(j(k+1))^{2j(k+1) - |Q\setminus
  Q'| + |Q'|},
\]
which proves the claim (here we discard the $|Q \setminus Q'|$ factor).
\end{proof}

Using the above lemma we obtain
\[
X \leq O(1)^{j(k+1)} n \left(\frac{\rmu^2}{np}\right)^{j(k+1)}
\sum_{q=0}^{2jk} O(j(k+1))^{2j(k+1) + q} n^{-q}.
\] 
Under the assumption $n \ge c_0 j(k+1)$ for some numerical constant
$c_0$, we can sum the series and obtain Theorem
\ref{teo:moment1}. 

\emph{Remark.} When $k = 0$, we have the better bound
\[
X \leq O(j)^{2j} n \left(\frac{\rmu}{np}\right)^{j}.
\] 

\section{Linear bound in the rank}
\label{sec:moment2}

We now prove the more sophisticated moment estimate \eqref{eq:moment2} under the hypotheses of Theorem \ref{teo:main2}.  Here, we cannot afford to take absolute values immediately, as in the proof of \eqref{eq:moment1}, but first must exploit some algebraic cancellation properties in the coefficients $c_{ab,a'b'}$, $E_{ab}$ appearing in \eqref{eq:formula} to simplify the sum.

\subsection{Cancellation identities}

Recall from \eqref{cab} that the coefficients $c_{ab,a'b'}$ are
defined in terms of the coefficients $\U_{a,a'}$, $\V_{b,b'}$
introduced in \eqref{eq:UV}.  We recall the symmetries $\U_{a,a'} =
\U_{a',a}, \V_{b,b'} = \V_{b',b}$ and the projection identities
\begin{align}\label{proj-id}
  \sum_{a'} \U_{a,a'} \U_{a',a''} & = \left(1 - 2\rho\right) \U_{a,a''}  - 
\rho\left(1-\rho\right) 1_{a = a''},\\
  \sum_{b'} \V_{b,b'} \V_{b',b''} & = \left(1 - 2\rho\right) \V_{b,b''}  - 
\rho\left(1-\rho\right) 1_{b = b''};
\end{align}
the first identity follows from the matrix identity
$$\sum_{a'} \U_{a,a'} \U_{a',a''} = \<e_a, Q_U^2 e_{a'}\>$$
after one writes the projection identity $P_U^2 = P_U$ in terms of $Q_U$ using \eqref{prhoq}, and similarly for the second identity.

In a similar vein, we also have the
identities
\begin{equation}\label{proj-i2d}
  \sum_{a'} \U_{a,a'} E_{a',b} = \left(1-\rho\right) E_{a,b} =  
\sum_{b'} E_{a,b'} \V_{b',b},  
\end{equation}
which simply come from $Q_U E = P_U E - \rho E = (1-\rho)E$ together
with $E Q_V = E P_V - \rho E = (1-\rho)E$. Finally, we observe the two equalities
\begin{equation}\label{proj-i3d}
\sum_{b} E_{a,b} E_{a',b} = \U_{a,a'} + \rho 1_{a=a'}, \quad
\sum_{a} E_{a,b} E_{a,b'} = \V_{b,b'} + \rho 1_{b=b'}.
\end{equation}
The first identity follows from the fact that $\sum_{b} E_{a,b} E_{a',b}$ is the
$(a,a')^{th}$ element of $EE^* = P_U = Q_U + \rho I$, and the second one similarly follows from the identity $E^* E =
P_V = Q_V + \rho I$.

\subsection{Reduction to a summand bound}

Just as before, our goal is to estimate
\[
X := \E \tr(A^* A)^j, \quad  A = (\QO \QT)^k \QO E.
\]

We recall the bound \eqref{eq:formula}, and 
expand out each of the $c$ coefficients using \eqref{cab} into three
terms.  To describe the resulting expansion of the sum we need more
notation.  Define an \emph{admissible quadruplet} $(s,t,\L_U, \L_V)$
to be an admissible pair $(s,t)$, together with two sets $\L_U, \L_V$
with $\L_U \cup \L_V = [j] \times \{0,1\} \times [k]$, such that
$s_{i,\mu,l-1}=s_{i,\mu,l}$ whenever $(i,\mu,l) \in ([j] \times
\{0,1\} \times [k]) \backslash \L_U$, and $t_{i,\mu,l-1}=t_{i,\mu,l}$
whenever $(i,\mu,l) \in ([j] \times \{0,1\} \times [k]) \backslash
\L_V$.  If $(s,t)$ is also strongly admissible, we say that
$(s,t,\L_U,\L_V)$ is a \emph{strongly admissible quadruplet}.

The sets $\L_U \backslash \L_V$, $\L_V \backslash \L_U$, $\L_U \cap
\L_V$ will correspond to the three terms $1_{b = b'} \U_{a,a'}$, $1_{a
  = a'} \V_{b,b'}$, $\U_{a,a'} \V_{b,b'}$ appearing in
\eqref{cab}. With this notation, we expand the product
\[
\prod_{i \in [j]} \prod_{\mu=0}^1 \prod_{l \in [k]}
c_{\alpha(s_{i,\mu,l-1}) \beta(t_{i,\mu,l-1}), \alpha(s_{i,\mu,l})
  \beta(t_{i,\mu,l})}
\]
as 
\begin{multline*}
  \sum_{\L_U, \L_V} (1-\rho)^{|\L_U\backslash\L_V| +
    |\L_U\backslash\L_V|} (-1)^{|\L_U \cap \L_V|}
  \Bigl[\prod_{(i,\mu,l) \in \L_U\backslash \L_V}
  1_{\beta(t_{i,\mu,l-1}) = \beta(t_{i,\mu,l})}
  \U_{\alpha(s_{i,\mu,l-1}),\alpha(s_{i,\mu,l})}\Bigr] \\ \Bigl[
  \prod_{(i,\mu,l) \in \L_V\backslash \L_U}
  1_{\alpha(s_{i,\mu,l-1}),\alpha(s_{i,\mu,l})}
  \V_{\beta(t_{i,\mu,l-1}),\beta(t_{i,\mu,l})} \Bigr]
  \Bigl[\prod_{(i,\mu,l) \in \L_U\cap\L_V}
  U_{\alpha(s_{i,\mu,l-1}),\alpha(s_{i,\mu,l})}
  \V_{\beta(t_{i,\mu,l-1}),\beta(t_{i,\mu,l})} \Bigr],
\end{multline*}
where the sum is over all partitions as above, and which we can
rearrange as 
\[
\sum_{\L_U, \L_V} [-(1-\rho)]^{2j(k+1)-|\L_U \cap \L_V|}
\Bigl[\prod_{(i,\mu,l) \in \L_U}
\U_{\alpha(s_{i,\mu,l-1}),\alpha(s_{i,\mu,l})} \Bigr] \, \Bigl[
\prod_{(i,\mu,l) \in \L_V}
\V_{\beta(t_{i,\mu,l-1}),\beta(t_{i,\mu,l})} \Bigr].
\]
From this and the triangle inequality, we observe the bound
\[
X \le \left(1-\rho\right)^{2j(k+1) - |\L_U \cap \L_V|} \,
\sum_{(s,t,\L_U,\L_V)} (1/p)^{2j(k+1) - |\Omega|} |X_{s,t,\L_U,
  \L_V}|,
\] 
where the sum ranges over all strongly admissible quadruplets, and
\begin{multline*}
  X_{s,t,\L_U,\L_V} := \sum_{\alpha,\beta}
  \Bigl[\prod_{i \in [j]} \prod_{\mu=0}^1 E_{\alpha(s_{i,\mu,k}) \beta(t_{i,\mu,k})} \Bigr] \\
  \Bigl[\prod_{(i,\mu,l) \in \L_U}
  \U_{\alpha(s_{i,\mu,l-1}),\alpha(s_{i,\mu,l})} \Bigr] \, \Bigl[
  \prod_{(i,\mu,l) \in \L_V}
  \V_{\beta(t_{i,\mu,l-1}),\beta(t_{i,\mu,l})} \Bigr].
\end{multline*}

\emph{Remark.}  A strongly admissible quadruplet can be viewed as the
configuration of a ``spider'' with several additional constraints.
Firstly, the spider must visit each of its vertices at least twice
(strong admissibility).  When $(i,\mu,l) \in [j] \times \{0,1\} \times
[k]$ lies out of $\L_U$, then only horizontal rook moves are allowed
when reaching $(i,\mu,l)$ from $(i,\mu,l-1)$; similarly, when
$(i,\mu,l)$ lies out of $\L_V$, then only vertical rook moves are
allowed from $(i,\mu,l-1)$ to $(i,\mu,l)$.  In particular, non-rook
moves are only allowed inside $\L_U \cap \L_V$; in the notation of the
previous section, we have $Q \subset \L_U \cap \L_V$.  Note though
that while one has the \emph{right} to execute a non-rook move to
$\L_U \cap \L_V$, it is not mandatory; it could still be that
$(s_{i,\mu,l-1}, t_{i,\mu,l-1})$ shares a common row or column (or
even both) with $(s_{i,\mu,l},t_{i,\mu,l})$.

We claim the following fundamental bound on the summand
$|X_{s,t,\L_U,\L_V}|$:

\begin{proposition}[Summand bound]\label{main} Let $(s,t,\L_U,\L_V)$
  be a strongly admissible quadruplet.  Then we have
\[
|X_{s,t,\L_U,\L_V}| \le O(j(k+1))^{2j(k+1)} (r/n)^{2j(k+1) - |\Omega|}
n.
\]
\end{proposition}

Assuming this proposition, we have
\[
X \le  O(j(k+1))^{2j(k+1)} \sum_{(s,t,\L_U,\L_V)} (r/np)^{2j(k+1) - |\Omega|} n
\]
and since $|\Omega| \leq j(k+1)$ (by strong admissibility) and $r \leq
np$, and the number of $(s,t,\L_U,\L_V)$ can be crudely bounded by
$O(j(k+1))^{4j(k+1)}$,
$$ 
X \le  O(j(k+1))^{6j(k+1)} (r/np)^{j(k+1)} n.
$$
This gives \eqref{eq:moment2} as desired. The bound on the number of
quadruplets follows from the fact that there are at most
$j(k+1)^{4j(k+1)}$ strongly admissible pairs and that the number of
$(\L_U,\L_V)$ per pair is at most $O(1)^{j(k+1)}$.

{\em Remark.} It seems clear that the exponent $6$ can be lowered by a
finer analysis, for instance by using counting bounds such as Lemma
\ref{paircount}.  However, substantial effort seems to be required in
order to obtain the optimal exponent of $1$ here.

\subsection{Proof of Proposition \ref{main}}

To prove the proposition, it is convenient to generalise it by
allowing $k$ to depend on $i, \mu$.  More precisely, define a
\emph{configuration} ${\mathcal C} = (j, k, J, K, s, t, \L_U, \L_V)$
to be the following set of data:
\begin{itemize}
\item An integer $j \geq 1$, and a map $k: [j] \times \{0,1\} \to \{0,
  1, 2, \ldots \}$, generating a set $\Gamma := \{ (i,\mu,l): i \in
  [j], \mu \in \{0,1\}, 0 \leq l \leq k(i,\mu) \}$;
\item Finite sets $J, K$, and surjective maps $s: \Gamma \to J$ and
  $t: \Gamma \to K$ obeying \eqref{compat2};
\item Sets $\L_U, \L_V$ such that
$$\L_U \cup \L_V := \Gamma_+ := \{ (i,\mu,l) \in \Gamma:  l > 0 \}$$
and such that $s_{i,\mu,l-1}=s_{i,\mu,l}$ whenever $(i,\mu,l) \in
\Gamma_+ \backslash \L_U$, and $t_{i,\mu,l-1}=t_{i,\mu,l}$ whenever
$(i,\mu,l) \in \Gamma_+ \backslash \L_V$.
\end{itemize}

\emph{Remark.} Note we do not require configurations to be strongly
admissible, although for our application to Proposition \ref{main}
strong admissibility is required.  Similarly, we no longer require
that the segments \eqref{segment} be initial segments. This removal of
hypotheses will give us a convenient amount of flexibility in a
certain induction argument that we shall perform shortly. One can
think of a configuration as describing a ``generalized spider'' whose
legs are allowed to be of unequal length, but for which certain of the
segments (indicated by the sets $\L_U$, $\L_V$) are required to be
horizontal or vertical.  The freedom to extend or shorten the legs of
the spider separately will be of importance when we use the identities
\eqref{proj-id}, \eqref{proj-i2d}, \eqref{proj-i3d} to simplify the
expression $X_{s,t,\L_U,\L_V}$, see Figure \ref{fig:fig3p5}. 
\begin{figure}
 \centering
\includegraphics[scale=.6]{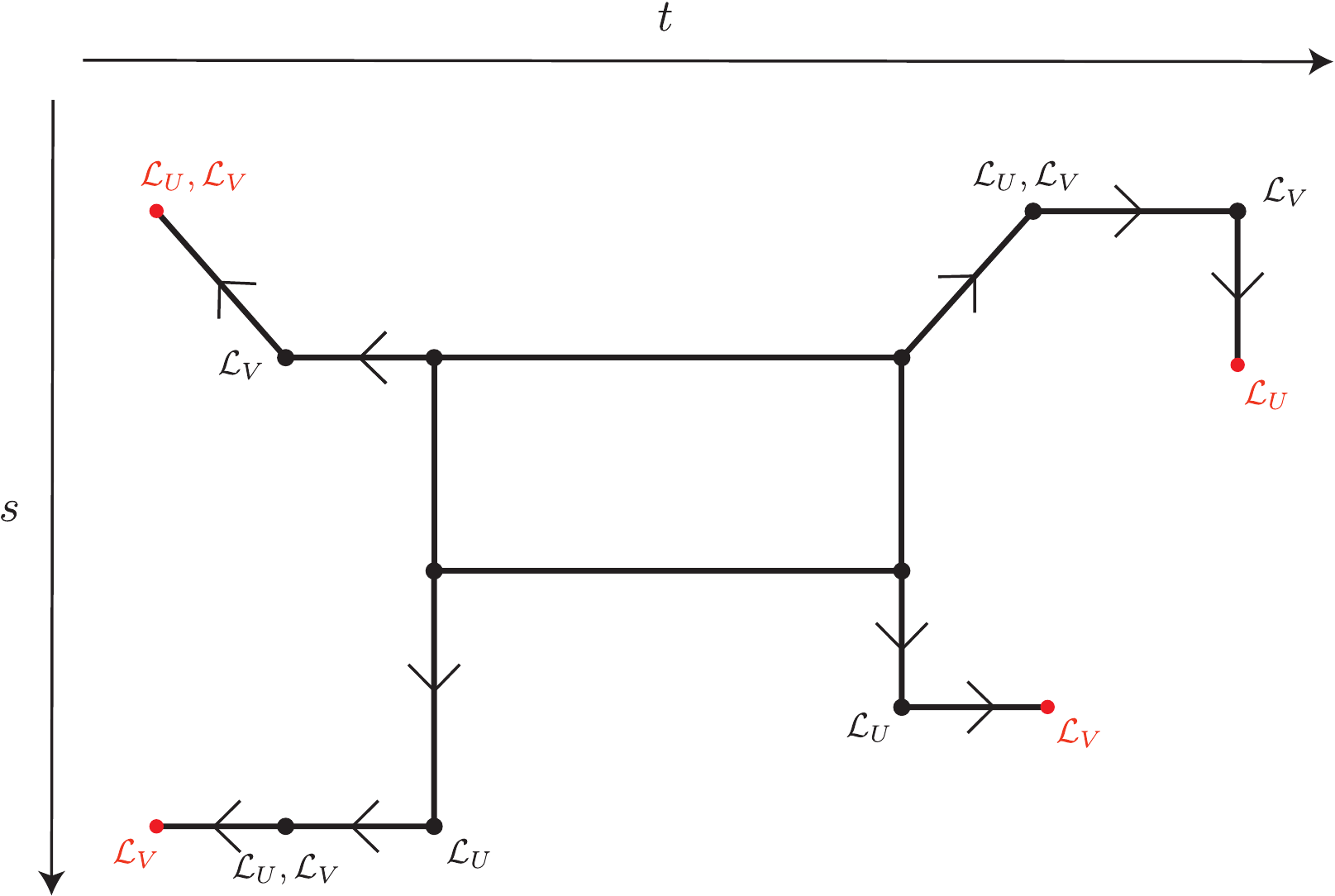}
\caption{\small A generalized spider (note the variable leg lengths).  A vertex labeled just by $\L_U$ must have been reached from its predecessor by a vertical rook move, while a vertex labeled just by $\L_V$ must have been reached by a horizontal rook move.  Vertices labeled by both $\L_U$ and $\L_V$ \emph{may} be reached from their predecessor by a non-rook move, but they are still allowed to lie on the same row or column as their predecessor, as is the case in the leg on the bottom left of this figure.  The sets $\L_U, \L_V$ indicate which $U$ and $V$ terms will show up in the expansion \eqref{xcdef}.}
\label{fig:fig3p5}
\end{figure}

Given a configuration ${\mathcal C}$, define the quantity $X_{\mathcal C}$ by
the formula
\begin{multline}\label{xcdef}
  X_{\mathcal C} := \sum_{\alpha,\beta}
  \Bigl[\prod_{i \in [j]} \prod_{\mu=0}^1 E_{\alpha(s(i,\mu,k(i,\mu))) \beta(t(i,\mu,k(i,\mu)))} \Bigr] \\
  \Bigl[ \prod_{(i,\mu,l) \in \L_U}
    \U_{\alpha(s(i,\mu,l-1)),\alpha(s(i,\mu,l))} \Bigr]  \Bigl[
    \prod_{(i,\mu,l) \in \L_V}
    \V_{\beta(t(i,\mu,l-1)),\beta(t(i,\mu,l))} \Bigr], 
\end{multline}
where $\alpha: J \to [n], \beta: K \to [n]$ range over all injections.
To prove Proposition \ref{main}, it then suffices to show that
\begin{equation}\label{xc}
 |X_{\mathcal C}| \leq (C_0 (1+|J|+|K|))^{|J|+|K|} (\rmu/n)^{|\Gamma| - |\Omega|} n
\end{equation}
for some absolute constant $C_0 > 0$, where
$$ \Omega := \{ (s(i,\mu,l), t(i,\mu,l)): (i,\mu,l) \in \Gamma \},$$
since Proposition \ref{main} then follows from the special case in which $k(i,\mu)=k$ is constant and $(s,t)$ is strongly admissible, in which case we have
$$ |J|+|K| \leq 2|\Omega| \leq |\Gamma| = 2j(k+1)$$
(by strong admissibility).

To prove the claim \eqref{xc} we will perform strong induction on the
quantity $|J|+|K|$; thus we assume that the claim has already been
proven for all configurations with a strictly smaller value of
$|J|+|K|$.  (This inductive hypothesis can be vacuous for very small
values of $|J|+|K|$.)  Then, for fixed $|J|+|K|$, we perform strong
induction on $|\L_U \cap \L_V|$, assuming that the claim has already
been proven for all configurations with the same value of $|J|+|K|$
and a strictly smaller value of $|\L_U \cap \L_V|$.

\emph{Remark.}  Roughly speaking, the inductive hypothesis is
asserting that the target estimate \eqref{xc} has already been proven
for all generalized spider configurations which are ``simpler'' than
the current configuration, either by using fewer rows and columns, or
by using the same number of rows and columns but by having fewer
opportunities for non-rook moves.

As we shall shortly see, whenever we invoke the inner induction
hypothesis (decreasing $|\L_U \cap \L_V|$, keeping $|J|+|K|$ fixed) we
are replacing the expression $X_{\mathcal C}$ with another expression
$X_{{\mathcal C}'}$ covered by this hypothesis; this causes no
degradation in the constant.  But when we invoke the outer induction
hypothesis (decreasing $|J|+|K|$), we will be splitting up
$X_{\mathcal C}$ into about $O(1+ |J|+|K|)$ terms $X_{{\mathcal C}'}$,
each of which is covered by this hypothesis; this causes a degradation
of $O(1+ |J|+|K|)$ in the constants and is thus responsible for the
loss of $(C_0(1+|J|+|K|))^{|J|+|K|}$ in \eqref{xc}.

For future reference we observe that we may take $\rmu \leq n$, as the
hypotheses of Theorem \ref{teo:main1} are vacuous otherwise ($m$
cannot exceed $n^2$).

To prove \eqref{xc} we divide into several cases.

\subsubsection{First case: an unguarded non-rook move}

Suppose first that $\L_U \cap \L_V$ contains an element
$(i_0,\mu_0,l_0)$ with the property that
\begin{equation}\label{ng}
 (s_{i_0,\mu_0,l_0-1},t_{i_0,\mu_0,l_0}) \not \in \Omega.
\end{equation}
Note that this forces the edge from
$(s_{i_0,\mu_0,l_0-1},t_{i_0,\mu_0,l_0-1})$ to
$(s_{i_0,\mu_0,l_0},t_{i_0,\mu_0,l_0})$ to be partially ``unguarded''
in the sense that one of the opposite vertices of the rectangle that
this edge is inscribed in is not visited by the $(s,t)$ pair.

When we have such an unguarded non-rook move, we can ``erase'' the element $(i_0,\mu_0,l_0)$ from $\L_U \cap \L_V$ by replacing
${\mathcal C} = (j, k, J, K, s, t, \L_U, \L_V)$ by the ``stretched''
variant ${\mathcal C}' = (j', k', J', K', s',$ $t', \L'_U, \L'_V)$,
defined as follows:
\begin{itemize}
\item $j' := j$, $J' := J$, and $K' := K$.
\item $k'(i,\mu) := k(i,\mu)$ for $(i,\mu) \neq (i_0,\mu_0)$, and $k'(i_0,\mu_0) := k(i_0,\mu_0)+1$.
\item $(s'_{i,\mu,l},t'_{i,\mu,l}) := (s_{i,\mu,l},t_{i,\mu,l})$ whenever $(i,\mu) \neq (i_0,\mu_0)$, or when $(i,\mu) = (i_0,\mu_0)$ and $l < l_0$.
\item $(s'_{i,\mu,l},t'_{i,\mu,l}) := (s_{i,\mu,l-1},t_{i,\mu,l-1})$ whenever $(i,\mu) = (i_0,\mu_0)$ and $l > l_0$.
\item $(s'_{i_0,\mu_0,l_0},t'_{i_0,\mu_0,l_0}) := (s_{i_0,\mu_0,l_0-1},t_{i_0,\mu_0,l_0})$.
\item We have
\begin{align*}
\L'_U := \{(i,\mu,l) \in \L_U & : (i,\mu) \neq (i_0,\mu_0) \} \\
& \cup \{(i_0,\mu_0,l) \in \L_U: l < l_0 \} \\
& \cup \{ (i_0,\mu_0,l+1): (i_0,\mu_0,l) \in \L_U; l > l_0+1 \} \\
& \cup \{ (i_0,\mu_0,l_0+1) \}
\end{align*}
and
\begin{align*}
\L'_V := \{(i,\mu,l) \in \L_V & : (i,\mu) \neq (i_0,\mu_0) \} \\
& \cup \{(i_0,\mu_0,l) \in \L_V: l < l_0 \} \\
& \cup \{ (i_0,\mu_0,l+1): (i_0,\mu_0,l) \in \L_V; l > l_0+1 \} \\
& \cup \{ (i_0,\mu_0,l_0) \}.
\end{align*}
\end{itemize}
All of this is illustrated in Figure \ref{fig:fig4}.
\begin{figure}
 \centering
\includegraphics[scale=.7]{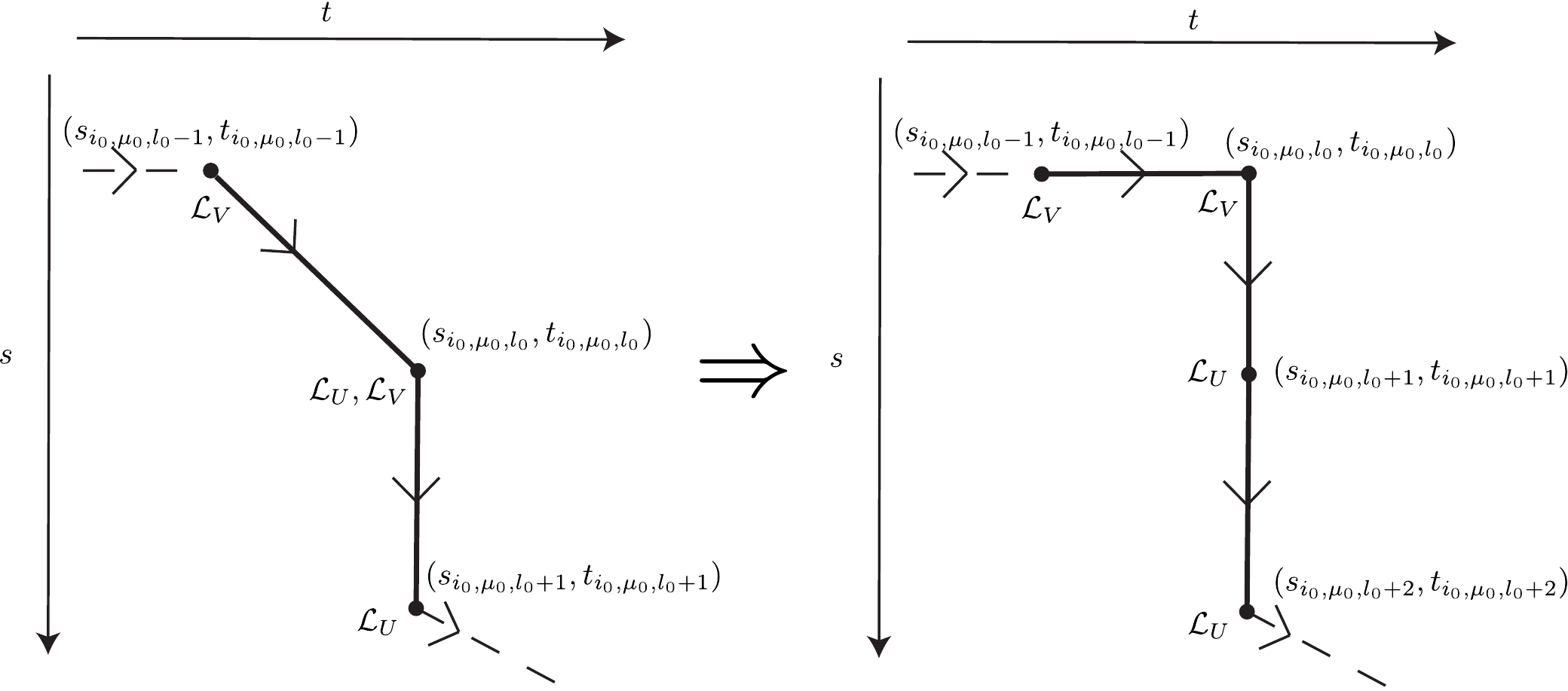}
\caption{\small A fragment of a leg showing an unguarded non-rook move from $(s_{i_0,\mu_0,l_0-1},t_{i_0,\mu_0,l_0-1})$ to $(s_{i_0,\mu_0,l_0},t_{i_0,\mu_0,l_0})$ is converted into two rook moves, thus decreasing $|\L_U \cap \L_V|$ by one.  Note that the labels further down the leg have to be incremented by one.}
\label{fig:fig4}
\end{figure}

One can check that ${\mathcal C}'$ is still a configuration, and $X_{{\mathcal C}'}$ is exactly equal to $X_{\mathcal C}$; informally what has happened here is that a single ``non-rook'' move (which contributed both a $\U_{a,a'}$ factor and a $\V_{b,b'}$ factor to the summand in $X_{\mathcal C}$) has been replaced with an equivalent pair of two rook moves (one of which contributes the $\U_{a,a'}$ factor, and the other contributes the $\V_{b,b'}$ factor).

Observe that, $|\Gamma'| = |\Gamma|+1$ and $|\Omega'| = |\Omega| + 1$ (here we use the non-guarded hypothesis \eqref{ng}), while $|J'|+|K'|=|J|+|K|$ and $|\L'_U \cap \L'_V|=|\L_U \cap \L_V|-1$.  Thus in this case we see that the claim follows from the (second) induction hypothesis.  We may thus eliminate this case and assume that
\begin{equation}\label{som}
 (s_{i_0,\mu_0,l_0-1},t_{i_0,\mu_0,l_0}) \in \Omega \hbox{ whenever } (i_0,\mu_0,l_0) \in \L_U \cap \L_V.
\end{equation}
For similar reasons we may assume
\begin{equation}\label{som2}
 (s_{i_0,\mu_0,l_0},t_{i_0,\mu_0,l_0-1}) \in \Omega \hbox{ whenever } (i_0,\mu_0,l_0) \in \L_U \cap L_V.
\end{equation}

\subsubsection{Second case: a low multiplicity row or column, no unguarded non-rook moves}

Next, given any $\j \in J$, define the \emph{row multiplicity} $\tau_\j$
to be
\begin{multline*}
  \tau_\j := |\{ (i,\mu,l) \in \L_U: s(i,\mu,l) = \j \}| \\
  + |\{ (i,\mu,l) \in \L_U: s(i,\mu,l-1) = \j \}| \\
  + |\{ (i,\mu) \in [j] \times \{0,1\}:  s(i,\mu,k(i,\mu)) = \j \}|
\end{multline*}
and similarly for any $\k \in K$, define the \emph{column multiplicity}
$\tau^\k$ to be
\begin{multline*} 
  \tau^\k := |\{ (i,\mu,l) \in \L_V: t(i,\mu,l) = \k \}| \\
  + |\{ (i,\mu,l) \in \L_V: t(i,\mu,l-1) = \k \}| \\
  + |\{ (i,\mu)  \in [j] \times \{0,1\}: t(i,\mu,k(i,\mu)) = \k \}|.
\end{multline*}

\emph{Remark.}  Informally, $\tau_\j$ measures the number of times $\alpha(\j)$ appears in \eqref{xcdef}, and similarly for $\tau^\k$ and $\beta(\k)$.  Alternatively, one can think of $\tau_\j$ as counting the number of times the spider has the opportunity to ``enter'' and ``exit'' the row $s=\j$, and similarly $\tau^\k$ measures the number of opportunities to enter or exit the column $t=\k$.

By surjectivity we know that $\tau_\j, \tau^\k$ are strictly positive
for each $\j \in J$, $\k \in K$.  We also observe that $\tau_\j, \tau^\k$
must be even.  To see this, write
\[ 
\tau_\j = \sum_{(i,\mu,l) \in \L_U} \bigl(1_{s(i,\mu,l) = \j} +
1_{s(i,\mu,l-1) = \j}\bigr) + \sum_{(i,\mu) \in [j] \times \{0,1\}}
1_{s(i,\mu,k(i,\mu)) = \j}.
\]
Now observe that if $(i,\mu,l) \in \Gamma_+ \backslash \L_U$, then
$1_{s(i,\mu,l) = \j} = 1_{s(i,\mu,l-1) = \j}$.  Thus we have
\[
\tau_\j \hbox{ mod } 2 = \sum_{(i,\mu,l) \in \Gamma_+}
\bigl(1_{s(i,\mu,l) = \j} + 1_{s(i,\mu,l-1) = \j}\bigr) + \sum_{i,\mu
  \in [j] \times \{0,1\}} 1_{s(i,\mu,k(i,\mu)) = \j} \hbox{ mod } 2.
\]
But we can telescope this to
\[ 
\tau_\j \hbox{ mod } 2 = \sum_{i,\mu \in [j] \times \{0,1\}}
1_{s(i,\mu,0)= \j} \hbox{ mod } 2,
\]
and the right-hand side vanishes by \eqref{compat2}, showing that
$\tau_\j$ is even, and similarly $\tau^\k$ is even.

In this subsection, we dispose of the case of a low-multiplicity row,
or more precisely when $\tau_\j = 2$ for some $\j \in J$.  By
symmetry, the argument will also dispose of the case of a
low-multiplicity column, when $\tau^\k = 2$ for some $\k \in K$.

Suppose that $\tau_\j = 2$ for some $\j \in J$.  We first remark
that this implies that there does not exist $(i,\mu,l) \in \L_U$ with
$s(i,\mu,l) = s(i,\mu,l-1)=\j$.  We argue by contradiction and define
$l^\star$ to be the first integer larger than $l$ for which
$(i,\mu,l^\star) \in \L_U$. First, suppose that $l^\star$ does not
exist (which, for instance, happens when $l = k(i,\mu)$). Then in this
case it is not hard to see that $s(i,\mu,k(i,\mu))=\j$ since for
$(i,\mu,l') \notin \L_U$, we have $s(i,\mu,l') = s(i,\mu,l'-1)$. In
this case, $\tau_\j$ exceeds $2$.  Else, $l^\star$ does exist but then
$s(i,\mu,l^\star-1) = \j$ since $s(i,\mu,l') = s(i,\mu,l'-1)$ for $l <
l' < l^\star$. Again, $\tau_\j$ exceeds $2$ and this is a
contradiction.  Thus, if $(i,\mu,l) \in \L_U$ and $s(i,\mu,l) = \j$,
then $s(i,\mu,l-1) \neq \j$, and similarly if $(i,\mu,l) \in \L_U$ and
$s(i,\mu,l-1)=\j$, then $s(i,\mu,l) \neq \j$.

Now let us look at the terms in \eqref{xcdef} which involve
$\alpha(\j)$.  Since $\tau_\j = 2$, there are only two such terms, and
each of the terms are either of the form $\U_{\alpha(\j), \alpha(\j')}$
or $E_{\alpha(\j), \beta(\k)}$ for some $\k \in K$ or $\j' \in J
\backslash \{\j\}$.  We now have to divide into three subcases.

{\bf Subcase 1: \eqref{xcdef} contains two terms $\U_{\alpha(\j),
    \alpha(\j')}$, $\U_{\alpha(\j), \alpha(\j'')}$.} Figure
\ref{fig:fig6}(a) for a typical configuration in which this is the
case.
\begin{figure}
\centering
 \begin{tabular}[t]{cc}
\raisebox{1.31cm}{\includegraphics[width=8cm]{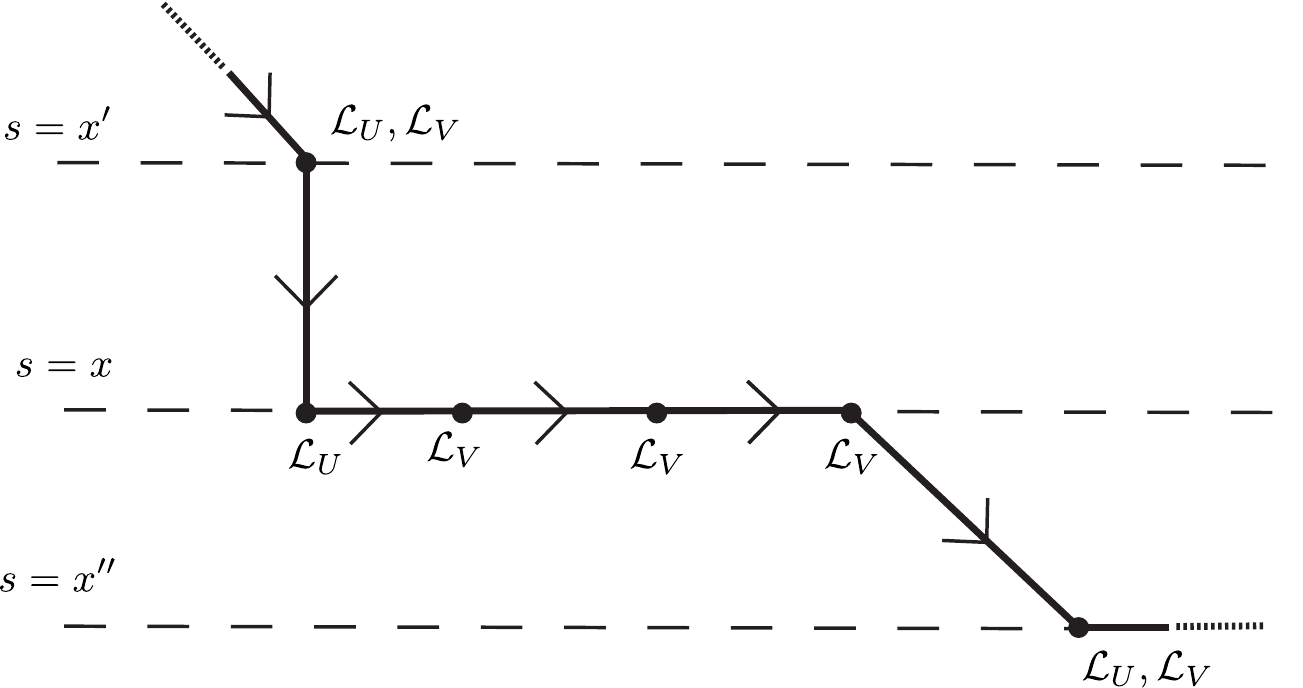}} &
   \includegraphics[width=8cm]{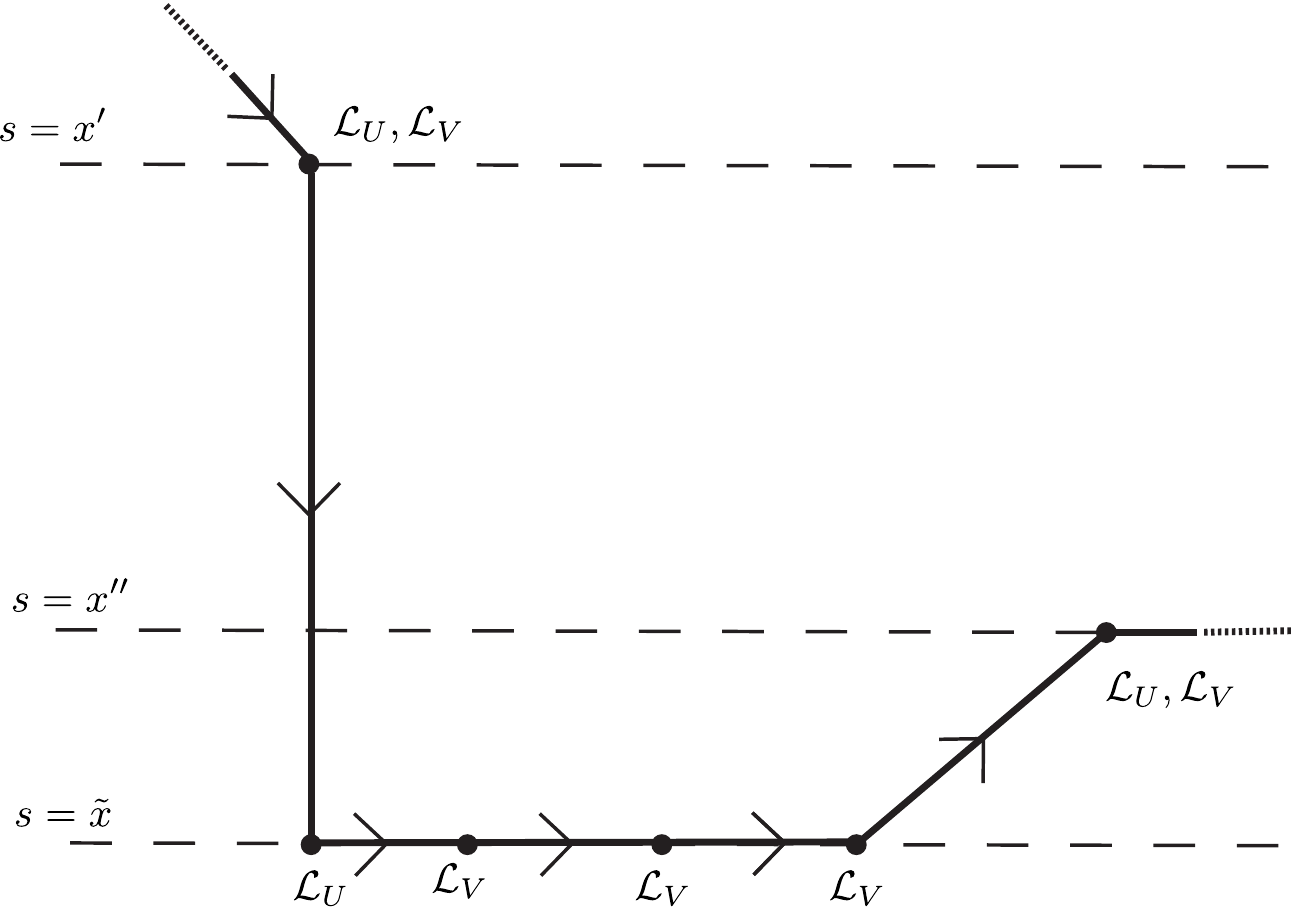}\\ (a) & (b)
     \end{tabular}
     \caption{\small In (a), a multiplicity 2 row is shown.  After
       using the identity \eqref{proj-id}, the contribution of this
       configuration is replaced with a number of terms one of which
       is shown in (b), in which the $\j$ row is deleted and replaced
       with another existing row $\tilde \j$.}
\label{fig:fig6}
\end{figure}

The idea is to use the identity \eqref{proj-id} to ``delete'' the row $\j$, thus reducing $|J|+|K|$ and allowing us to use an induction hypothesis.  Accordingly, let us define $\tilde J := J \backslash \{j\}$, and let $\tilde \alpha: \tilde J \to [n]$ be the restriction of $\alpha$ to $\tilde J$.  We also write $a := \alpha(\j)$ for the deleted row $a$.

We now isolate the two terms $\U_{\alpha(\j), \alpha(\j')}$, $\U_{\alpha(\j), \alpha(\j'')}$ from the rest of \eqref{xcdef}, expressing this sum as
$$
\sum_{\tilde \alpha,\beta} \ldots \Bigl[ \sum_{a \in [n] \backslash
    \tilde \alpha(\tilde J)} \U_{a, \tilde \alpha(\j')} \U_{a,
    \tilde \alpha(\j'')} \Bigr] $$ 
where the $\ldots$ denotes the product of all the terms in \eqref{xcdef} other than $\U_{\alpha(\j), \alpha(\j')}$
and $\U_{\alpha(\j), \alpha(\j'')}$, but with $\alpha$ replaced by
$\tilde \alpha$, and $\tilde \alpha, \beta$ ranging over injections from $\tilde J$ and $K$ to $[n]$ respectively.

From \eqref{proj-id} we have
$$ \sum_{a \in [n]} \U_{a, \tilde \alpha(\j')} \U_{a, \tilde \alpha(\j'')} = 
\left(1 - 2\rho\right) \U_{\tilde \alpha(\j'), \tilde \alpha(\j'')} -
\rho\left(1-\rho\right) 1_{\j'=\j''}$$ and thus
\begin{multline}\label{eq:term} 
\sum_{a \in [n] \backslash \tilde \alpha(\tilde J)} 
\U_{a, \tilde \alpha(\j')} \U_{a, \tilde \alpha(\j'')} 
= \\
\left(1 - 2\rho\right) \U_{\tilde \alpha(\j'), \tilde \alpha(\j'')}  -
\rho\left(1-\rho\right) 
1_{\j'=\j''} - \sum_{\tilde \j \in \tilde J}
\U_{\tilde \alpha(\tilde \j), \tilde \alpha(\j')} \U_{\tilde \alpha(\tilde \j),
  \tilde \alpha(\j'')}.
\end{multline} 
Consider the contribution of one of the final terms $\U_{\tilde
  \alpha(\tilde \j), \tilde \alpha(\j')} \U_{\tilde \alpha(\tilde \j),
  \tilde \alpha(\j'')}$ of \eqref{eq:term}.  This contribution is
equal to $X_{{\mathcal C}'}$, where ${\mathcal C}'$ is formed from
${\mathcal C}$ by replacing $J$ with $\tilde J$, and replacing every
occurrence of $\j$ in the range of $\alpha$ with $\tilde \j$, but
leaving all other components of ${\mathcal C}$ unchanged (see Figure
\ref{fig:fig6}(b)).  Observe that $|\Gamma'| = |\Gamma|$, $|\Omega'|
\leq |\Omega|$, $|J'|+|K'| < |J|+|K|$, so the contribution of these
terms is acceptable by the (first) induction hypothesis (for $C_0$
large enough).

Next, we consider the contribution of the term $\U_{\tilde
  \alpha(\j'), \tilde \alpha(\j'')}$ of \eqref{eq:term}.  This
contribution is equal to $X_{{\mathcal C}''}$, where ${\mathcal C}''$
is formed from ${\mathcal C}$ by replacing $J$ with $\tilde J$,
replacing every occurrence of $\j$ in the range of $\alpha$ with
$\j'$, and also deleting the one element $(i_0,\mu_0,l_0)$ in $\L_U$
from $\Gamma_+$ (relabeling the remaining triples $(i_0,\mu_0,l)$ for
$l_0 < l \leq k(i_0,\mu_0)$ by decrementing $l$ by $1$) that gave rise
to $\U_{\alpha(\j), \alpha(\j')}$, unless this element
$(i_0,\mu_0,l_0)$ also lies in $\L_V$, in which case one removes
$(i_0,\mu_0,l_0)$ from $\L_U$ but leaves it in $\L_V$ (and does not
relabel any further triples) (see Figure \ref{fig:fig8} for an example of the
former case, and \ref{fig:fig9} for the latter case).  One observes
that $|\Gamma''| \geq |\Gamma| - 1$, $|\Omega''| \leq |\Omega|-1$
(here we use \eqref{som}, \eqref{som2}), $|J''|+|K''| < |J| + |K|$,
and so this term also is controlled by the (first) induction
hypothesis (for $C_0$ large enough).
\begin{figure}
 \centering
\includegraphics[scale=.8]{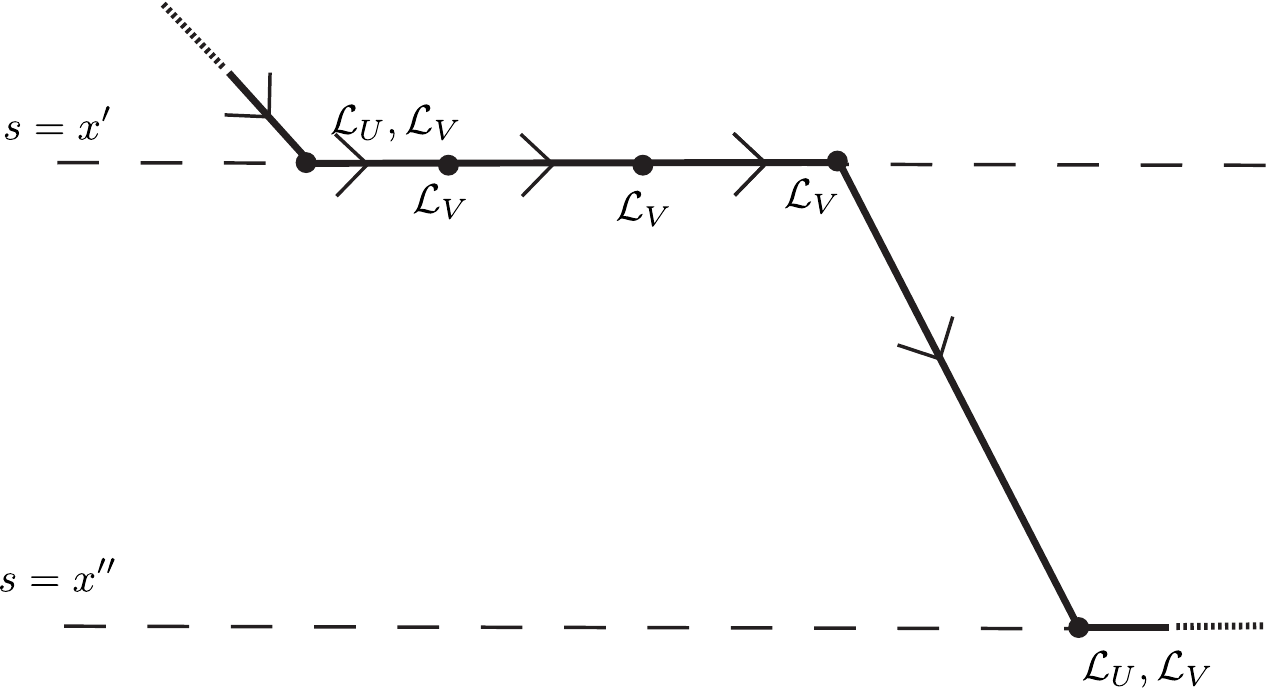}
\caption{\small Another term arising from the configuration in Figure 6(a), in which two $U$ factors have been collapsed into one.  Note the reduction in length of the configuration by one.}
\label{fig:fig8}
\end{figure}

\begin{figure}
 \centering
\includegraphics[scale=.8]{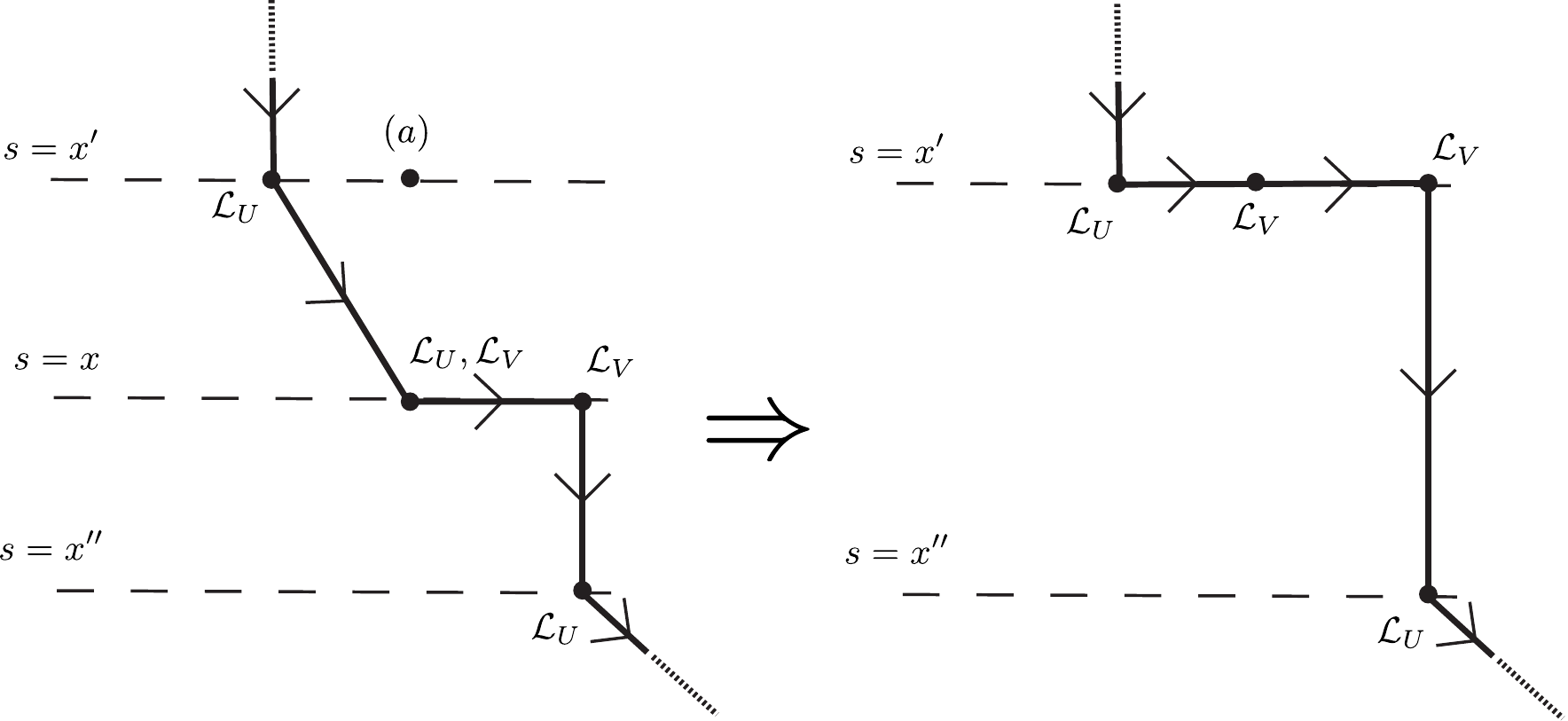}
\caption{\small Another collapse of two $U$ factors into one.  This time, the presence of the $\L_V$ label means that the length of the configuration remains unchanged; but the guarded nature of the collapsed non-rook move (evidenced here by the point (a)) ensures that the support $\Omega$ of the configuration shrinks by at least one instead.}
\label{fig:fig9}
\end{figure}

Finally, we consider the contribution of the term $\rho 1_{\j'=\j''}$
of \eqref{eq:term}, which of course is only non-trivial when
$\j'=\j''$.  This contribution is equal to $\rho X_{{\mathcal C}'''}$,
where ${\mathcal C}'''$ is formed from ${\mathcal C}$ by deleting $\j$
from $J$, replacing every occurrence of $\j$ in the range of $\alpha$
with $\j'=\j''$, and also deleting the two elements $(i_0,\mu_0,l_0)$,
$(i_1, \mu_1, l_1)$ of $\L_U$ from $\Gamma_+$ that gave rise to the
factors $\U_{\alpha(\j), \alpha(\j')}$, $\U_{\alpha(\j),
  \alpha(\j'')}$ in \eqref{xcdef}, unless these elements also lie in
$\L_V$, in which case one deletes them just from $\L_U$ but leaves
them in $\L_V$ and $\Gamma_+$; one also decrements the labels of any
subsequent $(i_0,\mu_0,l)$, $(i_1,\mu_1,l)$ accordingly (see Figure
\ref{fig:fig10}).  One observes that $|\Gamma'''|-|\Omega'''| \geq |\Gamma|-|\Omega| - 1$,
$|J'''|+|K'''| < |J| + |K|$, and
$|J'''| + |K'''| + |\L'''_U \cap \L'''_V| < |J| + |K| + |\L_U \cap
\L_V|$, and so this term also is controlled by the induction
hypothesis.  (Note we need to use the additional $\rho$ factor (which
is less than $\rmu/n$) in order to make up for a possible decrease in
$|\Gamma|-|\Omega|$ by $1$.)
\begin{figure}
 \centering
\includegraphics[scale=.7]{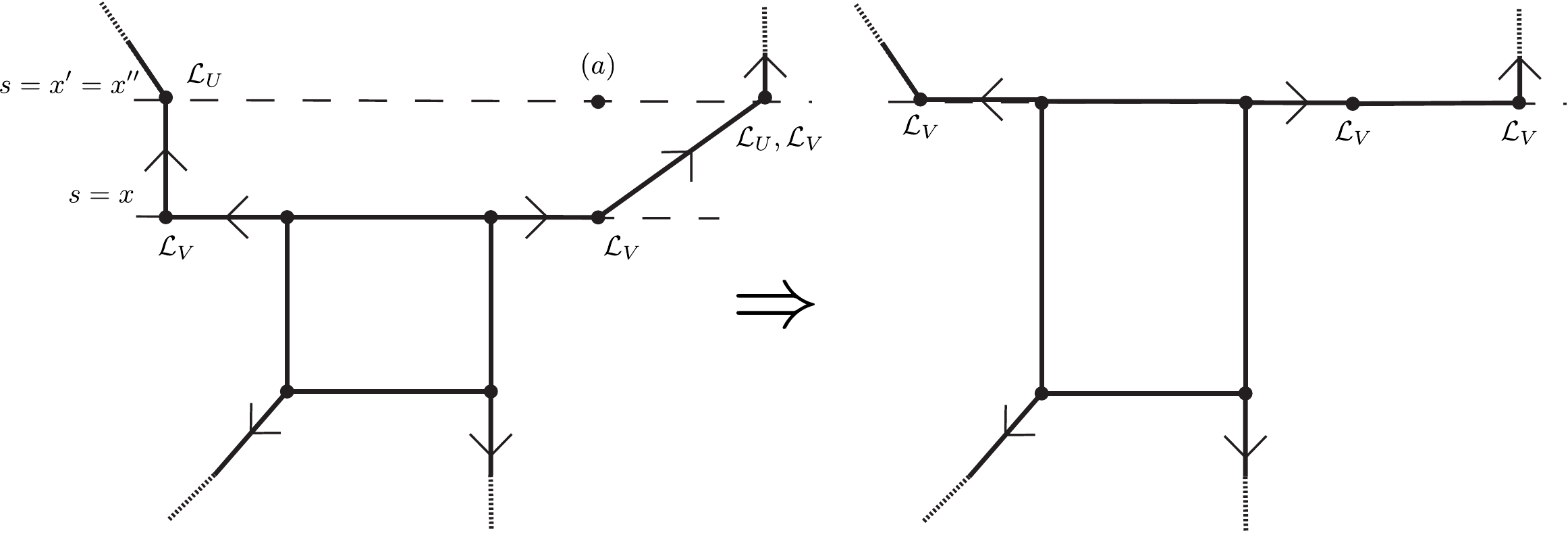}
\caption{\small A collapse of two $U$ factors (with identical indices) to a $\rho 1_{\j'=\j''}$ factor.  The point marked (a) indicates the guarded nature of the non-rook move on the right.  Note that $|\Gamma|-|\Omega|$ can decrease by at most $1$ (and will often stay constant or even increase).}
\label{fig:fig10}
\end{figure}

This deals with the case when there are two $\U$ terms involving
$\alpha(\j)$.  

{\bf Subcase 2: \eqref{xcdef} contains a term $\U_{\alpha(\j), \alpha(\j')}$ and a term $E_{\alpha(\j), \beta(\k)}$.}

A typical case here is depicted in Figure \ref{fig:fig11}.
\begin{figure}
 \centering
\includegraphics[scale=.8]{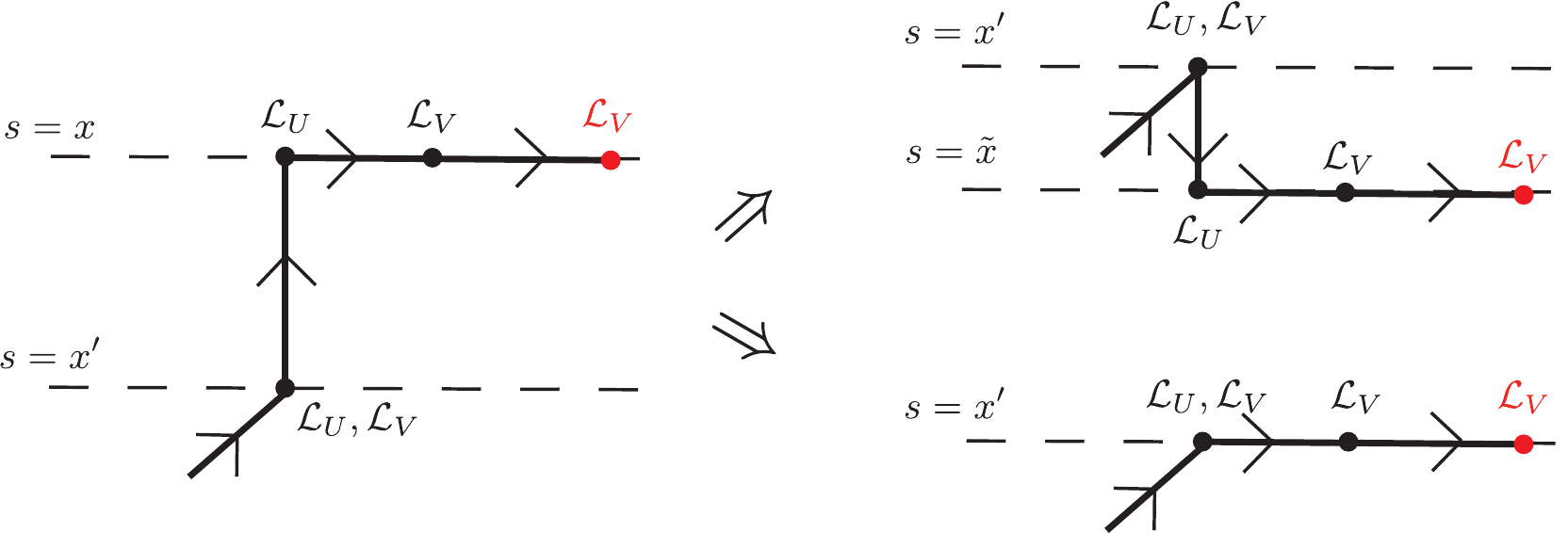}
\caption{\small A configuration involving a $U$ and $E$ factor on the left.  After applying \eqref{proj-i2d}, one gets some terms associated to configuations such as those in the upper right, in which the $\j$ row has been deleted and replaced with another existing row $\tilde \j$, plus a term coming from a configuration in the lower right, in which the $UE$ terms have been collapsed to a single $E$ term.}
\label{fig:fig11}
\end{figure}

The strategy here is similar to Subcase 1, except that one uses
\eqref{proj-i2d} instead of \eqref{proj-id}.  Letting $\tilde J,
\tilde \alpha, a$ be as before, we can express \eqref{xcdef} as
$$
\sum_{\tilde \alpha,\beta} \ldots \Bigl[ \sum_{a \in [n] \backslash
    \tilde \alpha(\tilde J)} \U_{a, \tilde \alpha(\j')} E_{a,
    \beta(\k)} \Bigr]$$ where the $\ldots$ denotes the product of all
the terms in \eqref{xcdef} other than $\U_{\alpha(\j), \alpha(\j')}$
and $E_{\alpha(\j), \beta(\k)}$, but with $\alpha$ replaced by $\tilde
\alpha$, and $\tilde \alpha, \beta$ ranging over injections from
$\tilde J$ and $K$ to $[n]$ respectively.

From \eqref{proj-i2d} we have
$$
  \sum_{a \in [n]} \U_{a,\tilde \alpha(\j')} E_{a,\beta(\k)} = \left(1-\rho\right) E_{\tilde \alpha(\j'),\beta(\k)} $$
and hence  
\begin{equation}\label{eq:term2}
  \sum_{a \in [n] \backslash \tilde \alpha(\tilde J)} \U_{a,\tilde \alpha(\j')} E_{a,\beta(\k)} = 
  \left(1-\rho\right) E_{\tilde \alpha(\j'),\beta(\k)} 
  - \sum_{\tilde \j \in \tilde J} \U_{\tilde \alpha(\tilde j),\tilde \alpha(\j')} E_{\tilde \alpha(\tilde j),\beta(\k)} 
\end{equation}
The contribution of the final terms in \eqref{eq:term2} are treated in
exactly the same way as the final terms in \eqref{eq:term}, and the
main term $E_{\tilde \alpha(\j'),\beta(\k)}$ is treated in exactly the
same way as the term $\U_{\tilde \alpha(\j'), \tilde \alpha(\j'')}$ in
\eqref{eq:term}.  This concludes the treatment of the case when there
is one $\U$ term and one $E$ term involving $\alpha(\j)$.

{\bf Subcase 3: \eqref{xcdef} contains two terms $E_{\alpha(\j),
    \beta(\k)}$, $E_{\alpha(\j), \beta(\k')}$.}

A typical case here is depicted in \ref{fig:fig12}.  The strategy here
is similar to that in the previous two subcases, but now one uses
\eqref{proj-i3d} rather than \eqref{proj-id}.  The combinatorics of
the situation are, however, slightly different.
\begin{figure}
 \centering
\includegraphics[scale=.8]{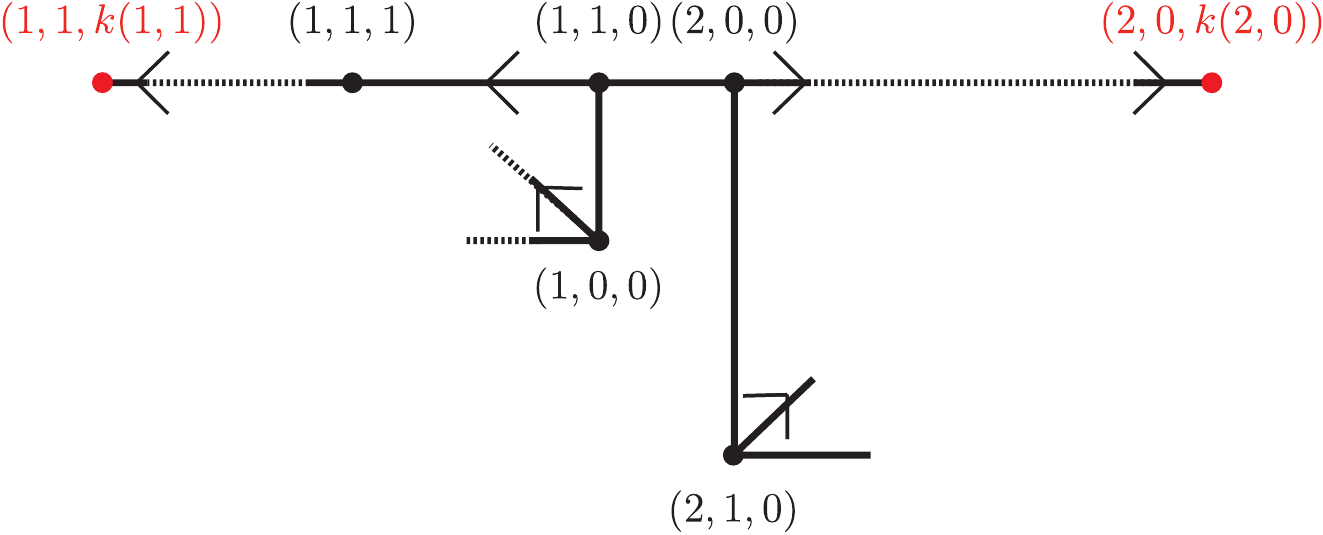}
\caption{\small A multiplicity 2 row with two Es, which are necessarily at the ends of two adjacent legs of the spider.  Here we use $(i,\mu,l)$ as shorthand for $(s_{i,\mu,l},t_{i,\mu,l})$.}
\label{fig:fig12}
\end{figure}

By considering the path from $E_{\alpha(\j), \beta(\k)}$ to
$E_{\alpha(\j), \beta(\k')}$ along the spider, we see (from the
hypothesis $\tau_\j=2$) that this path must be completely horizontal
(with no elements of $\L_U$ present), and the two legs of the spider
that give rise to $E_{\alpha(\j), \beta(\k)}$, $E_{\alpha(\j),
  \beta(\k')}$ at their tips must be adjacent, with their bases
connected by a horizontal line segment.  In other words, up to
interchange of $\k$ and $\k'$, and cyclic permutation of the $[j]$
indices, we may assume that
$$ (\j, \k) = (s(1,1,k(i,1)), t(1,1,k(i,1))); \quad (\j, \k') = (s(2,0,k(2,0)), t(2,0,k(2,0)))$$
with
$$ s(1,1,l) = s(2,0,l') = \j$$
for all $0 \leq l \leq k(1,1)$ and $0 \leq l' \leq k(2,0)$, where the
index $2$ is understood to be identified with $1$ in the degenerate
case $j=1$.  Also, $\L_U$ cannot contain any triple of the form
$(1,1,l)$ for $l \in [k(1,1)]$ or $(2,0,l')$ for $l' \in [k(2,0)]$
(and so all these triples lie in $\L_V$ instead).

For technical reasons we need to deal with the degenerate case $j=1$
separately.  In this case, $s$ is identically equal to $\j$, and so
\eqref{xcdef} simplifies to
$$ \sum_\beta \Bigl[\sum_{a \in [n]}  E_{a,\beta(\k)} E_{a,\beta(\k')}\Bigr]
\prod_{\mu=0}^1 \prod_{l=0}^{k(1,\mu)}
\V_{\beta(t(i,\mu,l-1)),\beta(t(i,\mu,l))}.$$ In the extreme
degenerate case when $k(1,0)=k(1,1)=0$, the sum is just $\sum_{a,b \in
  [n]} E_{ab}^2 = r$, which is acceptable, so we may assume that
$k(1,0)+k(1,1) > 0$.  We may assume that the column multiplicity
$\tau^{\tilde \k} \geq 4$ for every $\tilde \k \in K$, since otherwise
we could use (the reflected form of) one of the previous two subcases
to conclude \eqref{xc} from the induction hypothesis.  (Note when
$\k=\k'$, it is not possible for $\tau^{\k}$ to equal $2$ since
$k(1,0)+k(1,1) > 0$.)

Using \eqref{proj-i3d} followed by \eqref{eq:block} we have
\[
 \Bigl|\sum_{a \in [n]}  E_{a,\beta(\k)} E_{a,\beta(\k')}\Bigr| \lesssim \sqrt{\rmu}/n + 1_{\k=\k'} r/n \lesssim \rmu/n
\]
and so by \eqref{eq:block2} we can bound
$$ |X_{\mathcal C}| \lesssim \sum_\beta (\rmu/n) (\sqrt{\rmu}/n)^{k(1,0)+k(1,1)}.$$
The number of possible $\beta$ is at most $n^{|K|}$, so to establish
\eqref{xc} in this case it suffices to show that
$$ n^{|K|} (\rmu/n) (\sqrt{\rmu}/n)^{k(1,0)+k(1,1)} \lesssim
 (\rmu/n)^{|\Gamma| - |\Omega|} n.$$
Observe that in this degenerate case $j=1$, we have $|\Omega|=|K|$ and $|\Gamma| = k(1,0)+k(1,1)+2$.  One then checks that the claim is true when $\rmu=1$, so it suffices to check that the other extreme case $\rmu=n$, i.e.
$$ |K| - \frac{1}{2} (k(1,0)+k(1,1)) \leq 1.$$
But as $\tau^\k \geq 4$ for all $k$, every element in $K$ must be visited at least twice, and the claim follows.

Now we deal with the non-degenerate case $j>1$.  Letting $\tilde J, \tilde \alpha, a$ be as in previous subcases, we can express \eqref{xcdef} as
\begin{equation}\label{xcdef2}
\sum_{\tilde \alpha,\beta} \ldots \Bigl[ \sum_{a \in [n] \backslash
    \tilde \alpha(\tilde J)} E_{a, \beta(\k)} E_{a,
    \beta(\k')} \Bigr]
\end{equation}
where the $\ldots$ denotes the product of all the terms in \eqref{xcdef} other than $E_{\alpha(\j), \beta(\k)}$
and $E_{\alpha(\j), \beta(\k')}$, but with $\alpha$ replaced by $\tilde \alpha$, and $\tilde \alpha, \beta$ ranging over injections from $\tilde J$ and $K$ to $[n]$ respectively.

From \eqref{proj-i3d}, we have
$$
\sum_{a \in [n]} E_{a, \beta(\k)} E_{a,
    \beta(\k')} = \V_{\beta(\k), \beta(\k')} + \rho 1_{\k=\k'}$$
and hence 
\begin{equation}\label{eq:term3}
\sum_{a \in [n] \backslash \tilde \alpha(\tilde J)} E_{a, \beta(\k)} E_{a,
    \beta(\k')} = \V_{\beta(\k), \beta(\k')} + \rho 1_{\k=\k'} - 
       \sum_{\tilde \j \in \tilde J} E_{\tilde \alpha(\tilde j),\beta(\k)} E_{\tilde \alpha(\tilde j),\beta(\k')}.
\end{equation}
The final terms are treated here in exactly the same way as the final
terms in \eqref{eq:term} or \eqref{eq:term2}.  Now we consider the
main term $\V_{\beta(\k), \beta(\k')}$.  The contribution of this term
will be of the form $X_{{\mathcal C}'}$, where the configuration
${\mathcal C}'$ is formed from ${\mathcal C}$ by ``detaching'' the
two legs $(i,\mu) = (1,1), (2,0)$ from the spider, ``gluing them
together'' at the tips using the $\V_{\beta(\k),\beta(\k')}$ term, and
then ``inserting'' those two legs into the base of the $(i,\mu) =
(1,0)$ leg.  To explain this procedure more formally, observe that the
$\ldots$ term in \eqref{xcdef2} can be expanded further (isolating out
the terms coming from $(i,\mu) = (1,1), (2,0)$) as
\[
\Bigl[\prod_{l=1}^{k(2,0)}
\V_{\beta(t(2,0,l-1)),\beta(t(2,0,l))}\Bigr] \,
\Bigl[\prod_{l=k(1,1)}^1 \V_{\beta(s(1,1,l-1)),\beta(s(1,1,l))}\Bigr]
\ldots
\]
where the $\ldots$ now denote all the terms that do not come from
$(i,\mu)=(1,1)$ or $(i,\mu)=(2,0)$, and we have reversed the order of
the second product for reasons that will be clearer later.  Recalling
that $\k = t(1,1,k(1,1))$ and $\k' = t(2,0,k(2,0))$, we see that the
contribution of the first term of \eqref{eq:term3} to \eqref{xcdef2}
is now of the form
\[
\sum_{\tilde \alpha, \beta} \Bigl[\prod_{l=1}^{k(2,0)}
\V_{\beta(t(2,0,l-1)),\beta(t(2,0,l))}\Bigr]
V_{\beta(t(2,0,k(2,0))),\beta(t(1,1,k(1,1)))} \Bigl[\prod_{l=k(1,1)}^1
\V_{\beta(s(1,1,l-1)),\beta(s(1,1,l))}\Bigr] \ldots.
\]
But this expression is simply $X_{{\mathcal C}'}$, where the
configuration of ${\mathcal C}'$ is formed from ${\mathcal C}$ in the
following fashion:
\begin{itemize}
\item $j'$ is equal to $j-1$, $J'$ is equal to $\tilde J$, and $K'$ is
  equal to $K$.
\item $k'(1,0) := k(2,0) + 1 + k(1,1) + k(1,0)$, and $k'(i,\mu) :=
  k(i+1,\mu)$ for $(i,\mu) \neq (1,0)$.
\item The path $\{(s'(1,0,l),t'(1,0,l)) : l = 0, \ldots, k'(1,0)\}$ is
  formed by concatenating the path $\{(s(1,0,0), t(2,0,l)) : l = 0,
  \ldots, k(2,0)\}$, with an edge from $(s(1,0,0),t(2,0,k(2,0)))$ to
  $(s(1,0,0),t(1,1,k(1,1)))$, 
  with the path $\{(s(1,0,0), t(1,1,l)) : l = k(1,1), \ldots, 0\}$,
  with the path $\{(s(1,0,l),t(1,0,l)) : l =0, \ldots, k(1,0)\}$.
\item For any $(i,\mu) \neq (i,0)$, the path
  $\{(s'(i,\mu,l),t'(i,\mu,l)) : l=0, \ldots k'(i,\mu)\}$ is equal to
  the path $\{(s(i,\mu,l),t(i+1,\mu,l)) : l=0, \ldots, k(i+1,\mu)\}$.
\item We have
\begin{align*}
\L'_U &:= \{ (1,0,k(2,0)+1+k(1,1)+l): (1,0,l) \in \L_U \} \\
&\quad \cup \{ (i,\mu,l): (i+1,\mu,l) \in \L_U \} 
\end{align*}
and
\begin{align*}
\L'_V &:= \{ (1,0,k(2,0)+1+k(1,1)+l): (1,0,l) \in \L_V \} \\
&\quad \cup \{ (i,\mu,l): (i+1,\mu,l) \in \L_V \}  \\
&\quad \cup \{ (1,0,1),\ldots,(1,0,k(2,0)+1+k(1,1)) \}.
\end{align*}
\end{itemize}
This construction is represented in Figure \ref{fig:fig13}.
\begin{figure}
 \centering
\includegraphics[scale=.8]{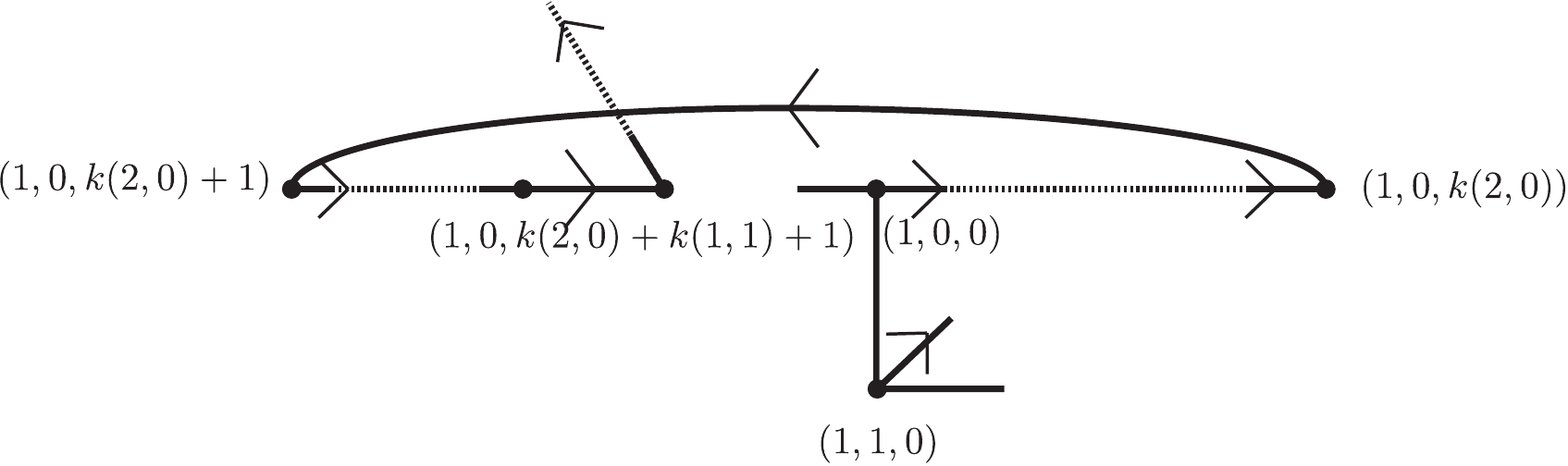}
\caption{\small The configuation from Figure 11 after collapsing the two $E$'s to a $V$, which is represented by a long curved line rather than a straight line for clarity.  Note the substantial relabeling of vertices.}
\label{fig:fig13}
\end{figure}

One can check that this is indeed a configuration.  One has $|J'|+|K'| < |J| + |K|$, $|\Gamma'| = |\Gamma|-1$, and $|\Omega'| \leq |\Omega|-1$, and so this contribution to \eqref{xc} is acceptable from the (first) induction hypothesis.

This handles the contribution of the $\V_{\beta(\k), \beta(\k')}$ term.  The $\rho 1_{\k=\k'}$ term is treated similarly, except that there is no edge between the points $(s(1,0,0),t(2,0,k(2,0)))$ and $(s(1,0,0),t(1,1,k(1,1)))$ (which are now equal, since $\k=\k'$).  This reduces the analogue of $|\Gamma'|$ to $|\Gamma|-2$, but the additional factor of $\rho$ (which is at most $\rmu/n$) compensates for this.  We omit the details.  This concludes the treatment of the third subcase.

\subsubsection{Third case: High multiplicity rows and columns}

After eliminating all of the previous cases, we may now may assume (since $\tau_\j$ is
even) that
\begin{equation}\label{rhojer}
 \tau_\j \geq 4 \hbox{ for all } \j \in J
\end{equation}
and similarly we may assume that
\begin{equation}\label{rhoker}
 \tau^\k \geq 4 \hbox{ for all } \k \in K.
\end{equation}

We have now made the maximum use we can of the cancellation identities \eqref{proj-id}, \eqref{proj-i2d}, \eqref{proj-i3d}, and have no further use for them. Instead, we shall now place absolute values everywhere and estimate $X_{\mathcal C}$ using \eqref{eab}, \eqref{eq:block}, \eqref{eq:block2}, obtaining the bound 
$$ |X_{\mathcal C}| \leq n^{|J| + |K|} O(\sqrt{\rmu}/n)^{|\Gamma| + |\L_U \cap \L_V|}.$$
Comparing this with \eqref{xc}, we see that it will suffice (by taking $C_0$ large enough) to show that
$$
n^{|J| + |K|} (\sqrt{\rmu}/n)^{|\Gamma| + |\L_U \cap \L_V|} \leq (\rmu/n)^{|\Gamma| - |\Omega|} n.
$$
Using the extreme cases $\rmu=1$ and $\rmu=n$ as test cases, we see that our task is to show that
\begin{equation}\label{joke}
|J| + |K| \leq |\L_U \cap \L_V| + |\Omega| + 1 
\end{equation}
and
\begin{equation}\label{joke-2}
|J| + |K| \leq \frac{1}{2} (|\Gamma| + |\L_U \cap \L_V|) + 1.
\end{equation}
The first inequality \eqref{joke} is proven by Lemma \ref{teo:Q}. The
second is a consequence of the double counting identity
$$ 4(|J| + |K|) \le \sum_{\j \in J} \tau_\j + \sum_{\k \in K} \tau^\k = 2 |\Gamma| + 2 |\L_U \cap \L_V|$$
where the inequality follows from \eqref{rhojer}--\eqref{rhoker} (and
we don't even need the $+1$ in this case).

\section{Discussion}
\label{sec:discussion}

Interestingly, there is an emerging literature on the development of
efficient algorithms for solving the nuclear-norm minimization problem
\eqref{eq:cvx2} \cite{Cai08,Ma08}. For instance, in \cite{Cai08}, the
authors show that the singular-value thresholding algorithm can solve
certain problem instances in which the matrix has close to a billion
unknown entries in a matter of minutes on a personal computer. Hence,
the near-optimal sampling results introduced in this paper are
practical and, therefore, should be of consequence to practitioners
interested in recovering low-rank matrices from just a few entries.

To be broadly applicable, however, the matrix completion problem needs
to be robust vis a vis noise. That is, if one is given a few entries
of a low-rank matrix contaminated with a small amount of noise, one
would like to be able to guess the missing entries, perhaps not
exactly, but accurately. We actually believe that the methods and
results developed in this paper are amenable to the study of ``the
noisy matrix completion problem'' and hope to report on our progress
in a later paper.

\section{Appendix}

\subsection{Equivalence between the uniform and Bernoulli models}
\label{sec:Ber}

\subsubsection{Lower bounds}

For the sake of completeness, we explain how Theorem \ref{teo:lower}
implies nearly identical results for the uniform model. We have
established the lower bound by showing that there are two fixed
matrices $M \neq M'$ for which $\PO(M) = \PO(M')$ with probability
greater than $\delta$ unless $m$ obeys the bound \eqref{eq:lower}.
Suppose that $\Omega$ is sampled according to the Bernoulli model with
$p' \ge m/n^2$ and let $F$ be the event $\{\PO(M) = \PO(M')\}$. Then
\begin{align*}
  \P(F) & = \sum_{k = 0}^{n^2} \P(F \, | \, |\Omega| = k) \P(|\Omega| = k)\\
  & \le \sum_{k = 0}^{m-1} \P(|\Omega| = k) + \sum_{k = m}^{n^2}
  \P(F \, | \, |\Omega| = k) \P(|\Omega| = k)\\ 
& \le \P(|\Omega| < m)
  + \P(F \, | \, |\Omega| = m), 
\end{align*}
where we have used the fact that for $k \ge m$, $\P(F \, | \, |\Omega|
= m) \ge \P(F \, | \, |\Omega| = k)$. The conditional distribution of
$\Omega$ given its cardinality is uniform and, therefore,
\[
\P_{\text{Unif}(m)} (F) \ge \P_{\text{Ber}(p')}(F) -
\P_{\text{Ber}(p')}(|\Omega| < m),
\]
in which $\P_{\text{Unif}(m)}$ and $\P_{\text{Ber}(p')}$ are
probabilities calculated under the uniform and Bernoulli models.  If
we choose $p' = 2m/n^2$, we have that $\P_{\text{Ber}(p')}(|\Omega| <
m) \le \delta/2$ provided $\delta$ is not ridiculously small. Thus if
$\P_{\text{Ber}(p')}(F) \ge \delta$, we have
\[
\P_{\text{Unif}(m)} (F) \ge \delta/2.
\]
In short, we get a lower bound for the uniform model by applying the
bound for the Bernoulli model with a value of $p = 2m^2/n$ and a
probability of failure equal to $2\delta$.

\subsubsection{Upper bounds}

We prove the claim stated at the onset of Section \ref{sec:strategy}
which states that the probability of failure under the uniform model
is at most twice that under the Bernoulli model. Let $F$ be the event
that the recovery via \eqref{eq:cvx2} is not exact. With our earlier
notations,
\begin{align*}
  \P_{\text{Ber}(p)}(F) & = \sum_{k = 0}^{n^2} \P_{\text{Ber}(p)}(F \, | \, |\Omega| = k) \P_{\text{Ber}(p)}(|\Omega| = k)\\
  & \ge \sum_{k = 0}^{m}
  \P_{\text{Ber}(p)}(F \, | \, |\Omega| = k) \P_{\text{Ber}(p)}(|\Omega| = k)\\
  & \ge \P_{\text{Ber}(p)}(F \, | \, |\Omega| = m)  \sum_{k = 0}^{m} \P_{\text{Ber}(p)}(|\Omega| = k)\\
  & \ge \frac{1}{2} \P_{\text{Unif}(m)}(F),
\end{align*}
where we have used $\P_{\text{Ber}(p)}(F \, | \, |\Omega| = k) \ge
\P_{\text{Ber}(p)}(F \, | \, |\Omega| = m)$ for $k \le m$ (the
probability of failure is nonincreasing in the size of the observed
set), and $\P_{\text{Ber}(p)}(|\Omega| \le m) \ge 1/2$.

\subsection{Proof of Lemma \ref{teo:equiv}}
\label{sec:QT}

In this section, we will make frequent use of \eqref{eq:QOsq} and of
the similar identity
\begin{equation}
  \label{eq:QTsq}
Q^2_T = (1-2\rho') \QT + \rho'(1-\rho') \OpId, 
\end{equation}
which is obtained by squaring both sides of \eqref{eq:QT} together
with $\PT^2 = \PT$. We begin with two lemmas. 

\begin{lemma}
  \label{teo:QTPT}
For each $k \ge 0$, we have 
\begin{multline}
  \label{eq:QTPT}
  (\QO \PT)^k \QO =  \sum_{j = 0}^{k} \alpha^{(k)}_j  (\QO \QT)^j \QO +  \sum_{j = 0}^{k-1}  \beta^{(k)}_j  (\QO \QT)^j\\
  + \sum_{j = 0}^{k-2} \gamma^{(k)}_j \QT (\QO \QT)^j \QO +
  \sum_{j = 0}^{k-3} \delta^{(k)}_j \QT (\QO \QT)^j,
\end{multline}
where starting from $\alpha^{(0)}_0 = 1$, the sequences
$\{\alpha^{(k)}\}$, $\{\beta^{(k)}\}$, $\{\gamma^{(k)}\}$ and
$\{\delta^{(k)}\}$ are inductively defined via
\begin{align*}
  \alpha_{j}^{(k+1)} & = [\alpha_{j-1}^{(k)} + (1-\rho') \gamma_{j-1}^{(k)}] + \frac{\rho'(1-2p)}{p} [\alpha_{j}^{(k)} + (1-\rho')\gamma_{j}^{(k)}] + 1_{j = 0} \rho'[\beta_{0}^{(k)} + (1-\rho')\delta_{0}^{(k)}]\\
 \beta_{j}^{(k+1)} & = [\beta_{j-1}^{(k)} + (1-\rho') \delta_{j-1}^{(k)}]
 + \frac{\rho'(1-2p)}{p} [\beta_{j}^{(k)} +
 (1-\rho')\delta_{j}^{(k)}] 1_{j > 0} + 1_{j = 0} \rho' \frac{1-p}{p} [\alpha_{0}^{(k)} +
 (1-\rho')\gamma_{0}^{(k)}]
\end{align*}
and
\begin{align*}
  \gamma_{j}^{(k+1)} & = \frac{\rho'(1-p)}{p} [\alpha_{j+1}^{(k)} + (1-\rho')\gamma_{j+1}^{(k)}]\\
  \delta_{j}^{(k+1)} & = \frac{\rho'(1-p)}{p} [\beta_{j+1}^{(k)} +
  (1-\rho')\delta_{j+1}^{(k)}].
\end{align*}
In the above recurrence relations, we adopt the convention that
$\alpha_{j}^{(k)} = 0$ whenever $j$ is not in the range specified by
\eqref{eq:QTPT}, and similarly for $\beta_{j}^{(k)}$,
$\gamma_{j}^{(k)}$ and $\delta_{j}^{(k)}$.
\end{lemma}
\begin{proof}
  The proof operates by induction. The claim for $k = 0$ is
  straightforward. To compute the coefficient sequences of $(\QO
  \PT)^{k+1} \QO$ from those of $(\QO \PT)^{k} \QO$,
  use the identity $\PT = \QT + \rho' \OpId$ to decompose $(\QO
  \PT)^{k+1} \QO$ as follows:
\[
(\QO \PT)^{k+1}\QO = \QO \QT (\QO \PT)^k \QO + \rho' \QO (\QO \PT)^k \QO.  
\]
Then expanding $(\QO \PT)^k \QO$ as in \eqref{eq:QTPT}, and
using the two identities  
\[
\QO (\QO \QT)^j \QO = \begin{cases}  \frac{1-2p}{p} \QO + \frac{(1-p)}{p} \OpId, & j = 0,\\
  \frac{1-2p}{p} (\QO \QT)^j \QO + \frac{(1-p)}{p} \QT
  (\QO \QT)^{j-1} \QO, & j > 0,
\end{cases}
\]
and 
\[
\QO (\QO \QT)^j = \begin{cases} \QO, & j = 0,\\
  \frac{1-2p}{p} (\QO \QT)^j + \frac{(1-p)}{p} \QT
  (\QO \QT)^{j-1}, & j > 0,
\end{cases}
\]
which both follow from \eqref{eq:QOsq}, gives the desired recurrence
relation. The calculation is rather straightforward and omitted.
\end{proof}

We note that the recurrence relations give $\alpha^{(k)}_k = 1$ for
all $k \ge 0$,  
\[
\beta^{(k)}_{k-1} = \beta^{(k-1)}_{k-2} = \ldots = \beta_0^{(1)} =
\frac{\rho'(1-p)}{p}
\]
for all $k \ge 1$, and
\begin{align*}
  \gamma^{(k)}_{k-2} & = \frac{\rho'(1-p)}{p} \alpha^{(k-1)}_{k-1} = 
  \frac{\rho'(1-p)}{p}, \\
  \delta^{(k)}_{k-3} & = \frac{\rho'(1-p)}{p} \beta^{(k-1)}_{k-2} = 
 \Bigl(\frac{\rho'(1-p)}{p}\Bigr)^2, \\
\end{align*}
for all $k \ge 2$ and $k \ge 3$ respectively. 

\begin{lemma}
  \label{teo:alpha-bound}
  Put $\lambda = \rho'/p$ and observe that by assumption \eqref{rnp},
  $\lambda < 1$. Then for all $j, k \ge 0$, we have
\begin{equation}
\label{eq:alpha-bound}
\max\bigl(|\alpha_{j}^{(k)}|, |\beta_{j}^{(k)}| , |\gamma_{j}^{(k)}|,  |\delta_{j}^{(k)}|\bigr) \le \lambda^{\lceil\frac{k-j}{2}\rceil} 4^{k}. 
\end{equation}
\end{lemma}
\begin{proof}
  We prove the lemma by induction on $k$. The claim is true for $k =
  0$. Suppose it is true up to $k$, we then use the recurrence
  relations given by Lemma \ref{teo:QTPT} to establish the claim up to
  $k+1$. In details, since $|1 -\rho'| < 1$, $\rho' < \lambda$ and
  $|1-2p| < 1$, the recurrence relation for $\alpha^{(k+1)}$ gives
\begin{align*}
  |\alpha_{j}^{(k+1)}| & \le |\alpha_{j-1}^{(k)}| + |\gamma_{j-1}^{(k)}| + \lambda [|\alpha_{j}^{(k)}| + |\gamma_{j}^{(k)}|] + 1_{j = 0} \lambda [|\beta_{0}^{(k)}| 
+ |\delta_{0}^{(k)}|]\\
  & \le 2 \, \lambda^{\lceil\frac{k+1-j}{2}\rceil} 4^{k} 1_{j > 0}
  + 2 \lambda^{\lceil\frac{k-j}{2}\rceil+1} 4^{k} + 2
 \lambda^{\lceil\frac{k}{2}\rceil+1} 4^{k} 1_{j = 0}\\
  & \le 2 \, \lambda^{\lceil\frac{k+1-j}{2}\rceil} 4^{k} 1_{j > 0}
  + 2 \, \lambda^{\lceil\frac{k+1-j}{2}\rceil} 4^{k} +
  2 \, \lambda^{\lceil\frac{k+1}{2}\rceil} 4^{k} 1_{j = 0}\\
  & \le \lambda^{\lceil\frac{k+1-j}{2}\rceil} 4^{k+1},
\end{align*}
which proves the claim for the sequence $\{\alpha^{(k)}\}$. We bound
$|\beta_{j}^{(k+1)}|$ in exactly the same way and omit the
details.  Now the recurrence relation for $\gamma^{(k+1)}$ gives
\begin{align*}
  |\gamma_{j}^{(k+1)}| & \le \lambda [|\alpha_{j+1}^{(k)}| +
  |\gamma_{j+1}^{(k)}|]\\
& \le 2  \lambda^{\lceil\frac{k-j-1}{2}\rceil+1} 4^k \\
& \le 4^{k+1}   \lambda^{\lceil\frac{k+1-j}{2}\rceil},
\end{align*}
which proves the claim for the sequence $\{\gamma^{(k)}\}$. The quantity
$|\delta_{j}^{(k+1)}|$ is bounded in exactly the same way, which
concludes the proof of the lemma.
\end{proof}

We are now well positioned to prove Lemma \ref{teo:equiv} and begin by
recording a useful fact. Since for any $X$, $\|\PTp(X)\| \le \|X\|$,
and
\[
\QT = \PT - \rho' \OpId = (I - \PTp) - \rho' \OpId = (1-\rho') \OpId - \PTp,
\]
the triangular inequality gives that for all $X$,
\begin{equation}
  \label{eq:naivebound}
  \|\QT(X)\| \le 2 \|X\|.
\end{equation}

Now
\begin{multline*}
  \|(\QO \PT)^k \QO(E)\| \le  
\sum_{j = 0}^{k} |\alpha^{(k)}_j|  \|(\QO \QT)^j \QO(E)\| +  
\sum_{j = 0}^{k-1}  |\beta^{(k)}_j|  \|(\QO \QT)^j(E)\|\\
  + \sum_{j = 0}^{k-2} |\gamma^{(k)}_j| \|\QT (\QO \QT)^j
  \QO (E)\| + \sum_{j = 0}^{k-3} |\delta^{(k)}_j| \|\QT (\QO
  \QT)^j(E)\|,
\end{multline*}
and it follows from \eqref{eq:naivebound} that 
\[
\|(\QO \PT)^k \QO(E)\| \le \sum_{j = 0}^{k}
(|\alpha^{(k)}_j| + 2 |\gamma^{(k)}_j|) \|(\QO \QT)^j
\QO(E)\| + \sum_{j = 0}^{k-1} (|\beta^{(k)}_j| + 2
|\delta^{(k)}_j|) \|(\QO \QT)^j(E)\|.
\]
For $j = 0$, we have $\|(\QO \QT)^j(E)\| = \|E\| = 1$ while for $j > 0$
\[
\|(\QO \QT)^j(E)\| = \|(\QO \QT)^{j-1} \QO \QT (E)\| =
(1-\rho') \|(\QO \QT)^{j-1} \QO (E)\|
\] 
since $\QT(E) = (1-\rho')(E)$. By using the size estimates given by
Lemma \ref{teo:alpha-bound} on the coefficients, we have 
\begin{align*}
  \frac{1}{3} \|(\QO \PT)^k \QO(E)\| & \le \frac{1}{3}
  \sigma^{\frac{k+1}{2}} + 4^k \sum_{j = 0}^{k-1}
  \lambda^{\lceil\frac{k-j}{2}\rceil} \sigma^{\frac{j+1}{2}} + 4^{k}
  \sum_{j = 0}^{k-1}
  \lambda^{\lceil\frac{k-j}{2}\rceil} \sigma^{\frac{j}{2}}\\
  & \le \frac{1}{3} \sigma^{\frac{k+1}{2}} + 4^k \sigma^{\frac{k+1}{2}} \sum_{j
    = 0}^{k-1} \lambda^{\lceil\frac{k-j}{2}\rceil} \sigma^{-\frac{k-j}{2}}
  + 4^{k} \sigma^{\frac{k}{2}} \sum_{j = 0}^{k-1}
  \lambda^{\lceil\frac{k-j}{2}\rceil} \sigma^{-\frac{k-j}{2}}\\
  & \le \frac{1}{3} \sigma^{\frac{k+1}{2}} + 4^k \left(\sigma^{\frac{k+1}{2}} +
    \sigma^{\frac{k}{2}}\right) \sum_{j = 0}^{k-1}
  \lambda^{\lceil\frac{k-j}{2}\rceil} \sigma^{-\frac{k-j}{2}}.
\end{align*}
Now, 
\[
\sum_{j = 0}^{k-1} \lambda^{\lceil\frac{k-j}{2}\rceil}
\sigma^{-\frac{k-j}{2}} \le \left(\frac{\lambda}{\sqrt{\sigma}} +
  \frac{\lambda}{\sigma}\right)\frac{1}{1-\frac{\lambda}{\sigma}} \le \frac{2}{3} \sqrt{\sigma}
\]
where the last inequality holds provided that $4\lambda \le
\sigma^{3/2}$. The conclusion is
\[
\|(\QO \PT)^k \QO(E)\| \le (1+4^{k+1}) \sigma^{\frac{k+1}{2}},
\]
which is what we needed to establish.

\subsection*{Acknowledgements}
E.~C. is supported by ONR grants N00014-09-1-0469 and N00014-08-1-0749
and by the Waterman Award from NSF. E.~C. would like to thank Xiaodong
Li and Chiara Sabatti for helpful conversations related to this
project.  T.~T. is supported by a grant from the MacArthur Foundation,
by NSF grant DMS-0649473, and by the NSF Waterman award.

\small
\bibliographystyle{plain}
\bibliography{OptimalCompletion}

\end{document}